 \documentclass[prx, superscriptaddress,longbibliography,margin=1mm 11pt]{revtex4-2} 

\usepackage{graphicx}
\usepackage{bm}
\usepackage{amsmath}
\usepackage{environ}
\usepackage{appendix}
\usepackage{physics}

\usepackage[nopar]{lipsum}

\makeatletter
\NewEnviron{widerequation}{%
  \begin{equation*}
  \sbox\z@{\let\label\@gobble$\displaystyle\BODY$}
  \makebox[\textwidth]{%
    \begin{minipage}{\dimexpr\wd\z@+3em}
    \vspace{-\baselineskip}
    \begin{equation}
    \BODY
    \end{equation}
    \end{minipage}%
  }
  \end{equation*}
}
\usepackage{amsthm}
\usepackage{amssymb}
\usepackage{amsfonts}
\usepackage{fancyhdr}
\usepackage{slashed}
\usepackage{bbm}
\usepackage{latexsym,epsfig,bbm}
\usepackage{todonotes}
\def \d {\mathrm{d}}
\newcommand{\littleo}{{\scriptstyle \mathcal{O}}}
\usepackage{color}
\definecolor{blue}{rgb}{0,0.2,1}

\definecolor{red}{rgb}{0.9,0,0}

\newtheorem{theorem}{Theorem}
\newtheorem{lemma}[theorem]{Lemma}
\newtheorem{definition}[theorem]{Definition}
\newtheorem{problem}{Problem}
\newtheorem{corollary}[theorem]{Corollary}

\newcommand{\vect}[1]{\boldsymbol{#1}}

\begin{document}

\title{Quantum algorithms for computing observables of nonlinear partial differential equations}

\date{\today}

\author{Shi Jin}
\affiliation{Institute of Natural Sciences, Shanghai Jiao Tong University, Shanghai 200240, China}
\affiliation{School of Mathematical Sciences, Shanghai Jiao Tong University, Shanghai 200240, China}
\affiliation{Ministry of Education, Key Laboratory in Scientific and Engineering Computing, Shanghai Jiao Tong University, Shanghai 200240, China}

\author{Nana Liu}
\email{Corresponding author:  nana.liu@quantumlah.org}
\affiliation{Institute of Natural Sciences, Shanghai Jiao Tong University, Shanghai 200240, China}
\affiliation{Ministry of Education, Key Laboratory in Scientific and Engineering Computing, Shanghai Jiao Tong University, Shanghai 200240, China}
\affiliation{University of Michigan-Shanghai Jiao Tong University Joint Institute, Shanghai 200240, China}

\begin{abstract}
We construct quantum algorithms to compute physical observables of nonlinear PDEs with $M$ initial data. Based on an exact mapping between nonlinear and linear PDEs using the level set method, these new quantum algorithms for nonlinear Hamilton-Jacobi and scalar hyperbolic PDEs can be performed with a computational cost that is independent of  $M$, for arbitrary nonlinearity. Depending on the details of the initial data, it can also display up to exponential advantage in both the dimension of the PDE and the error in computing its observables. For general nonlinear PDEs, quantum advantage with respect to $M$ is possible in the large $M$ limit. 
\end{abstract}
\maketitle 


\section{Introduction}

Nonlinear ordinary and partial differential equations (ODEs and PDEs) have been central to modelling of some of the most significant problems in physics, chemistry, engineering, biology and finance, including climate modelling, aircraft design, molecular dynamics and drug design, deep learning neural networks and financial markets. In physics, the most important mathematical equations--from quantum mechanics, classical mechanics to kinetic theory and hydrodynamics--are all modelled by linear or nonlinear (integro)-differential equations.  Although quantum algorithms can be potentially advantageous for certain linear problems like linear PDEs (e.g. \cite{clader2013preconditioned, childs2021high, costa2019quantum, linden2020quantum, engel2019quantum, cao2013quantum}) it is still unclear to what extent quantum algorithms can be leveraged for nonlinear problems. Although quantum mechanics itself is fundamentally linear (as far as we know), most natural phenomena--and their associated mathematical equations or models--are {\it nonlinear}, hence the ability to simulate nonlinear problems--including the nonlinear PDEs--will
significantly extend the horizon of quantum computing. \\

The most natural way to approach a nonlinear problem using quantum algorithms is to find a way to represent the nonlinear problem in a linear way, where quantum computational advantage in the former problem can still be maintained.  Here we distinguish between two types of approaches that converts a nonlinear PDE into a linear PDE. One approach involves approximations (e.g. either through linearisation of the nonlinearity or through discretisation). They include the \textit{linear approximation} and the \textit{linear representation of nonlinear ODEs} methods. The second approach is the \textit{linear representation for nonlinear PDEs}, where no approximations are required to map between the nonlinear and linear PDEs. \\

In the linear approximation approach, in which the nonlinear term is {\it linearized}, errors are introduced, so the approach may only be valid for a short time, for weak nonlinearities and consequently may lose significant nonlinear features of the problem after a long time.   On the other hand, approaches like Carlemann linearisation \cite{liu2021efficient} or in \cite{lloyd2020quantum, leyton2008quantum}, require that the size of the corresponding linear problem to increase in a way that is dependent on the degree of nonlinearity and is also restricted to polynomial nonlinearities. This requirement of expanding into ever higher dimensions to deal with strong nonlinearities comes at a sacrifice of the (sometimes significant) resource cost. For general nonlinear functions that do not have low-order polynomial expansions, this can be infeasible and important  nonlinear features can be lost in the long time limit due to the truncation in the Carlemann linearization. This is similarly true for methods in \cite{lloyd2020quantum, leyton2008quantum}.\\

The Koopman-von Neumann approach \cite{joseph2020koopman} (a similar approach taken in \cite{dodin2021applications}), is a type of  {\it linear representation method for nonlinear ODEs}, though not for nonlinear PDEs. This means it allows one to convert nonlinear ODEs into a linear transport equation, without any loss of information, thus is valid globally in time, for any nonlinearity. However, to apply this method to solve nonlinear PDEs, one needs to first discretise its spatial derivatives, giving rise to   a system of nonlinear ODEs which can then be implemented via the Koopman-von Neumann approach. This, unfortunately, would give rise to ODE systems with exceedingly large dimensions, where the dimension depends on the discretisation error, hence the quantum advantages can be lost. Since one needs to discretize the nonlinear PDE first, in this regard this method is not a linear representation method for PDEs, as we defined above. \\


In this paper, for several important  classes of nonlinear PDEs, including  the Hamilton-Jacobi equations and scalar nonlinear hyperbolic  equations,  we  map the underlying nonlinear PDEs of $(d+1)$-dimension to a (not more than) $(2d+1)$-dimensional linear PDEs, by utilizing the level set formalism \cite{JL03}. One can then solve the {\it linear} PDEs -- whose dimension is at most twice that of the original nonlinear PDE -- on a quantum computer, without losing any physical information.  This procedure is {\it exact}, so no approximations are made. This is an example of \textit{linear representation method for nonlinear PDEs}.\\

We also find another pleasing bonus of this approach, thanks to the linearity being valid globally in time, in that one can run problems with many ($M>>1$) different initial data, and obtain ensemble averages of the physical observables at a later time, by solving the PDE {\it just once!}\\

\begin{figure}[t] \label{fig:Nonlinearchart}
\includegraphics[width=15cm]{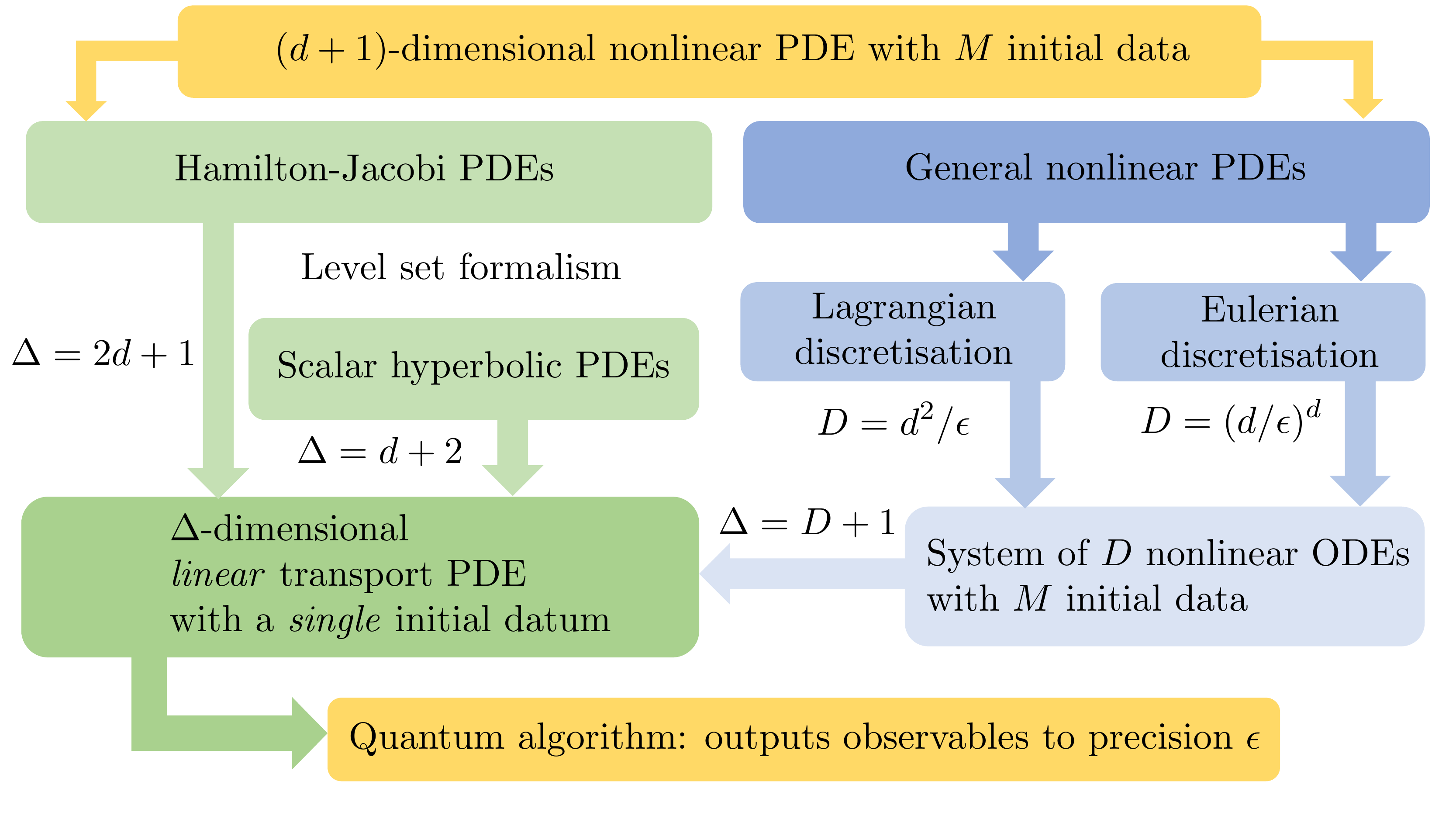}
\centering
\caption{Computing observables from $(d+1)$-dimensional nonlinear PDE with $M$ initial data (a) \textit{Linear representation for nonlinear PDE approach (Green)}: Using the level set formalism, $(d+1)$-dimensional Hamilton-Jacobi PDEs (Section~\ref{sec:HJ}) with $M$ initial data are mapped onto a $(2d+1)$-dimensional linear transport PDE with a single initial datum. Similarly, $(d+1)$-dimensional scalar hyperbolic PDEs (Section~\ref{sec:hyperbolic}) with $M$ initial data are mapped onto a $(d+2)$-dimensional linear transport PDE with a single initial datum; (b) \textit{Linear representation for nonlinear ODE approach (Blue)}: $(d+1)$-dimensional general nonlinear PDEs can be mapped onto $D$ nonlinear ODEs with $M$ initial data using either the Lagrangian or Eulerian discretisation methods (Section~\ref{sec:generalnonlinear}). The system of $D$ nonlinear ODEs with $M$ initial data can be mapped onto a $(D+1)$-dimensional linear transport PDE with a single initial datum and ensemble averages can be computed with a quantum algorithm (Section~\ref{sec:ODEs}).}
\end{figure}

Using this linear representation for nonlinear PDEs, we develop quantum algorithms to compute physical observables of general $(d+1)$-dimensional nonlinear Hamilton-Jacobi and scalar hyperbolic PDEs with multiple ($M$) initial data. In addition to being of practical interest, computing observables also allows one to compare quantum computational costs with classical computational costs on an equal footing, since they are solving the \textit{same} problem. This is in contrast to many previous quantum algorithms for ODEs and PDEs (e.g. \cite{berry2014high, berry2017quantum, childs2021high}) that focus only on the cost of quantum state preparation of the PDE solutions, {\it without} considering (potentially large) costs in extracting those solutions and interpreting them. Furthermore, this level set linear representation allows us to introduce a new embedding of classical data into a quantum state which we call the level set embedding. This embedding is shown to be advantageous in computing physical observables compared to the amplitude embedding commonly used. \\

When $M=1$, we show that, for certain classes of initial data, there can be up to exponential quantum advantage in computing observables to precision $\epsilon$, with respect to both $d$ and $\epsilon$, with no large overheads in any other parameter, including time $T$. The quantum cost is also \textit{independent} of the nature of the nonlinearity. For $M>1$, our algorithm computes ensemble averages over $M$ different initial data {\it in one computation}, instead of computing the problem $M$-times. Hence the quantum resource cost is \textit{independent} of $M$, while classical costs are linear in $M$. The same advantages in $d$ and $\epsilon$ still apply. \\

There are important applications where computing a PDE or ODE with multiple initial data ($M>1$) is of interest. For example, one needs to run numerical simulations with many different initial data to obtain ensemble averaged solutions in uncertainty quantification with random initial data, using Monte-Carlo sampling techniques or stochastic collocation \cite{Mishira-UQ}. In Rayleigh-Taylor instability one needs to run many experiments with different initial data to obtain an ensemble averaged numerical solution that can converge (to the Young-measure solution \cite{MishTad}). In geometric optics one needs to solve multiple rays that follows a Hamiltonian system  (eikonal equation) with different initial data \cite{Can-Ying}.  In quantum wave packet methods for quantum dynamics simulations,  one needs to solve for  multiple Gaussian wave packets, each evolving by a system of nonlinear ODEs \cite{Heller}, and later seek their linear superposition. We show that our quantum algorithms {\it always} have
quantum advantage in $M$, under various conditions. \\

For more general nonlinear PDEs, the idea of first discretising the spatial derivatives to convert them into system of nonlinear ODEs is also discussed. Two approximation techniques,
the Lagrangian methods, such as particle, vortex or mesh-free methods that do not use grids, and Eulerian methods which use grid-based discretisations of spatial derivatives, are studied. We show that using this approach, there can only be quantum advantage in the large $M$ limit, while there are no quantum advantages in $d$ and $\epsilon$. Thus computing ensemble averages for general nonlinear PDEs using quantum algorithms still requires further work to demonstrate useful quantum advantages. \\

The outline of the paper is as follows. We begin in section \ref{section-2}  with the background to nonlinear Hamilton-Jacobi and scalar hyperbolic PDEs and a system of ODEs with multiple initial conditions. We finish with a basic summary of quantum protocols for solving linear systems of equations. Readers familiar with these areas can skip this section. In Section~\ref{sec:HJ} we present the linear representation method using level sets and the quantum algorithm used to compute the physical observables. We compare the classical and quantum computational costs and derive the conditions under which there are quantum advantages.  In Section~\ref{sec:hyperbolic} we repeat the analysis for scalar nonlinear hyperbolic equations. The quantum query and gate complexities for computing ensemble averages of a system of nonlinear ODEs is presented in Section~\ref{sec:ODEs}, and the application of this method to more general nonlinear PDEs is in Section~\ref{sec:generalnonlinear}. \\

In Figure~\ref{fig:Nonlinearchart} we give a basic  outline of the methods considered in this paper. Table~\ref{table:nonlin} gives a summary of the comparison between classical and quantum resource costs in computing ensemble averages for nonlinear PDEs and ODEs studied in this paper. Throughout the paper, we always use first order approximations to space and time derivatives for all systems under consideration. Extensions to higher order approximations can be done in a straightforward way and won't be pursued here.

    \begin{table}[ht] 
\caption{Quantum ($\mathcal{Q}$) and classical ($\mathcal{C}$) cost comparison in computing observables at time $T$ to precision $\epsilon$} 
\centering 
\begin{align}
    & \text{PDE: }\mathcal{O}\left(\frac{\mathcal{C}}{\mathcal{Q}}\right)=\tilde{O}\left(\frac{M}{T^2}d^{r_1}\left(\frac{1}{\epsilon}\right)^{r_{2}}\right) \nonumber \\
     & \text{ODE: }\mathcal{O}\left(\frac{\mathcal{C}}{\mathcal{Q}}\right)=\tilde{O}\left(\frac{M}{T^2}D^{r_1}\left(\frac{1}{\epsilon}\right)^{r_{2}}\right) \nonumber 
    \end{align} 
    
\begin{tabular}{c c c c c c c} 
\hline\hline 
Nonlinear equations  & $r_1$ & & $r_2$ & $b$ range  & & Quantum \\
($M$ initial data) & & & & (initial data-& & advantage  \\ 
 & & &  & dependent)& & (possible)  \\[2ex]
\hline 
\\
$(d+1)$-dimensional & $d-4-b$ & & $d-9-3b$ & $b \in[0, \frac{d}{3}-3)$ &  & $M$, $d$, $\epsilon$\\ 
Hamilton-Jacobi PDE & & & & & & \\[3ex]
$(d+1)$-dimensional   & $d-5-b$ & & $d-9-3b$ & $b \in[0, \frac{d}{3}-3)$ &  & $M$, $d$, $\epsilon$ \\
hyperbolic PDE & & & & & & \\[3ex]
System of $D$ ODEs & $-5$ & & $-9$ & $b=0$ &  & $M$\\ [3ex]
$(d+1)$-dimensional general PDE  & $-7$ & & $-13$ & $b=0$ &  & $M$ \\ 
 (Lagrangian discretisation) & & & & & & \\[3ex]
$(d+1)$-dimensional general PDE & $-4d $&  & $-9-4d$ & $b=0$ &  & Large $M$\\
(Eulerian discretisation) & & & & & & \\[3ex] 
\hline 
\end{tabular}

\label{table:nonlin} 
\end{table}

\section{Background}\label{section-2}
We introduce nonlinear PDEs, in particular Hamilton-Jacobi and hyperbolic PDEs, and a system of nonlinear ODEs. We also provide the background of the relevant quantum algorithms.

\subsection{Nonlinear PDEs}

    Nonlinear partial difference equations can be written in the following general form

\begin{equation}
    \frac{\partial u}{\partial t} + F(u, \nabla u, \nabla^2 u, \cdots)=0, \quad t\in \mathbb{R}^+, x\in \mathbb{R}^d, u\in \mathbb{R}^d
\end{equation}
Here $t\ge 0$ is time, $x$ is the spatial variable, while $F$ is a nonlinear function or functional. \\

For most differential equations, linear or nonlinear, numerical computations have been the most important tools to solve them, since analytical solutions seldom exist. In fact, classical algorithms--referring  to numerical algorithms using classical computers--for differential equations have been among the most important achievements in scientific computing in more than half a decade. \\

Different partial differential equations have drastically different behavior in their solutions, from different regularities (smoothness) 
to different physical behaviors (e.g. conservation, invariances, entropy conditions), thus they call for drastically different numerical strategies. In this paper, we will mainly focus on first order quasi-linear PDEs (such as nonlinear Hamilton-Jacobi and
hyperbolic equations). \\

For more general PDEs, one can write it as a system of ODEs by first discretising in space. There are two classes of discretisation methods. One is the so-called Eulerian framework, in which one discretises the spatial variables by a finite difference, finite element, finite volume,
or spectral method, on a {\it fixed} mesh. The advantage of Eulerian methods is their high order accuracy but they suffer from the curse-of-dimensionality in high spatial dimensions. The other is the Lagrangian method, including the so-called particle or mesh-free method. This is popular for high dimensional problems, for example particle or Monte-Carlo methods for kinetic equations (the Boltzmann equation and Vlasov type equations), and vortex methods for incompressible
Euler and Navier-Stokes equations in fluid dynamics. The Lagrangian methods do not suffer from the curse-of-dimensionality but they are of lower order (typically only half to first order) methods. We briefly discuss the two methods in Section~\ref{sec:generalnonlinear}.\\

In our notation throughout the paper, we use $O(\cdot)$ to denote the case where all constant factors are suppressed unless otherwise stated. $\mathcal{O}(\cdot)$ is the standard big-O notation denoting a tight asymptotic upper bound and $ {\scriptstyle \mathcal{O}}(\cdot)$ is the standard little-O notation, which is a looser asymptotic upper bound. $\tilde{O}(\cdot)$ indicates that in addition, all logarithmic factors are suppressed. The notation $``(\cdot) \sim (\cdot)"$ denotes equivalence with constant factors ignored. 

\subsubsection{Hamilton-Jacobi PDEs}\label{sec:H-J}

A distinct feature of first-order quasi-linear PDEs, including Hamilton-Jacobi or nonlinear hyperbolic equations, is that solutions can become {\it singular} even if the initial data is smooth. For nonlinear hyperbolic equations, shocks (solution becomes discontinuous) may develop, while for Hamilton-Jacobi equations the solution may form cusps (at which points the derivatives of the solution becomes discontinuous) corresponding to caustics in geometric optics. \\

When solutions become singular, one needs to make sense of the equations since the derivative terms are no longer well-defined.
In this regard one either uses the notion of {\it viscosity} solutions, which are physically relevant in applications such in gas dynamics for compressible Euler equations \cite{Lax-notes}, or optimal control
using the Hamilton-Jacobi-Bellman equation \cite{CL83, OS88, lasry2007mean}. Another notion is the
{\it multi-valued} solutions, which are relevant to applications such as semi-classical quantum dynamics, and geometric optics, in which the dynamics is time-irreversible and the solution satisfies the linear superposition principle \cite{Whitham, SpMaMa, JO03}.
When the solutions are smooth, both notions define the same solution, but not while singularities emerge. 
It is the latter case that will be considered in this paper.\\

\noindent Hamilton-Jacobi equations arise for instance in geometric optics, the semiclassical limit of the Schr\"odinger equation, the level set formulation of front propagation, optimal control, mean-field games, sticky particles or pressureless gases and KPZ equations.  It has the following general form 
\begin{eqnarray}\label{H-J}
 && \partial_t S^{[k]}+ H(\nabla S^{[k]}, x) = 0, \quad t\in \mathbb{R}^+,  x \in \mathbb{R}^d, S^{[k]}(t,x)\in \mathbb{R},\\
 &&S^{[k]}(0,x)=S_0^{[k]}(x), \quad k=1,...,M
\end{eqnarray}
subject to $M$ different initial data, where the same Hamilton-Jacobi PDE is satisfied for each $k$. \\

A direct numerical approximation to Eq.~\eqref{H-J} usually gives rise to the so-called viscosity solution \cite{CL83, OS88}. This notion of the solution is not valid in geometric optics (multiple arrivals in seismic waves \cite{FS02}, for example), in the semiclassical limit of quantum dynamics, or in the high frequency limit of linear wave equations
(elastic waves, electromagnetic waves, etc.) \cite{ER03, SpMaMa, JMS11}, since it violates the linear superposition principle. In these applications, one is interested in computing the {\it multi-valued} solution,  \cite{SpMaMa, JL03, JO03, JLOT1}. There are several classes of  algorithms that were developed to capture such solutions:
\begin{itemize}
    \item Ray tracing \cite{Gla89, Ben96}. This is based on solving the characteristics
    of the system, which is a Hamiltonian system
    \begin{equation}\label{H-S}
    \partial_t x(t) = \nabla_p H(x,p)\,, \quad
    \partial_t p(t) = -\nabla_x H(x,p)\,.
    \end{equation}

The advantage of this method is its simplicity since it just solves  a system of ODEs. The disadvantage of such Lagrangian type methods is that particles are not uniformly distributed, hence at later time there may be regions where there are not enough particles to guarantee numerical accuracy and one needs to add more particles and then use interpolations to define these particle, which are quite delicate.
  \item Moment methods \cite{ER96, JL03}. For multivalued solutions one can use moment systems, which are superposition of Hamilton-Jacobi equations. The advantage of this method is that one stays in the physical space, however, the moment systems are difficult to derive in higher-dimensions and, in particular,  one needs to know, {\it a priori}, the number of branches in order to have the right number of moments. This is pretty much an impossible task.
  \item Level set formulation. Here one builds the gradient of $S$  of Eq.~\eqref{H-J} into the zero level sets (defined later) of functions which solves a system of Liouville equations \cite{JO03, CLO, JLOT1}. This method is globally valid and one solves a {\it linear} system of PDEs. Its disadvantage, for classical computers, is the curse-of-dimensionality since the equations are defined in the phase space, hence the dimension is doubled. Since the curse-of-dimensionality can in cases be resolved with quantum computers, and thanks to the linearity of the system, this is the approach we advocate in this paper.
\end{itemize}

\noindent Define $u^{[k]}=\nabla S^{[k]} \in \mathbb{R}^d$. Then $u^{[k]}$ solves a hyperbolic system of conservation laws in gradient form:
\begin{eqnarray}\label{forced-Burgers}
 && \partial_t u^{[k]} + \nabla H(u^{[k]}, x) = 0, \\
 &&u^{[k]}(0,x)=\nabla S^{[k]}_0(x).
\end{eqnarray}

\noindent Two classical examples of Hamiltonians are
\begin{equation}\label{H-newton}
H(x, p) = \frac{p^2}{2} + V(x),
\end{equation}
corresponding to classical Newtonian particles, and
\begin{equation}\label{H-GO}
H(x, p) = c(x) |p| 
\end{equation}
that arises in geometric optics and the level set formulation of front propagation \cite{ER03, OS88}, in which $c(x)$ is the reciprocal of the index of reflection or the speed of a propagating front in the normal direction. \\

When we encounter multi-valued solutions of Eq.~\eqref{forced-Burgers}, we define an ensemble average in Section~\ref{sec:HJ}, which can be identified as physical observables of the system. Then we show how it can be computed with a quantum algorithm. The classical cost for computing these observables is in the following lemma. 

\begin{lemma} \label{lem:HJ}
If a finite difference or finite volume method is used based on  a regular mesh using  $N_{\text {HJ}}$ spatial points  in each dimension, with mesh size $h_{\text {HJ}}=1/N_{\text {HJ}}$, and time step 
$\Delta t_{\text {HJ}}=O(h_{\text {HJ}}/d)$ (due to the CFL stability condition), then the computational cost of solving the Hamilton-Jacobi equation
Eq.~\eqref{H-J} is $O(d N_{t,\text{HJ}}N_{\text{HJ}}^d)=O(d^2 T N_{\text {HJ}}^{d+1})$.  To reach an error tolerance of $\epsilon_{\text{HJ}}$ one needs $N_{\text {HJ}}=O(d/\epsilon_{\text{HJ}})$, hence the total cost is $O(Td^{d+3}(1/\epsilon_{\text{HJ}})^{d+1})$. If one is interested in computing the ensemble average of $S^{[k]}, (k=1, \cdots, M)$ in the presence of $M$ different initial data, then the cost will be  $O(MTd^{d+3}(1/\epsilon_{\text{HJ}})^{d+1})$. The cost of computing the ensemble average of $u^{[k]}$ is $O(MTd^{d+4}(1/\epsilon_{\text{HJ}})^{d+1})$. 
\end{lemma}

\begin{proof}
The cost of $O(dN_{t, \text{HJ}}N_{\text {HJ}}^d)$ is obtained by directly counting. If one computes to time 
$T=N_{t,{\text {HJ}}} \Delta t_{\text {HJ}}$, since $\Delta t=O(h_{\text {HJ}}/d)$ (due to the CFL condition) and $h_{\text {HJ}}=1/N_{\text {HJ}}$ one also obtains 
$O(d^2 T N_{\text {HJ}}^{d+1})$. Since the truncation error of a first order method to approximate the derivative in each dimension is of $O(1/N_{\text {HJ}})$, and there are $d$ differential operators to be discretised, the total truncation error is of  $O(d/N_{\text {HJ}})$.
Hence to reach the error $\epsilon_{\text {HJ}}$ one needs $N_{\text {HJ}}=O(d/\epsilon_{\text {HJ}})$. This means the cost is $O(d^2 T N_{\text {HJ}}^{d+1})=O(Td^{d+3}(1/\epsilon_{\text {HJ}})^{d+1})$. If one is interested in computing the ensemble average of $S^{[k]}. (k=1, \cdots, M)$ in the presence of $M$ different initial data, the total cost will then be $O(MTd^{d+3}(1/\epsilon_{\text {HJ}})^{d+1})$. The cost of computing the  ansemble average of $u^{[k]}$ is then  $O(MTd^{d+4}(1/\epsilon_{\text {HJ}})^{d+1})$ since $u^{[k]}$  is $d$-dimensional. 
\end{proof}

\bigskip


\noindent{\bf Remark:} Spectral methods are usually not used for Hamilton-Jacobi solutions since the solutions develop  singularities (caustics) and spectral methods introduce numerical oscillations.\\

\subsubsection{Nonlinear scalar hyperbolic PDEs}

 Nonlinear hyperbolic PDEs arise for instance in gas dynamics, combustion, magnetohydrodynamics, shallow water and traffic flows. Here we focus on the scalar equation where $u^{[k]}(t,x)\in \mathbb{R}$ is a scalar solving an initial
value problem of an $(d+1)$-dimensional first-order hyperbolic
PDE with a non-zero source term
\begin{eqnarray}\label{hyp-PDE}
  &\partial_t u^{[k]} + F(u^{[k]}) \cdot \nabla_x u^{[k]} + Q(x,u^{[k]})=0, \quad t\in \mathbb{R}^+, \quad x\in \mathbb{R}^d,\\
  \label{hyp-IC}
   & u^{[k]}(0,x)=u^{[k]}_0(x), \quad k=1,...,M.
  \end{eqnarray}
  Here $F(u^{[k]}): \mathbb{R} \to \mathbb{R}^d$ is a vector and
  $Q: \mathbb{R}^{d+1} \to \mathbb{R}$ is the source
  term. This equation includes any such hyperbolic
  PDE in conservative or non-conservative form. The cost in solving this equation is essentially the same as 
  those in Lemma \ref{lem:HJ} for  $S^{[k]}$.
  
  \begin{lemma} \label{lem:classicalhyperbolic}
  If a finite difference or finite volume method is used with a regular mesh using $N_{hyp}$ points per space dimension, the computational cost of solving 
Eq.~\eqref{hyp-PDE} with $M$ initial data in Eq.~\eqref{hyp-IC} and error tolerance of $\epsilon_{\text {hyp}}$ is  $O(MTd^{d+3}(1/\epsilon_{\text {hyp}})^{d+1})$.
  \end{lemma}
  

  Finally we point out that for Hamilton-Jacobi equations and general nonlinear
  hyperbolic PDEs, not all initial data lead to caustics or shocks, hence the solutions
  may remain smooth for all time, thus multivalued-solutions will not appear. In such cases,
  the solutions computed are exactly the smooth solutions to the original
  nonlinear PDEs without needing to use the notion of viscosity or multi-valued solutions.
  
  
  \subsection{System of nonlinear ODEs}

A system of $D$ nonlinear ODEs subject to $M$ different initial data can be written as 

 \begin{align} \label{eq:upde}
 &   \frac{d X^{[k]}(t)}{d t}=F( X^{[k]}(t)), \qquad X^{[k]}\in \mathbb{R}^D\,,\\
 & X^{[k]}(0)=X_0^{[k]}, \qquad  k=1,\cdots, M
\end{align}
where the same ODE is satisfied for each $k$. This can be  interpreted as a system of $M$ non-interacting particles in $D$ dimensions, each with a trajectory described by $X^{[k]}(t)$. For nonlinear ODEs, $F(X)\in \mathbb{R}^D$ are  nonlinear functions of its argument $X$. To compute ensemble averages 
 \begin{align} \label{Ens-ODE}
     \langle A(t)\rangle =\frac{1}{M} \sum_{k=1}^M A(X^{[k]}(t))
\end{align}
one can directly solve for the system of ODEs  starting from $M$ different initial data $X_0^{[k]}$, and then carry out the summation in Eq.~\eqref{Ens-ODE}.  \\

\begin{lemma} \label{lem:odeclassical}
The classical computational cost in computing the ensemble observable $\langle A(T)\rangle$ for $M$ initial data, to precision $\epsilon_{\text {ODE}}$ at time  $T= N_{t, \text {ODE}} \Delta t_{\text {ODE}}$ for time-step size $\Delta t_{ODE}$ is $O({D^3}MT/\epsilon_{\text {ODE}})$.
\end{lemma}
\begin{proof}
Assume the evaluation of each component of $F(X)$  costs, at most, $O(D)$ operations. Then for the first order method (say the forward Euler method), the error is of  $O(D\Delta t_{\text{ODE}})$. To reach an error
$\epsilon_{\text {ODE}}$ one needs $\Delta t_{\text{ODE}}=O(\epsilon_{\text {ODE}}/D)$. The ensemble average step in Eq.~\eqref{Ens-ODE} costs $O(M)$ since it is just a sum of $M$ terms. The computational
cost of solving Eq.~\eqref{eq:upde} and computing the ensemble average in Eq.~\eqref{Ens-ODE} is  $O(D^2MN_{t,{\text{ODE}}})=O({D^3}MT/\epsilon_{\text {ODE}})$. 
\end{proof}

\noindent \textbf{Remark:} The ODE solver, by the spectral deferred correction method \cite{SpecDef}, can achieve the complexity
of \\
$O(DMm\log(m)/\epsilon_{\text {ODE}}^{1/m})$ for any pre-chosen positive integer $m$. So if one wants to
choose a large $m$, the cost is about $O(DMm\log(m))$. For simplicity, we won't consider spectral methods in this paper, but they can be the investigation of future work.


\subsection{Quantum subroutines}
We now briefly review the quantum linear systems problem (QLSP), which are useful for solving ODEs and PDEs. It should be emphasized that the output of QLSP for ODEs and PDEs are {\it quantum states} and {\it not} the classical solutions of the ODEs and PDEs, thus making these {\it quantum subroutines} rather than full quantum algorithms. We then describe the system of linear equations problem (SLEP) which allows one to compute observables at the output of QLSP, thus allowing the quantum algorithm to solve the same problem as the classical algorithm.

\subsubsection{Quantum linear systems problem: QLSP}
The quantum linear systems problem (QLSP)  \cite{childs2017quantum, barrynew} can be stated informally in the following way.

\begin{problem} \label{prob:one}
\textbf{(QLSP)} Let $\mathcal{M}$ be a $2^m\times 2^m$ Hermitian matrix such that $\|\mathcal{M}\|\le 1$. Assume vectors $x$ and $y$ with elements $\{x_i\}$, $\{y_i\}$ that satisfy $\mathcal{M}x= y$. One can then define the following $m$-qubit quantum states $|x\rangle \equiv \sum_i x_i/N_{x} |i\rangle$,  $|y\rangle \equiv \sum_i y_i/N_{y} |i\rangle$ where $N_{x}=\sqrt{\sum_i |x_i|^2}$, $N_{y}=\sqrt{\sum_i |y_i|^2}$ are normalisation constants. The aim of any QLSP algorithm is, when given access to $\mathcal{M}$ and unitary $U_{initial}$ (where $U_{initial}|0\rangle=|y\rangle$), to prepare the quantum state $|x'\rangle$ that is $\eta$-close to $|x\rangle$, i.e., $\| |x'\rangle-|x\rangle \|  \leq \eta$.
\end{problem}
The most notable algorithms to solve QLSP are the HHL algorithm \cite{harrow2009quantum} built on quantum phase estimation and the alternative CKS algorithm \cite{childs2017quantum} that bypasses quantum phase estimation, where the latter can provide an exponential improvement in precision. Their output is an approximation to the quantum state $|x\rangle$ instead of an approximation to the solution of the original vector problem $x=\mathcal{M}^{-1}y$. Thus these are often termed quantum subroutines instead of full quantum algorithms, since they provide a stepping stone but do not solve the same problem as the corresponding classical algorithm. \\

In Problem~\ref{prob:one}, the informal phrase `when given access to $\mathcal{M}$' in order to solve QLSP, refers to a description of how the entries of $\mathcal{M}$ can be accessed during the protocol. The total cost of the protocol would also be computed with respect to the type of access one has. For both algorithms \cite{harrow2009quantum, childs2017quantum} the assumption is of sparse access to $\mathcal{M}$, defined in the following way \cite{berry2015hamiltonian, barrynew}.

\begin{definition}
Sparse access to a Hermitian matrix $\mathcal{M}$ is a $4$-tuple $(s, \|\mathcal{M}\|_{max}, O_{\mathcal{M}}, O_{F})$ and the $(i,j)^{\text{th}}$ entry of $\mathcal{M}$ is denoted $\mathcal{M}_{ij}$. Here $s$ is the sparsity of $\mathcal{M}$ and $\|\mathcal{M}\|_{max}=\max_{i,j}(|\mathcal{M}_{ij}|)$ is the max-norm of $\mathcal{M}$. $O_M$ and $O_F$ are unitary black boxes which can access the matrix elements $\mathcal{M}_{ij}$ such that 
\begin{align}
    &O_M|j\rangle|k\rangle|z\rangle=|j\rangle|k\rangle|z\oplus\mathcal{M}_{jk}\rangle \nonumber \\
    &O_{F} |j\rangle|l\rangle=|j\rangle|F(j,l)\rangle
\end{align}
where the function $F$ takes the row index $j$ and a number $l=1,2,...,s$ and outputs the column index of the $l^{\text{th}}$ non-zero elements in row $j$. 
\end{definition}

Then the HHL algorithm \cite{harrow2009quantum} has the following query complexity, which denotes the number of times oracles $O_M$, $O_F$ and $U_{initial}$ are used throughout the protocol. The gate complexity refers to the number of $2$-qubit gates required in the algorithm.

\begin{lemma} \label{lem:HHL}
\cite{harrow2009quantum,childs2017quantum} Let $\mathcal{M}$ be Hermitian and each copy of $|y\rangle$ is provided by the oracle $U_{initial}$. To create a state that is $\eta$-close to $|x\rangle$, it is sufficient that the oracles $O_M, O_F$ are queried $O((s \kappa^2 \|\mathcal{M}\|_{max}/\eta) \text{poly}\log(s \kappa \|\mathcal{M}\|_{max}/\eta))$ times, where $\kappa$ is the condition number of $\mathcal{M}$. The number of 2-qubit gates required in this algorithm is at most logarithmically larger than the query complexities of $O_M, O_F$. The query complexity for the oracle $U_{initial}$ is $\mathcal{O}(s \kappa \|\mathcal{M}\|_{max} \text{poly}\log(s \kappa \|\mathcal{M}\|_{max}/\eta))$. \\
\end{lemma}

The $1/\eta$ dependence in the query complexities of $O_M, O_F$ comes from using quantum phase estimation. To improve upon this factor, the CKS algorithm in \cite{childs2017quantum} bypasses phase estimation and instead uses a sequence of unitaries whose sum approximates $\mathcal{M}^{-1}$. 

\begin{lemma} \label{lem:series}
\cite{childs2017quantum} A state $\eta$-close to $| x\rangle$ can be created by querying the $O_M, O_F$ and $U_{initial}$ oracles \\
$\mathcal{O}(s \kappa \|\mathcal{M}\|_{max} \text{poly}\log(s \kappa \|\mathcal{M}\|_{max}/\eta))$ times, where the number of 2-qubit gates are at most logarithmally larger than the query complexity. \end{lemma}

An important application of QLSP is preparing quantum states whose amplitudes are proportional to the solutions for linear ODEs and PDEs, known as amplitude-encoding of the solutions. The first step is to discretise the ODEs and PDEs and transform the equations into a linear algebra problem of the form $x=\mathcal{M}^{-1}y$, where the size of $\mathcal{M}$ can be very large. The quantum subroutines for QLSP is then in the matrix inversion process to prepare the corresponding $|x\rangle$ states. The exact query and gate complexities with respect to error $\eta$, time $T$ and the dimension $d$ of the problem would also depend on details of the discretisation procedure.\\

We give full details of our discretisation procedures and the corresponding matrix inversion problem in the main body of the paper. \\


\subsubsection{System of linear equations problem: SLEP}
When solving ODEs and PDEs, the actual desired outcomes of the problem are the observables associated with the solutions of the ODEs and PDEs. Solving QLSP only prepares $|x\rangle$, whereas the system of linear equations problem (SLEP) \cite{barrynew} aims to compute observables from $x$. Then given the same definitions as Problem~\ref{prob:one}, one can state SLEP in the following way. 
\begin{problem} \label{prob:two}
\textbf{(SLEP)} Given a Hermitian matrix $\mathcal{G}$, which is of the same size as $\mathcal{M}$, access to $\mathcal{M}$ and $U_{initial}$, the aim of SLEP is to compute the expectation value $(\mathcal{M}^{-1}y)^T \mathcal{G} (\mathcal{M}^{-1}y)=x^T \mathcal{G} x$ to precision $\epsilon'$.
\end{problem}
There are various different methods of measuring the outcome $x^T \mathcal{G} x$ directly after obtaining $|x\rangle \propto \mathcal{M}^{-1}|y\rangle$ from the output of QLSP. Many of these methods have a query complexity with error scaling as $1/(\epsilon')^{2}$ \cite{harrow2009quantum}, for instance the quantum swap test \cite{buhrman2001quantum} when $\mathcal{G}$ is a density matrix or applying the Hadamard test when $\mathcal{G}$ is decomposed as a sum of two unitary operators \cite{aharonov2009polynomial}. An elegant formalism that generalises the quantum methods for matrix inversion (that uses sparse access to $\mathcal{M}$), as well as neatly achieving the improved optimal $1/\epsilon'$ scaling via amplitude estimation \cite{knill2007optimal}, without too many extra assumptions, is the formalism of block access \cite{low2019hamiltonian,qsvd, barrynew}, defined below.   It is possible to create block access to $\mathcal{M}$ from sparse access to $\mathcal{M}$ (see Lemma~\ref{lem:sparsetoblock}), a fact which we later exploit to approximate $x^T \mathcal{G} x$ while still beginning from sparse access to $\mathcal{M}$.

\begin{definition}
Let $\mathcal{M}$ be a $m$-qubit Hermitian matrix, $\delta_{\mathcal{M}}>0$ and $n_{\mathcal{M}}$ is a positive integer. A $(m+n_{\mathcal{M}})$-qubit unitary matrix $U_{\mathcal{M}}$ is a $(\alpha_{\mathcal{M}}, n_{\mathcal{M}}, \delta_{\mathcal{M}})$-block encoding of $\mathcal{M}$ if 
\begin{align} \label{G-def}
    \|\mathcal{M}-\alpha_{\mathcal{M}}\langle 0^{n_{\mathcal{M}}}|U_{\mathcal{M}} |0^{n_{\mathcal{M}}} \rangle \|\leq \delta_{\mathcal{M}}.
\end{align}
Block access to $\mathcal{M}$ is then the 4-tuple $(\alpha_{\mathcal{M}}, n_{\mathcal{M}}, \delta_{\mathcal{M}}, U_{\mathcal{M}})$ where $U_{\mathcal{M}}$ is the unitary black-box block-encoding of $\mathcal{M}$. 
\end{definition}
In the rest of the paper, we assume that if the block access $(\alpha_{\mathcal{M}}, n_{\mathcal{M}}, \delta_{\mathcal{M}}, U_{\mathcal{M}})$ to $\mathcal{M}$ is given, then $U^{\dagger}_{\mathcal{M}}$, controlled-$U_{\mathcal{M}}$ and controlled-$U^{\dagger}_{\mathcal{M}}$ are also given. \\

Then in Problem~\ref{prob:two}, if one assumes that access to $\mathcal{M}$ refers to block access to $\mathcal{M}$ instead of sparse access, there exists an algorithm \cite{barrynew} that solves SLEP with the following query and gate complexities.

\begin{lemma} \label{lem:bslep}
\cite{barrynew} A quantum algorithm can be constructed that takes $m$ qubits, block access $(\alpha_{\mathcal{M}}, n_{\mathcal{M}}, 0, U_{\mathcal{M}})$ to a $2^m \times 2^m$ invertible Hermitian matrix $\mathcal{M}$ with condition number $\kappa$, block access $(\alpha_{\mathcal{G}}, n_{\mathcal{G}}, 0, U_{\mathcal{G}})$ to a $2^m \times 2^m$ Hermitian matrix $\mathcal{G}$, an accuracy $\epsilon \in [\alpha_{\mathcal{G}}/2^m, \alpha_{\mathcal{G}}]$, and a $m$-qubit unitary black box $U_{initial}$, and returns with probability at least $2/3$ an $\epsilon'$-additive approximation to $x^T \mathcal{G} x$, where $x=\mathcal{M}^{-1}y$ and $y$ is the $2^m$-dimensional vector with entries $y_i=\langle i|U_{initial}|0\rangle=\langle i|y\rangle$, by making $\mathcal{O}(\alpha_{\mathcal{G}}\kappa^2/\epsilon')$ queries to $U_\mathcal{G}$, $\mathcal{O}(\alpha_{\mathcal{G}}\kappa^3\log(\alpha_{\mathcal{G}}\kappa^2/\epsilon)/\epsilon')$ queries to $U_{\mathcal{M}}$, $\mathcal{O}(\alpha_{\mathcal{G}}\kappa^2/\epsilon')$ queries to $U_{initial}$ and $\mathcal{O}(\alpha_{\mathcal{G}}\kappa^2(m+n_{\mathcal{G}}+n_{\mathcal{M}}+n_{\mathcal{M}}\kappa \log(\alpha_{\mathcal{G}}\kappa^2/\epsilon'))/\epsilon')$ additional $2$-qubit gates.  
\end{lemma}
We later use a modified version of this result, when given sparse access to $\mathcal{M}$, for computing ensemble averages from nonlinear ODEs and PDEs. 


\section{Solving Hamilton-Jacobi equations} \label{sec:HJ}

In this section, we develop a quantum algorithm for computing observables from Hamilton-Jacobi equations. The first step is to transform the nonlinear equation into a linear equation \textit{without} making any approximations. One convenient method is to use the level set formulation \cite{JO03}. We demonstrate below (Section~\ref{sec:levelset}) how to use this technique to translate a $(d+1)$-dimensional nonlinear Hamilton-Jacobi equations into $(2d+1)$-dimensional (linear) Liouville equations. By discretising this new PDE, we can convert this into a linear algebra problem. We also present physical observables that can be computed in our formalism in Section~\ref{sec:levelsetobservable} and present explicit error bounds. Finally in Section~\ref{sec:quantumlevelset} we show how to compute the observable with a quantum algorithm and present its query and gate complexities.

\subsection{Linear represenation of the  nonlinear Hamilton-Jacobi equation} \label{sec:levelset}

Here we show how to transform a nonlinear Hamilton-Jacobi equation into a linear algebra problem by a two-step procedure: using a level set function to convert the nonlinear equation into a linear equation and then to discretise that linear equation.\\

The level set function $\phi^{[k]}_i(t,x,p)$ can be defined by
\begin{equation}\label{LS-def}
\phi^{[k]}_i(t, x,p=u^{[k]}(t,x))=0
\end{equation}
where $i=1, \cdots, d$ and $\, x, p\in \mathbb{R}^d$, $k=1,...,M$ and $\{u^{[k]}(t,x)\}$ are the solutions of the nonlinear Hamilton-Jacobi equation in gradient form with $M$ initial conditions in Eq.~\eqref{forced-Burgers}. The \textit{zero level set} of $\phi^{[k]}_i$ is the set  $\{(t,x,p)|\phi^{[k]}_i(t,x,p)=0\}$.\\

Since $u^{[k]}(t,x)$ solves Eq.~\eqref{forced-Burgers}, then one can show that  $\phi^{[k]}=(\phi^{[k]}_1, \cdots, \phi^{[k]}_d)\in \mathbb{R}^d$ solves a (linear!) Liouville equation \cite{JO03}
\begin{equation} \label{LS-Liouville}
\partial_t \phi^{[k]} + \nabla_p H \cdot \nabla_x \phi^{[k]} - \nabla_x H \cdot \nabla_p \phi^{[k]}=0.
\end{equation}
Note that the (bi)-characteristics of the Liouville equation
in Eq.~\eqref{LS-Liouville} is the Hamiltonian system in Eq.~\eqref{H-S}. The initial data can be chosen as 
\begin{equation}\label{LS-ic}
\phi^{[k]}_i(0, x,p)=p_i - u^{[k]}_i(0,x), \quad i=1, \cdots, d.
\end{equation}
Then $u^{[k]}$ can be recovered from the intersection of the zero level
sets of $\phi_i \, (i=1, \cdots, d)$, as in \eqref{intersect}, namely 
\begin{equation}
    u^{[k]}(t,x)=\{p(t,x)| \,\phi_i^{[k]}(t,x,p)=0, \, i=1,\cdots, d \}. 
\end{equation}
Note that  $\phi_i^{[k]}(t,x,p)=0$ may have multiple (say $J_k$) roots, denoted by $p_{\gamma(t,x)} (\gamma=1, \cdots, J_k)$, hence the so-called  multi-valued solutions will arise, which is denoted by
\begin{equation}\label{intersect}
    u^{[k]}_{\gamma}(t,x)= \cap_{i=1}^d \{p_{\gamma}(t,x) | \phi^{[k]}_i(t,x,p_{\gamma})=0\}, \quad \gamma=1, \cdots, J_k.
\end{equation}
 We include this possibility when defining our observable in Section~\ref{sec:levelsetobservable}. \\

It is crucial to observe here that we have now transformed a $(d+1)$-dimensional nonlinear PDE-- the Hamilton-Jacobi equation-- to a $(2d+1)$-dimensional \textit{linear} PDE --the Liouville equation--without \textit{any} assumptions on either the form or extent of the original nonlinearity. {\it No linear approximation is made}. The mapping is {\it exact}. The cost in going from a nonlinear to a linear system is only at the expense of doubling the dimension. Doubling the dimension of the problem may seem too costly for a classical device because the cost increases exponentially with dimension. However, we will see that for quantum algorithms, the relative overhead in doubling the dimension can be up to exponentially smaller. \\

While in principle one could apply quantum subroutines for QLSP to Eq.~\eqref{LS-Liouville} and generate the quantum states $|\phi\rangle$ whose amplitudes are proportional to $\phi_i$, it would be too costly to recover our desired solutions $u$ or observables from $|\phi\rangle$ due to the extra measurement costs in finding the zero level set of $\phi_i^{[k]}\, (i=1, \cdots d)$. An alternative method is to solve for $\psi$, defined by the following problem 
\begin{eqnarray} \label{Liouville-delta}
&& \partial_t \psi + \nabla_p H \cdot \nabla_x \psi - \nabla_x H \cdot \nabla_p \psi=0 
\end{eqnarray}
with the initial condition
\begin{eqnarray}\label{Liou-I}
&& \psi(0, x,p)=\frac{1}{M}\sum_{k=1}^M\prod_{i=1}^d\delta(p_i - u^{[k]}_i(0,x)). 
\end{eqnarray}
Then we have the following result. 

\begin{lemma}\label{thm-psi} The analytical solution of (\ref{Liouville-delta}) with initial data (\ref{Liou-I}) is 
\begin{align} \label{eq:psisolution}
\psi(t,x,p) =\frac{1}{M} \sum_{k=1}^M \delta(\phi^{[k]}(t,x,p))\,.
\end{align}
\end{lemma}
\begin{proof}
See Appendix~\ref{app:hjjustification}
\end{proof}

An important observation here is that {\it all} $M$ distinct initial conditions of $u^{[k]}$ in the original problem have now been converted into a \textit{single} initial condition in $\psi$. We will exploit this property later to  show that the resource cost for the quantum algorithm is {\it independent} of $M$. \\

\noindent{\bf Remark:} Instead of \eqref{Liou-I}, one can define a more general average
\begin{eqnarray}\label{Liou-I-g}
&& \psi(0, x,p)=\frac{1}{M}\sum_{k=1}^M c_k \prod_{i=1}^d\delta(p_i - u^{[k]}_i(0,x)). 
\end{eqnarray}
where $0<c_k<1$ and $\sum_{k} c_k =1$. These weights will persist to all later times, and  one has

\begin{align} \label{eq:generalpsi}
    \psi(t,x,p) =\frac{1}{M} \sum_{k=1}^M c_k \delta(\phi^{[k]}(t,x,p))
\end{align}
and all physical observables, which are moments of $\psi$, are weighted average --with the same weights-- of individual corresponding observables.\\

One can convert the linear PDE in Eq.~\eqref{Liouville-delta} with initial conditions in Eq.~\eqref{Liou-I} into a linear algebra problem by discretising the function
\begin{align}
    & \psi (t, x, p) \rightarrow \psi^{\omega}_{n, \vect{j}, \vect{l}}
\end{align}
where $n$ denotes the time step from $t \rightarrow t_n \equiv n \Delta t$ and $n=0,...,N_t$. The vectors $\vect{j}=(j_1,...,j_d)$, $\vectorbold{l}=(l_1,...,l_d)$ denote spatial grid indices with grid size $h=1/N$ where $x \rightarrow h \vect{j}$, $p \rightarrow h \vect{l}$ are grid points with  $j_i,l_i=1,..,N$ for $i=1,..,d$. \\

Throughout the paper we always assume, without loss of generality, the computational domain to be in a box of $[0,1]^D$, where $D$ is the spatial dimension of the problem. Since all the linear PDEs considered in this paper are transport equations, due to their finite propagation speeds,
by time $T<\infty$, as long as the initial data have compact support in $x$--which we assume here for $u_0(x)$-- solution will still have compact support for $t\le T$ in  $x$ domain. As far as the domain in $p$ is concerned, notice that we always start with delta functions in $p$, which has a compact support as long as $u(t,x)$ is bounded, so the computational domain in $p$ can also be restricted to a finite domain as long as 
$T<\infty$. Upon a suitable scaling we confine our computational domain, in both $x$ and $p$, within the box $[0,1]^D$.\\

Since the initial condition in Eq.~\eqref{Liou-I} involves a delta function, one also needs a discretised delta function $\delta_\omega$ where $\omega$ is a smoothing parameter of the delta function. Hence for the initial state at $n=0$ we have 
\begin{align} \label{eq:psiinitialtext}
    \psi^{\omega}_{0, \vect{j}, \vect{l}}= \frac{1}{M}\sum_{k=1}^M  \prod_{i=1}^d\delta_\omega(l_ih-u^{[k]}_i(n=0, h\vect{j})).
\end{align}
As conventionally done, we  choose $\delta_\omega$ to be smooth and to satisfy, for $x\in\mathbb{R}^1$,
\begin{equation}
\delta_\omega (x) = 0 \quad {\text {if} } \quad |x|> \omega; \quad
\int_{|x|\le \omega} \delta_\omega(x) \, dx =1.
\end{equation}
One usually  approximates $\delta_\omega$ by the form
\begin{equation} \label{eq:deltaomega}
\delta_\omega(x) =
\begin{cases} \frac{1}{\omega}\beta(x/\omega) \quad 
           & |x|\le \omega; \\
            0 \quad 
            &|x|> \omega
\end{cases}
\end{equation}
where typical choices of $\beta(x)$ include
$\beta(x)=1-|\beta|$ and $\beta(x)=\frac{1}{2}(1+\cos(\pi x))$ \cite{EngTorn}. Here one can choose $\omega=mh$ where $m$ is the number of mesh points within the support of $\delta_\omega$. For $x=(x_1, \cdots, x_d) \in \mathbb{R}^d$, one defines $\delta_\omega(x) \equiv  \Pi_{i=1}^d \delta_\omega(x_i)$. \\

The solution of the discretised PDE can then be written as the following matrix equation 
\begin{align} \label{eq:psimatrix-1}
    \begin{pmatrix}
    \psi_{1, \vect{j}, \vect{l}} \\
    \psi_{2, \vect{j} ,\vect{l}} \\
    \vdots \\
    \psi_{N_t-1, \vect{j}, \vect{l}} \\
    \psi_{N_t, \vect{j}, \vect{l}}
    \end{pmatrix}=\mathcal{K}^{-1} \begin{pmatrix}
    \psi_{0, \vect{j}, \vect{l}} \\
    0 \\
    \vdots \\
    0 \\
    0
    \end{pmatrix}
\end{align}
where $\mathcal{K}$ is a $N_t N^{2d} \times N_t N^{2d}$ Toeplitz matrix. See Appendix~\ref{app:discretisation} for the form of $\mathcal{K}$ and the details of the discretisation procedure. \\

Both QLSP and SLEP involve the inversion of an Hermitian matrix $\mathcal{M}$. Since $\mathcal{K}$ is not Hermitian, one can define a new Hermitian matrix 
\begin{align}
    \mathcal{M}=\begin{pmatrix} 0 & \mathcal{K} \\
    \mathcal{K}^{\dagger} & 0 
    \end{pmatrix}
\end{align}
which has the same sparsity and condition number as $\mathcal{K}$. Using $\mathcal{M}$, the matrix inversion problem to solve $\psi_{n, \vect{j}, \vect{l}}$ becomes
\begin{align} \label{B11}
\begin{pmatrix}
    \vect{0} \\
    \psi_{n, \vect{j}, \vect{l}}
    \end{pmatrix}=\mathcal{M}^{-1} \begin{pmatrix}
    \psi_{0, \vect{j}, \vect{l}} \\
    \vect{0}
    \end{pmatrix}
\end{align}
   where $\vect{0}$ is a zero-vector of the same dimension as $\mathcal{K}$.

\begin{lemma} \label{lem:skappa} The condition number of $\mathcal{M}$ is $\kappa \leq O(d NT)$ where $T=N_t \Delta t$ is the stopping time, and sparsity is $s=O(d)$. 
\end{lemma}
\begin{proof}
See Appendix~\ref{app:skappaproof}.
\end{proof}

\subsection{The observables} \label{sec:levelsetobservable}


Given any function $G: \mathbb{R}^{d} \rightarrow \mathbb{R}$ one can define the following ensemble average which we call the  observable.
\begin{definition} \label{def:levelsetobservable}
The following ensemble average we define as the observable 
\begin{equation}\label{G-HJ}
  \langle G(t,x) \rangle \equiv \int_{\mathbb{R}^d} G(p) \psi(t,x,p) dp=\frac{1}{M} \sum_{k=1}^M \int_{\mathbb{R}^d} G(p)  \delta(\phi^{[k]}(t,x,p)) dp=\frac{1}{M} \sum_{k=1}^M  \sum_{\gamma=1}^{J_k} \frac{G(u^{[k]}_{\gamma}(t,x))}{\mathcal{J}^{[k]}_{\gamma}},
\end{equation}
where the Jacobian $\mathcal{J}^{[k]}_{\gamma} \equiv|\det(\partial{\phi^{[k]}}/\partial p)|_{p=u^{[k]}_{\gamma}(t,x)}$. 
\end{definition}
This is an ensemble average of solution over the $M$ different initial data,
each with  the multi-valued solutions of $J_k$ branches. The multi-valued solution is a weighted average of each branch with weights $1/\mathcal{J}^{[k]}_{\gamma}$ that depends on the level set function $\phi^{[k]}$. \\

\noindent \textbf{Remark:} This is easily extended to the case of more general weighting of the initial conditions in Eq.~\eqref{eq:generalpsi}.\\

 Assume $\psi$ is the solution to Eq.~\eqref{Liouville-delta}, and $u^{[k]}$ is the solution to Eq.~\eqref{forced-Burgers}. One can then compute $\langle G(x,t) \rangle$ by using the numerical quadrature rule
 \begin{align}\label{eq:discreteensemblepsi1}
     \langle G(t_n,x=\vect{j}/N) \rangle=\int_{\mathbb{R}^d} G(p) \psi(t_n,\vect{j}/N,p) dp \approx \frac{1}{N^d}\sum_{\vect{l}}^N G_{\vect{l}} \psi^{\omega}_{n, \vect{j}, \vect{l}} \equiv \langle G^\omega_{n,\vect{j}}\rangle
\end{align}
where $\psi^{\omega}_{n, \vect{j}, \vect{l}} \equiv \frac{1}{M} \sum_{k=1}^M \delta_{\omega}(\phi^{[k]}(t_n,\vect{j}/N,\vect{l}/N))$ and $G_{\vect{l}} \equiv G(\vect{l}/N)$. We also use the notation $\sum_{\vect{l}}^N \equiv \sum_{l_1=1}^N...\sum_{l_d=1}^N$.\\

We can also define an ensemble average normalised by its zeroth moment:
\begin{align} \label{eq:goobservable}
    G_O(t,x)\equiv\frac{\langle G(t,x)\rangle}{\langle \mathbf{1}\rangle}=\frac{\int_{\mathbb{R}^n} G(p)\psi(t,x,p)dp}{\int_{\mathbb{R}^n}\psi(t,x,p)dp}
\end{align}
where $\mathbf{1}$ is the identity function, corresponding to the zeroth moment of $\psi$. In the cases of  $M=1$ and when there are no multi-valued solutions, one has
\begin{align}
     G_O(t,x)=G(u(t,x)).
\end{align}
Our goal is to devise a quantum algorithm to compute the observable $\langle G(t,x) \rangle$ (also $G_O(t,x)$) to precision $\epsilon$ and demonstrate that it can be more efficient on a quantum device with respect to parameters $M$, $d$ and $\epsilon$ compared to a purely classical algorithm. In the next subsection we show how different $G$ leads to different {\it physical observables} captured by $\langle G(t,x)\rangle$ and $G_O$, such as density, momentum and energy.\\

\subsubsection{Physical interpretation} \label{sec:physicalinterpretation}

 We now give physical interpretations for the  observables defined in Eq.~\eqref{G-HJ} through several physically important examples. We show in these cases that for nonlinear Hamilton-Jacobi where the WKB approximation is applicable, the corresponding Wigner function, in the semiclassical limit of the Schr\"odinger equation, obeys exactly the same  Liouville equation satisfied by  $\psi(t,x,p)$. Given different choices of $G(p)$, the observable $\langle G(t,x) \rangle$ correspond to moments of this Wigner function and they
 are {\it exactly} the  physical observables like density, momentum and energy.\\
 

We first consider the classical limit (WKB approximation) of the Schr\"odinger equation with $M=1$:
\begin{align}
    \label{schro}
    i \hbar \partial_t \Psi = -\frac{\hbar^2}{2} \Delta \Psi + V(x) \Psi\,,\qquad
    \Psi(0, x) = A_0(x) e^{i \frac{S_0(x)}{\hbar}},
\end{align}
with wavefunction $\Psi(t,x)$. The WKB analysis uses the ansatz $\Psi(t,x)=A(t,x)\exp(iS(t,x)/\hbar)$, where $A(t,x)$ and $S(t, x)$ are the amplitude and phase respectively. Ignoring $O(\hbar^2)$ terms, this ansatz results in two independent PDEs: the eikonal equation for $S$, which is a Hamilton-Jacobi PDE in Eq.~\eqref{H-J} with Hamiltonian $H(x, \nabla S)=(1/2)|\nabla S|^2+V(x)$, and the transport equation for $|A|^2$ . These two equations can be deduced from the moment-closure of  the Liouville equation
\begin{eqnarray} \label{Psi}
\partial_t w+ \nabla_p H \cdot \nabla_x w - \nabla_x H \cdot \nabla_p w=0 
\end{eqnarray} with initial data
\begin{equation}\label{Psi-IC}
    w(0,x,p)=|A_0(x)|^2 \delta(p-\nabla S_0(x)),
\end{equation}
with mono-kinetic ansatz $w(t,x,p)=|A(t,x)|^2\delta(p-\nabla_x S(t,x))$, but are not valid beyond caustics since $\delta(p-\nabla_x S(t,x))$ is not well-defined when $\nabla_x S(t,x)=u(t,x)$ becomes discontinuous. But EQ.~\eqref{Psi} is valid globally in time \cite{LP1993, GPPM}, since it unfolds the caustics in the phase space.  The problem defined by Eqs.~\eqref{Psi}-\eqref{Psi-IC} can be solved by first solving 
\begin{equation} \label{Psi-1}
\mathcal{L} \tilde{\phi}=0 \,,\qquad
    \tilde{\phi}(0,x,p)=|A_0(x)|^2 \,,
\end{equation}
and \begin{equation} \label{Psi-2}
\mathcal{L} \phi=0 \,, \quad
    \phi(0,x,p)= p-\nabla S_0(x)\,,
\end{equation}
independently and then $w$ can be obtained by $w=\tilde{\phi} \delta(\phi)$ \cite{JLOT1}. Here $\phi$ is {\it exactly} the level set function defined by \eqref{LS-Liouville} and \eqref{LS-ic}, while 
$\psi$ defined \eqref{Liouville-delta} and \eqref{Liou-I} is
exactly $w$ in the case of $A_0(x)=1$ for $M=1$.
Therefore different choices of $G$ can be used to recover moments of the Wigner function. These moments provide the classical limits of the observables of the original Schr\"odinger
equation. For instance, the zeroth, first and second moments of the Wigner function $w$ are
\begin{equation}\label{physical-observables}
   \rho(t,x)=\int w \, dp\,,\quad \rho(t,x) u(t,x) =\int pw \, dp\,,\quad 
  \frac{1}{2}\rho(t,x)u^2(t,x) =\int \frac{|p|^2}{2}  w \, dp\quad
\end{equation}
which correspond to using $G(p)=1, p, |p|^2$ respectively.  Here $\rho(t,x)$ is the classical limit to the position density $|\Psi|^2$, $\rho u$ is the momentum, or the classical limit to the current density $\hbar\, \text{Im} (\overline{\Psi} \nabla{\Psi})$, while the second moment is the classical limit of the kinetic energy $(\hbar^2/2)|\nabla \Psi|^2$. The total energy $E(t,x)$ can be recovered by combining the second and zeroth moments to choose $G(p)=(1/2)|p|^2+V(x)$ to obtain $E(t,x)=\int (|p|^2/2+V(x))w dp$. Likewise, the moments of $\psi$, if similarly defined as in \eqref{physical-observables}, give the same physical observables in the special  case of $A_0(x)=1.$\\


If one begins with more general $M>1$ initial data
\begin{equation}\label{Psi-IC-k}
    w(0,x,p)=\frac{1}{M}\sum_{k=1}^M |A^{[k]}_0(x)|^2 \delta(p-\nabla S^{[k]}_0(x))
\end{equation}
then $w(t,x,p)=\frac{1}{M}\sum_{k=1}^M w^{[k]}(t,x,p)$ where $w^{[k]}(t,x,p)$ solves  equation \eqref{Psi} with initial condition $w^{k]}(0,x,p)=|A^{[k]}_0(x)|^2 \delta(p-\nabla S^{[k]}_0(x))$, thus consequently, due to the linear superposition principle, the observables 
\begin{equation}\label{physical-observables-k}
   \int w \,dp=\frac{1}{M}\sum_{k=1}^M \rho^{[k]}(t,x) \,,\quad \int pw \, dp=\frac{1}{M} \sum_{k=1}^M\rho^{[k]}(t,x) u^{[k]}(t,x) \,,\quad 
  \int \left(\frac{p^2}{2} + V(x) \right) w \, dp=\frac{1}{M}\sum_{k=1}^M E^{[k]}(t,x)
  \quad 
\end{equation}
which are the ensemble average of the physical observables (position density, momentum and kinetic energy) of each individual observables. \\

We can also consider examples for general symmetric hyperbolic systems. The high frequency limit of general symmetric hyperbolic systems--for which geometric optics is one example-- possesses a strong similarity with
the semi-classical limit of the Schr\"odinger equation via the Wigner transform \cite{RPK}. \\

In the example of acoustic waves, the WKB approximation of such systems gives a Hamilton-Jacobi equation in Eq.~\eqref{H-J} for the phase $S$ with Hamiltonian 
\begin{equation}\label{H-GO}
    H(x,p)=\frac{1}{\sigma(x)\tau(x)}|p|
\end{equation}
where $\sigma(x)$ and $\tau(x)$ are the density and compressibility of the wave. The WKB approximation also provides a transport equation for the amplitude $|A|$. The corresponding Wigner function for this approximation also satisfies the Liouville equation in Eq.~\eqref{Psi} with the Hamiltonian in Eq.~\eqref{H-GO} \cite{GPPM}. Then its zeroth and first moments are
\begin{equation}\label{physical-observables2}
   \rho(t,x)=\int w \, dp\,,\quad \rho(t,x) u(t,x) =\int pw \, dp,
\end{equation}
which can be used to obtain the intensity $|A|=\sqrt{\rho}$ and the velocity $u=\rho(t,x) u(t,x)/\rho(t,x)=\nabla S$, where $\nabla S$ is also known as the slowness vector. \\

The second example is Maxwell's equation in an isotropic medium. One can similarly retrieve the moments of its high frequency limit via a Wigner analysis. The observables recovered include the energy and the Poynting vector \cite{RPK}. The third example is elastic waves, where a similar analysis can allow one to retrieve the kinetic energy, strain energy and energy flux of the elastic wave \cite{RPK}. The fourth example is the Dirac equation,  the relativistic version of the Schr\"odinger equation that describes very fast electrons in an electromagnetic field. In the semiclassical limit, it gives rise to a similar Liouville equation, like in geometric optics, but with a Lorentz term due to relativistic effects \cite{GPPM}.
The Wigner function can similarly recover physical observables such as the density of positrons and electrons. \cite{GPPM}.\\

We can also explain the presence of the Jacobian weights $\mathcal{J}^{[k]}_{\gamma}$ in Eq.~\eqref{G-HJ}. The Wigner function $w$ for $M=1$ has the general form $w(t, x, p)=a(t, x, p)\delta(\phi(t, x, p))$, where in general $\phi(t, x, p) \neq p-u(t, x)$ for $t>0$. By definition
$\delta(\phi(t,x,p))=\sum_{\gamma=1}^{J_k} \delta(p-u_{\gamma}(t,x))/\mathcal{J}_{\gamma}$, 
in the presence of multiple zeros of $\phi(t,x,p)$, corresponding to multi-valued solutions. 
Thus one can rewrite the Wigner function as $w(t, x, p)=\sum_{\gamma=1}^{J_k} \rho_{\gamma}(t, x, p)\delta(p-u_{\gamma}(t,x))$, where $\rho_{\gamma}(t, x, p)=a(t, x, p)/\mathcal{J}_{\gamma}$ \cite{GPPM}. This means $\rho_{\gamma}$ (with the Jacobian term contained in its definition) is the true density associated with each multi-valued solution, since $\rho(t,x)=\int w dp=\sum_{\gamma=1}^{J_k} \rho_{\gamma}$. 
The extension to $M>1$ and other moments is straightforward. The fact that this non-trivial Jacobian term (that is naturally captured by using $\psi(t,x,p)$) cannot be ignored in order to obtain the correct physical quantities is  important, an issue  we will return to at the end of Section~\ref{sec:quantumlevelset}.

\subsubsection{Classical error bounds} \label{sec:classicalerrorbounds}

Suppose one is interested in computing observables coming from solutions to Eq.~\eqref{forced-Burgers} with $M$ different initial data. Several approximations will be used in our computation.
(1) We solve instead the corresponding linear PDE problem in Eq.~\eqref{Liouville-delta} with initial data in Eq.~\eqref{Liou-I}, where the PDE will be approximated by
some classical finite difference or finite volume method on a $2d$-dimensional phase space mesh on $[0,1]^{2d}$ with mesh size $h=1/N$ and forward Euler method in time with time step $\Delta t$;  (2) The delta function $\delta$ in the initial condition will be approximated by a discrete delta function $\delta_{\omega}$ defined in Eq.~\eqref{eq:deltaomega}, with smoothing parameter $\omega=mh$, for some small integer $m>0$; (3) The observable defined by the integral in Eq.~\eqref{G-HJ} will be approximated by a quadrature rule in Eq.~\eqref{eq:discreteensemblepsi1}.
Taking all three sources of error into account, we obtain the following error bound in estimating the ensemble average. 
\begin{lemma} \label{thm:classicalerror}   The ensemble average $\langle G(t_n,\vect{j}/N) \rangle$ can be estimated by $\langle G^{\omega}_{n,\vect{j}}\rangle$ with error 
\begin{align}
    \epsilon_{CL} \equiv |\langle G(t_n,\vect{j}/N) \rangle-\langle G^{\omega}_{n,\vect{j}}\rangle| \leq C(\omega + dh/\omega^2)
\end{align}
where $C>0$ is independent of $h, \omega$. By choosing $\omega=(dh)^{1/3}$ one gets
\begin{align}
    \epsilon_{CL} \equiv |\langle G(t_n,\vect{j}/N) \rangle-\langle G^{\omega}_{n,\vect{j}}\rangle| \leq C (dh)^{1/3}.
\end{align}
\end{lemma}
\begin{proof}
See Appendix~\ref{app:classicalerrorproof}.
\end{proof}

{\bf Remark:} By using a second order finite difference scheme one can improve the above error
to $O(\sqrt{dh})$ with $\omega=\sqrt{dh}$. Higher order finite difference approximations will further reduce $\omega$.\\

This bound can also be used to estimate the total computational cost for a classical computer to solve the $(2d+1)$-dimensional linear) Liouville PDE. 

\begin{lemma} \label{lem:liouville}
If one uses $N_{\text{CL}}$ points in each dimension of the $2d$ phase space (hence $h_{\text{CL}}=1/N_{\text{CL}}$) to approximate Eq.~\eqref{Liouville-delta} to time $T$ by a first order finite difference or finite volume scheme, then $N_{\text{CL}}=O(d/\epsilon_{\text{CL}}^3)$ for $\omega=(dh_{\text{CL}})^{1/3}$, while the overall computational cost will be of $O(d^{2d+3}T(1/\epsilon_{\text{CL}})^{6d+3})$.   
\end{lemma}
\begin{proof}
If one uses $N_{\text{CL}}$ points in each dimension of the $2d$ phase space (hence $h_{\text{CL}}=1/N_{\text{CL}}$) to approximate \eqref{Liouville-delta} to time $T$ by a first order finite difference or finite volume scheme, the CFL condition will require $\Delta t_{\text{CL}}=O(h_{\text{CL}}/(2d))$. To reach an error of $O(\epsilon_{\text{CL}})$ one needs $N_{\text{CL}}=O(d/(\omega^2(O(\epsilon_{\text{CL}})-O(\omega))))$. (If one chooses $\omega=(dh_{\text{CL}})^{1/3}$ then $h_{\text{CL}}=O(\epsilon_{\text{CL}}^3/d)$, $N_{\text{CL}}=O(d^{1/3}/(h_{\text{CL}}^{2/3}\epsilon_{\text{CL}}))=O(d/\epsilon_{\text{CL}}^3)$.) The overall computational cost will be of $O(d N_{t, \text{CL}} N_{\text{CL}}^{2d})=O(d^2T N_{\text{CL}}^{2d+1})=O(d^{2d+3}T(1/\epsilon_{\text{CL}})^{6d+3})$.  
\end{proof}

\noindent{\bf Remark:} The above computational cost is {\it independent of} $M$. However, the cost now is much greater than that in Lemma~\ref{lem:HJ}, since one now solves a higher dimensional PDE, and with much smaller ($\omega$-dependent) mesh size to numerically resolve the discrete delta function $\delta_\omega$. \\



\subsection{The quantum algorithm to approximate physical observables} \label{sec:quantumlevelset}

Our aim is to approximate the observable $\langle G(t_n, \vect{j}/N)\rangle $ by devising the corresponding SLEP quantum algorithm. Consider $\langle G^{\omega}_{n,\vect{j}}\rangle$ defined by
\begin{align} \label{eq:discreteensemblepsi}
    \langle G^{\omega}_{n,\vect{j}}\rangle \equiv \frac{1}{N^d}\sum_{\vect{l}}^N G_{\vect{l}} \psi_{n,\vect{j}, \vect{l}}.
    \end{align}
 The initial condition in its discretised form in Eq.~\eqref{eq:psiinitialtext} can be represented by a quantum state
\begin{align}
    |\psi_0\rangle=\frac{1}{N_{\psi_{0}}}\sum_{\vect{j}^{N}} \sum_{\vect{l}}^N \psi_{0, \vect{j}, \vect{l}} |\vect{j}\rangle |\vect{l}\rangle |n=0\rangle
\end{align}
 and the normalisation is given by $N_{\psi_{0}}=\sqrt{\sum_{\vect{j}}^{N}\sum_{\vect{l}}^N |\psi_{0, \vect{j}, \vect{l}}|^2}$. In this paper, we assume access to a unitary operation $U_{initial}$ that prepares $U_{initial}|0\rangle=|\psi_{0}\rangle$. We also define the state
\begin{align}
    |G_{n, \vect{j}}\rangle \equiv \frac{1}{N_G}\sum_{\vect{l}}^{N}G^*_{\vect{l}}|\vect{l}\rangle |\vect{j}\rangle |n\rangle
\end{align}
where we note that here we do \textit{not} sum over the time step index $n$ or the spatial index $\vect{j}$.
The normalisation is $N_G=\sqrt{\sum_{\vect{l}}^N |G^*_{\vect{l}}|^2}$. \\

Using the states $|G_{n, \vect{j}, \vect{l}}\rangle$ and applying matrix inversion algorithms on $|\psi_{0}\rangle$, the observable can be recovered. Given the density matrix $\mathcal{G} \equiv |G_{n, \vect{j}}\rangle \langle G_{n,\vect{j}}|$, one observes that the expectation value
\begin{align} \label{eq:upsilondefinition}
   & \Upsilon \equiv \langle \psi_{0}|(\mathcal{M}^{-1})^{\dagger}\mathcal{G}\mathcal{M}^{-1}|\psi_{0}\rangle=\frac{1}{N^2_{\psi_{0}}}\sum_{\vect{j}', \vect{l}'}(\mathcal{M}^{-1}\psi_{0, \vect{j}', \vect{l}'})^T \mathcal{G}(\mathcal{M}^{-1}\psi_{0, \vect{j}', \vect{l}'}) \nonumber \\
   &=\frac{1}{N^2_{\psi_{0}}N_G^2}\Big|\sum_{\vect{l}}G_{\vect{l}}\psi_{n,\vect{j}, \vect{l}}\Big|^2=\frac{N^{2d}}{N^2_{\psi_{0}}N_G^2}\langle G^{\omega}_{n,\vect{j}}\rangle^2
\end{align}
where the last equality comes from Eq.~\eqref{eq:discreteensemblepsi}. Our aim is to use a quantum algorithm to extract $\Upsilon$, from which we can approximate the ensemble average 
\begin{align} \label{eq:gequation}
    \langle G(t_n,\vect{j}/N)\rangle \approx \langle G^{\omega}_{n,\vect{j}}\rangle \equiv \frac{1}{N^d}\sum_{\vect{l}=1}^N G_{\vect{l}} \psi_{n,\vect{j}, \vect{l}}=\frac{1}{N^d}N_{\psi_{0}}N_G|\sqrt{\Upsilon}|=n_{\psi_0}n_G|\sqrt{\Upsilon}|,
\end{align}
where we define $n_G \equiv N_G/N^{d/2}$ and $n_{\psi_0} \equiv N_{\psi_{0}}/N^{d/2}$. Then we have the following lemma.
\begin{lemma} \label{lem:normalisations}
The constant $n_G \equiv N_G/N^{d/2}=O(1)$ and depending on the application, the range of $n_{\psi_0}$ lies in $O(1) \leq n_{\psi_0} \equiv N_{\psi_{0}}/N^{d/2} \leq O(N^{d/2})$. Different $n_{\psi_0}$ corresponds to different initial data. If we assume the initial data has support in a box of size $\beta$, then $n_{\psi_0}=O((\beta \sqrt{N})^d)$. 
\end{lemma}
\begin{proof}

See Appendix~\ref{sec:normalisations}.
\end{proof}
\noindent \textbf{Remark:} The definition $n_{\psi_0}=N_{\psi_0}/N^{d/2}$ can be an underestimate in situations where the $N^d$ term in Eq.~\eqref{eq:gequation} overestimates the number of non-zero factors in the summation. For instance, for a point source where $\beta \sim 1/N$, $n_{\psi_0}<1$. To deal with similar scenarios, more information about the problem is required to improve the quadrature rule approximation in Eq.~\eqref{eq:gequation}. \\

 However, since the quantum device can only output an approximation $\tilde{\Upsilon}$ of $\Upsilon$ to finite precision, we can only access the estimate
\begin{align}
    \langle \tilde{G}^{\omega}_{n,\vect{j}}\rangle \equiv n_{\psi_0}n_G|\sqrt{\tilde{\Upsilon}}|
\end{align}
where $|\sqrt{\tilde{\Upsilon}}-\sqrt{\Upsilon}|\leq \epsilon_G$. Then we have the following contributions to the total error in estimating $\langle G(t_n,\vect{j}/N)) \rangle$. \\

\begin{lemma} \label{lem:3errors} The error $|\langle G(t_n,\vect{j}/N) \rangle-\langle \tilde{G}^{\omega}_{n,\vect{j}}\rangle|$ can be broken into two independent sources of error
\begin{align}
    & |\langle G(t_n,\vect{j}/N) \rangle-\langle \tilde{G}^{\omega}_{n,\vect{j}}\rangle| \leq |\langle G(t_n,\vect{j}/N) \rangle-\langle G^{\omega}_{n,\vect{j}}\rangle|+|\langle G^{\omega}_{n,\vect{j}}\rangle-\langle \tilde{G}^{\omega}_{n,\vect{j}}\rangle| \nonumber \\
    & =\epsilon_{\text{CL}}+\epsilon_{Q}\leq \epsilon.
\end{align}
Here $\epsilon_{\text{CL}}$ comes from approximating the solutions of the discretised linear PDE (Liouville equation) corresponding to the original Liouville equation, as estimated in Lemma~\ref{thm:classicalerror}. The contributions to the quantum sources of error is denoted $\epsilon_Q$. If $|\sqrt{\tilde{\Upsilon}}-\sqrt{\Upsilon}| \leq \epsilon_G$ and we impose $\epsilon_{\text{CL}}\sim\epsilon_Q$, then it is sufficient to choose $\epsilon_G \sim \epsilon/(n_Gn_{\psi_0}) \sim \epsilon_{\text{CL}}/(n_Gn_{\psi_0})$ where $n_G=O(1)$, $n_{\psi_0}\geq O(1)$.
\end{lemma}

\begin{proof}
The quantum error can be written as
\begin{align} \label{eq:epsilonQG}
    &\epsilon_Q=|\langle G^{\omega}_{n,\vect{j}}\rangle-\langle \tilde{G}^{\omega}_{n,\vect{j}}\rangle|=n_Gn_{\psi_0}(|\sqrt{\Upsilon}|-|\sqrt{\tilde{\Upsilon}}|)\leq \epsilon_Gn_Gn_{\psi_0}.
\end{align}
If $\epsilon_{\text{CL}} \sim\epsilon_Q$, it's  sufficient to choose $\epsilon_Gn_Gn_{\psi_0} \sim \epsilon$ and the results follow.
 \end{proof} 


Before stating our main theorems, we prove the following lemma.

\begin{lemma} \label{lem:bslep2}
 A quantum algorithm can be constructed that takes the following inputs: (i) sparse access \\
 $(s, \|\mathcal{M}\|_{max}, O_{M}, O_{F})$ to a $2^m \times 2^m$ invertible Hermitian matrix $\mathcal{M}$ such that $\|\mathcal{M}\| \leq s\|\mathcal{M}_{max}\|$ (ii) $m$-qubit unitary $L(t_n, \vect{j}/n)$ where $L(t_n, \vect{j}/N)|0\rangle=|G_{n, \vect{j}}\rangle$ (iii) an accuracy $\epsilon' \in [1/2^m, 1]$ and (iv) $m$-qubit unitary black box $U_{initial}$ where $U_{initial}|0\rangle=|\psi_{0}\rangle$. The algorithm then returns with probability at least $2/3$ an $\epsilon'$-additive approximation to $\Upsilon \equiv \langle \psi_{0}|(\mathcal{M}^{-1})^{\dagger}\mathcal{G}\mathcal{M}^{-1}|\psi_{0}\rangle$ where $\mathcal{G} \equiv |G_{n, \vect{j}}\rangle \langle G_{n,\vect{j}}|$. This algorithm makes $\mathcal{O}(\kappa^2/(\|\mathcal{M}\|\epsilon'))$ queries to $U_{\mathcal{G}}$ and $U_{initial}$,  $\mathcal{O}(s\|\mathcal{M}\|_{max}\kappa^3\log(\kappa^2/(\|\mathcal{M}\|\epsilon'))/(\|\mathcal{M}\|\epsilon'))$ queries to sparse oracles for $\mathcal{M}$ and $\mathcal{O}(\kappa^2(m+ s\|\mathcal{M}\|_{max}\kappa\log(\kappa^2/(\|\mathcal{M}\|\epsilon'))))/(\|\mathcal{M}\|\epsilon'))$ additional $2$-qubit gates.
\end{lemma}

\begin{proof}
See Appendix~\ref{app:A} for details of the proof. 
\end{proof}
 To estimate the observable in Definition~\ref{def:levelsetobservable}, a quantum algorithm is employed to estimate $\Upsilon$ (see Appendix~\ref{app:algorithmsteps} for a summary of the basic steps) with the following quantum query and gate complexities.

\begin{theorem} \label{thm:qquery} A quantum algorithm that takes sparse access $(s=O(d), \|\mathcal{M}\|_{max}=O(1), O_{M}, O_{F})$ to $\mathcal{M}$, where $\|\mathcal{M}\|= O(1)$, and access to the unitaries $L(t_n,\vect{j}/N)$, where $L(t_n,\vect{j}/N)|0\rangle=|G_{n, \vect{j}}\rangle$ and $U_{initial}$, where $U_{initial}|0\rangle=|\psi_{0}\rangle$, is able to estimate the ensemble average $\langle G(T, x)\rangle$ at time $T=t_n$, with $M$ initial data, to precision $\epsilon$ with an upper bound on the query complexity $\mathcal{Q}$ 
\begin{align} \label{eq:qquery}
  \mathcal{Q}= \mathcal{O}\left(\frac{n_{\psi_0}^2d^7T^3}{\epsilon^{10}}\log \left(\frac{n_{\psi_0}^2d^4T^2}{\epsilon^{7}}\right) \right)
\end{align}
and the same order of additional $2$-qubit gates, where we suppress all $O(1)$ terms except $T=O(1)$ and $n_{\psi_0}\geq O(1)$. 
\end{theorem}

\begin{proof}
See Appendix~\ref{app:theoremhj} for details of the proof. 
\end{proof}

The complexity in Eq.~\eqref{eq:qquery} has several notable features. Firstly, we see that $\mathcal{Q}$ is {\it independent} of $M$, whereas the classical cost $\mathcal{C}= O(MTd^{d+4}(1/\epsilon)^{d+1})$ to solve the Hamilton-Jacobi equation from Lemma~\ref{lem:HJ} is linear in $M$. Secondly, the classical cost $\mathcal{C}$ contains exponential terms in dimension like $d^d$ and $(1/\epsilon)^d$, which is absent in applications where $n_{\psi_0}$ does not grow so quickly. Thirdly, this quantum algorithm is valid for all $T$, whereas some previous quantum algorithms for nonlinear PDEs the linear approximations are no longer valid for larger $T$, for {\it any}  nonlinearity (e.g \cite{lloyd2020quantum}). This is particularly important in the case of strong nonlinearity where it maybe interesting to study the large $T$ behavior where the nonlinearity cannot be well-approximated by low-order polynomials. However, there is no quantum advantage in $T$ for $T>1$ since $\mathcal{C}$ is linear in $T$ while the quantum algorithm depends on $T^3$. We note that this absence of quantum advantage in $T$ also holds for existing quantum algorithms for both linear and nonlinear ODEs and PDEs. Thus, there is potential quantum advantage on the three parameters $M, d, \epsilon$. \\
\begin{corollary} \label{cor:quantumhj1}  Let $\mathcal{C}$ be the cost to compute an observable from a $(d+1)$-dimensional Hamilton-Jacobi equation with purely classical methods and $\mathcal{Q}$ be the cost of our quantum algorithm. We say there is a quantum advantage in estimating the observables when $\mathcal{Q}=\littleo(\mathcal{C})$.  To attain a quantum advantage it is then sufficient for the following condition to hold
\begin{align}
\littleo\left(\frac{Md^{d-4}}{n_{\psi_0}^2T^2}\left(\frac{1}{\epsilon}\right)^{d-9}\right)=\tilde{O}(1)
\end{align}
where $\tilde{O}$ suppresses all logarithmic terms in $d, 1/\epsilon, T$. 
\end{corollary}
\begin{proof}
The requirement $\mathcal{Q}=\littleo(\mathcal{C})$ implies that $\mathcal{C}$ increases much more quickly than $\mathcal{Q}$, i.e., in the asymptotic limit, the ratio $\mathcal{C}/\mathcal{Q}$ goes to infinity. This can be equivalently expressed as $\littleo(\mathcal{C}/\mathcal{Q})=1$.  For instance, this can include both polynomial and exponential advantages for the quantum algorithm. From Lemma~\ref{lem:normalisations}, the constant $n_{\psi_0}$ cannot be greater than $N^{d/2}$. This means that the extra logarithmic terms in the upper bound of $\mathcal{Q}$ is at most linear in $d$ multiplied by logarithmic factors in $d, 1/\epsilon, T$. Using the tight upper bound for $\mathcal{Q}$ from Theorem~\ref{thm:qquery} and $\mathcal{C}= O(MTd^{d+4}(1/\epsilon)^{d+1})$ from Lemma~\ref{lem:HJ}, the result follows. 
\end{proof}
In this case, it is clear that even when $M=1$, there is potential quantum advantage in parameters $d$ and $\epsilon$. Depending on the initial data, we have $n_{\psi_0}^2=O(N^b)$, where $b$ ranges from $0$ to $d$ from Lemma~\ref{lem:normalisations}. Ignoring $T=O(1)$ factors, using Lemma \ref{lem:liouville}, we see that Corollary~\ref{cor:quantumhj1} requires $\tilde{O}(d^{d-4-b}/\epsilon^{d-9-3b})=\littleo(1)$. Thus if $0\leq b<d/3-3$ and $d$ is large, one can obtain quantum advantage in both $d$ and $\epsilon$. For instance, in the best case scenario where $b=O(1)$, then quantum advantage for $d$ and $\epsilon$ is always possible for high enough $d$. However, in the worst-case scenario where $b=d$, we see that no advantage is possible and the quantum algorithm actually performs worse than the classical algorithm when $M=1$. \\

 Different $b$ corresponds to different physical scenarios. From Lemma~\ref{lem:normalisations} we see that if the initial data has support in a box of size $\beta$, then $n_{\psi_0}=O((\beta \sqrt{N})^d)=O(N^{b/2})$.  For a Gaussian source of width $\beta \sim 1/\sqrt{N}$ for instance, one has $n_{\psi_{0}}=O(1)$ so $b=0$. For $b=k$ a constant, one needs $\beta=O(1/N^{1/2-k/(2d)})$. The requirement $b<d/3-3$ for quantum advantage implies a support of the initial data with the upper bound $\beta<O(1/N^{1/3+3/(2d)})$. \\  

Often one needs to solve a PDE with many different initial data, such as those in Monte-Carlo simulation, uncertainty quantification or Bayesian inference-based inverse problems in uncertainty quantification, which demands large $M$. Since $\mathcal{Q}$ is independent of $M$ and $\mathcal{C}$ is linear in $M$, it is always (theoretically) possible to have large enough $M$ for a quantum advantage. The reason for this $M$-independence lies in the fact that, in our formulation, the initial  data for $\psi$ is written as a single sum over $M$ initial data. This means that even if we consider state preparation costs for $|\psi_0\rangle$, it would be still independent of $M$ since the sum over $M$ terms lies within a single amplitude. In addition, the $M$ different initial data can be arbitrarily distributed and there are no constraints on neighbouring initial points to be close together to form a good mesh for high accuracy in solving the PDE. \\

It is important to point out that, since the Hamilton-Jacobi equation is no longer valid beyond the formation of caustics, the comparison in Corollary~\ref{cor:quantumhj1} only makes sense when the solution to the Hamilton-Jacobi equation is {\it smooth}. Beyond the time of caustics, one should compare the quantum cost with the classical cost of solving the Liouville equation (or other more sophisticated classical algorithms laid out in at the beginning of section \ref{sec:H-J}).
As shown in Corollary~\ref{cost-Liou}, our quantum advantage is much bigger when compared with the classical solver for the Liouville equation, except now there is no advantage in $M$. 

\begin{corollary}\label{cost-Liou}
If $\mathcal{C}_L$ is the cost to compute an observable of a $(2d+1)$-dimensional Liouville equation using  purely classical methods and $\mathcal{Q}$ is the cost with a quantum algorithm, then $\mathcal{Q}=\littleo(\mathcal{C}_L)$ when 
\begin{align}
 \littleo\left(\frac{d^{2d-5}}{n^2_{\psi_0}T^2}\left(\frac{1}{\epsilon}\right)^{6d-7}\right)>\tilde{O}(1).
\end{align}
\end{corollary}
\begin{proof}
From Lemma~\ref{lem:liouville},  $\mathcal{C}_L= O(Td^{2d+3}(1/\epsilon)^{6d+3})$, which we note is also independent of $M$. \\
Since $\mathcal{Q}= \tilde{O}(n^2_{\psi_0}d^8T^3(1/\epsilon)^{10})$ where $\tilde{O}$ suppresses all logarithmic terms in $d, T, 1/\epsilon$, we have our result.
\end{proof}
In this case, the classical cost is also independent of $M$, but has a large overhead in $d$ and $1/\epsilon$. Since $b$ ranges from $0$ to $d$ when $n^2_{\psi_0}=O(N^b)$, we see that $\mathcal{C}_L>\mathcal{Q}$ if $\tilde{O}(d^{2d-5-b}/\epsilon^{6d-7-3b})>O(1)$, ignoring $T\sim O(1)$ terms. This means that even in the worst-case scenario where $s=b$ (when the initial condition has the largest support $\beta=O(1)$), there is quantum advantage in both $d$ and $\epsilon$ for large $d$. \\


We remark that it is also always possible to have a quantum subroutine for the PDE problem (QLSP), where we use the HHL or CKS algorithms to prepare the quantum state 
\begin{align}
    |\psi\rangle \equiv \frac{1}{N_{\psi}}\sum_{n=0}^{N_t-1}\sum_{\vect{j}}^{N} \sum_{\vect{l}}^N \psi_{n, \vect{j}, \vect{l}} |\vect{j}\rangle |\vect{l}\rangle |n\rangle
\end{align}
with $\psi_{n, \vect{j}, \vect{l}}$ being the solutions to the discretised version of the Liouville equation in Eq.~\eqref{Liouville-delta} and $N_{\psi}$ is the normalisation constant. This encoding of the solutions of the PDE, which we can call the \textit{level-set encoding}, differs from the amplitude encoded states usually considered in the literature, where the amplitudes of the quantum state are directly proportional to the solutions of the \textit{original} PDE one wants to solve. Although the amplitudes of this state $\psi_{n, \vect{j}, \vect{l}}$ likewise contains all the information required for the solution, this encoding is in fact much more appropriate in this instance to allow observables to be extracted. Given this quantum subroutine, one can for instance apply the quantum swap-test between states $|\psi\rangle$ and $|G\rangle$ to recover $\Upsilon$ (although with worse query complexity compared to the current algorithm in Theorem~\ref{thm:qquery}). However, any similar approach to recover physical observables from quantum states whose amplitudes are directly proportional to the original solutions $u$ themselves, requires one to inject the Jacobian factors in Definition~\ref{def:levelsetobservable} explicitly. However, these factors are not generally a priori known since they depend on the solutions of the PDE and also cannot be ignored, since they can have real physical consequences as discussed at the end of Section~\ref{sec:physicalinterpretation}. The benefit of a level-set encoded state is that these Jacobian factors are automatically taken care of. This highlights the importance of choosing the right encoding of the PDE solutions, even on the level of a quantum subroutine.  



\section{Solving scalar hyperbolic PDEs} \label{sec:hyperbolic}
 We begin with a $(d+1)$-dimensional scalar nonlinear hyperbolic PDE 
\begin{eqnarray}\label{hyp-PDE2}
  &\partial_t u^{[k]} + F(u^{[k]}) \cdot \nabla_x u^{[k]} + Q(x,u^{[k]})=0, \quad u\in \mathbb{R}, \\
  \label{hyp-IC}
   & u^{[k]}(0,x)=u^{[k]}_0(x), \quad k=1,...,M
\end{eqnarray}
  subject to $M$ different initial data, where $x \in \mathbb{R}^d$. We can also employ the level set formalism for this PDE and introduce a level set function $\phi(t,x,p)$ in $(d+1)+1=d+2$ dimensions, where $p \in \mathbb{R}^1$. Its zero level set is the solution
  $u$:
  \begin{equation}
        \phi^{[k]}(t,x,p)=0  \quad {\text{at}} \quad p=u^{[k]}(t,x).
  \end{equation}
  Then $\phi$ satisfies \cite{JO03}
  \begin{eqnarray}
      & \partial_t \phi^{[k]} + F(p) \cdot \nabla_x \phi^{[k]} - Q(x,p) \partial_p \phi^{[k]} = 0,\\
      & \phi^{[k]}(0,x,p)=p-u^{[k]}_0(x).
  \end{eqnarray}
Like for the Hamilton-Jacobi PDEs, we can similarly define a function $\psi$ by the following problem
\begin{eqnarray}
      & \partial_t \psi + F(p) \cdot \nabla_x \psi - Q(x,p) \partial_p \psi = 0,\\
      & \psi(0,x,p)=\frac{1}{M}\sum_{k=1}^M \delta(p-u^{[k]}_0(x)),
  \end{eqnarray}
and one has
\begin{align}
    \psi(t, x, p)=\frac{1}{M}\sum_{k=1}^M \delta(\phi^{[k]}(t,x,p)).
\end{align}
The only difference compared to the Hamilton-Jacobi equation is that now $p \in \mathbb{R}^1$ instead of being a $d$-dimensional vector. This means the observable is now an integral over $\mathbb{R}$:
\begin{align}
    \langle g(t,x)\rangle=\int_{\mathbb{R}}g(p)\psi(t,x,p)dp \approx \frac{1}{N}\sum_{l=1}^N g_l \psi_{n,\vect{j}, l} \equiv \langle g^{\omega}_{n,\vect{j}}\rangle
\end{align}
where after discretisation of the level set PDE $\psi(t, x, p) \rightarrow \psi_{n,\vect{j}, l}$, where  $l$ is a scalar, while $\vect{j}$ remains a vector. The initial condition can be encoded in the quantum state
\begin{align}
    |\psi'_{0}\rangle=\frac{1}{N'_{\psi_{0}}}\sum_{\vect{j}}^N \sum_l^N \psi_{0,\vect{j}, l}|\vect{j}\rangle |l\rangle|n=0\rangle
\end{align}
with the normalisation $N'_{\psi_{0}}=\sqrt{\sum_{\vect{|j|=1}}^N \sum_{l=1}^N |\psi_{0, \vect{j}, l}|^2}$
and we also define the state 
\begin{align}
    |g_{n, \vect{j}}\rangle \equiv \frac{1}{N_g}\sum_{l=1}^N g^*_l|l\rangle |\vect{j}\rangle |n\rangle
\end{align}
with the normalisation $N_g=\sqrt{\sum_l^N|g^*_l|^2}$. Defining $\Upsilon$ in the same way as Eq.~\eqref{eq:upsilondefinition}, with $\mathcal{G}\equiv |g_{n, \vect{j}}\rangle \langle g_{n,\vect{j}}|$, we can write
\begin{align}
    \langle g(t,x)\rangle \approx \frac{1}{N}N'_{\psi_{0}}N_g|\sqrt{\Upsilon}|=n'_{\psi_{0}}n_g|\sqrt{\Upsilon}|
\end{align}
where $n'_{\psi_{0}} \equiv N'_{\psi_{0}}/\sqrt{N}$, $n_g \equiv N_G/\sqrt{N}$. The normalisations are now the following.

\begin{lemma} \label{lem:normalisations2}
The constant $n_g \equiv N_G/\sqrt{N}= O(1)$ and the range of $n'_{\psi_{0}}$  lies in $O(1) \leq n'_{\psi_{0}} \equiv N'_{\psi_{0}}/\sqrt{N} \leq O(N^{d/2})$. If we assume the initial data has support in a box of size $\beta$, then $n_{\psi_0}=O((\beta \sqrt{N})^{d/2})$.
\end{lemma}
\begin{proof}
The proof proceeds in a similar way to Lemma~\ref{lem:normalisations}. See Appendix~\ref{sec:normalisations}.
\end{proof}

The quantum query and gate complexities are equivalent to the Hamilton-Jacobi case up to constants and the proof is identical.

\begin{theorem} \label{thm:qqueryhyperbolic} A quantum algorithm that takes sparse access $(s=O(d), \|\mathcal{M}\|_{max}=O(1), O_{\mathcal{M}}, O_{F})$ to $\mathcal{M}$ with $\|\mathcal{M}\|=O(1)$ and access to the unitaries $l(t_n,\vect{j}/N)$ where $l(t_n, \vect{j}/N)|0\rangle=|g_{n, \vect{j}}\rangle$, and $U_{initial}$ where $U_{initial}|0\rangle=|\psi'_{0}\rangle$, is able to estimate the ensemble average $\langle g(T, \vect{j}/N)\rangle$ at time $T$, with $M$ initial data, to precision $\epsilon$ with an upper bound on the query complexity $\mathcal{Q}$ 
\begin{align} \label{eq:qquery-1}
  \mathcal{Q}=\mathcal{O} \left(\frac{(n'_{\psi_{0}})^2 d^7T^3}{\epsilon^{10}}\log \left(\frac{(n'_{\psi_{0}})^2 d^4T^2}{\epsilon^{7}}\right) \right)
\end{align}
and the same order of additional $2$-qubit gates, where we suppress all $O(1)$ terms except $T=O(1)$ and $n'_{\psi_{0}}\geq O(1)$. 
\end{theorem}

The classical cost for the scalar hyperbolic PDE is just one factor of $d$ smaller than for the Hamilton-Jacobi PDE we considered. 

\begin{corollary} \label{cor:hyperbolic} Let $\mathcal{C}$ be the cost to compute an observable from a $(d+1)$-dimensional scalar hyperbolic equation with purely classical methods, and $\mathcal{Q}$ be the cost to solve the identical problem on a quantum algorithm. We say there is a quantum advantage in estimating physical observables when $\mathcal{Q}=\littleo({\mathcal{C}})$. This requires
\begin{align}
    \littleo \left(\frac{Md^{d-5}}{(n'_{\psi_{0}})^2T^2}\left(\frac{1}{\epsilon}\right)^{d-9}\right)>\tilde{O}(1)
\end{align}
where the $\tilde{O}$ notation suppresses all factors logarithmic in $1/\epsilon, d, T$.  
\end{corollary}
\begin{proof}
Identical to Corollary~\ref{cor:quantumhj1} except using $\mathcal{C}=O(MTd^{d+3}/\epsilon^{d+1})$ from Lemma~\ref{lem:classicalhyperbolic}. 
\end{proof}

The consequences are very similar to that of Hamilton-Jacobi equations. When $M=1$ and $(n'_{\psi_{0}})^2=O(N^b)$, for small enough $b$ where $0\leq b<d/3-3$ and large $d$ one can expect quantum advantage in $d$ and $\epsilon$. From Lemma~\ref{lem:normalisations2}, $n'_{\psi_{0}}=O((\beta \sqrt{N})^{d/2})=O(N^{b/2})$. A quantum advantage thus requires $\beta<O(1/N^{1/3+3/(2d)})$. Since $\mathcal{Q}$ is independent of $M$, the quantum algorithm will always be advantageous with respect to $M$ in the large $M$ limit. 

\section{Solving system of nonlinear ODEs} \label{sec:ODEs}

We aim to compute observables from the following system of $D$ nonlinear ODEs
\begin{align} \label{eq:upde2}
 &   \frac{d X^{[k]}(t)}{d t}=F( X^{[k]}(t)), \qquad X^{[k]}\in \mathbb{R}^D\,,\\
 & X^{[k]}(0)=X_0^{[k]}, \qquad  k=1,\cdots, M
\end{align}
 subject to $M$ different initial data, where the same ODE (i.e., same function $F(X)$ which can be nonlinear) is satisfied for each $k$. This can be viewed as a  system of $M$ non-interacting particles in $D$ dimensions, each with a trajectory described by $X^{[k]}(t)$. A traditional way to compute ensemble averages 
 \begin{equation} \label{Ens-ODE-1}
     A(t)=\frac{1}{M} \sum_{k=1}^M A(X^{[k]}(t))
\end{equation}
is to directly solve for the system of ODEs starting from $M$ different initial data $X_0^{[k]}$, and then carry out the summation in Eq.~\eqref{Ens-ODE-1}. \\

We note that the proposal below has appeared in \cite{dodin2021applications}, but no explicit details on the quantum costs were given, which we provide below. This is also very similar to the Koopman-von Neumann approach in \cite{joseph2020koopman}, except here we consider modelling the $M$ initial conditions as trajectories of $M$ non-interacting particles.\\


\subsection{Mapping nonlinear ODEs to a linear PDE} 

Assume $X^{[k]}(t)$ is the solution to Eq.~\eqref{eq:upde2}. Then one can define a function $\Phi(t, q):\mathbb{R}^+\times \mathbb{R}^D \rightarrow \mathbb{R}$  
\begin{align} \label{eq:phiaequation}
   \Phi(t, q) = \frac{1}{M} \sum_{k=1}^M   \delta(q-X^{[k]}(t)) 
\end{align}

One can easily check (Lemma~\ref{weaksoln-proof}) that $\Phi$ satisfies, in the weak sense, the {\it linear} $D+1$-dimensional PDE
\begin{align} \label{eq:pdephi}
  &  \frac{\partial \Phi(t, q)}{\partial t}+\nabla_q \cdot  [F (q) \Phi(t, q)]=0\nonumber \\
 &\Phi(0, q)=\frac{1}{M}\sum_{k=1}^M  \delta(q-X_0^{[k]}).
\end{align}

{\bf Remark:}  If $\nabla \cdot F=0$, which is the case for the Hamiltonian system,  then $J=1$, the dynamics in Eq.~\eqref{eq:upde2} is volume preserving, and equation (\ref{eq:pdephi})
becomes
\[
\frac{\partial \Phi(t, q)}{\partial t}+F(q) \cdot \nabla_q  \Phi(t, q)=0
\]

To justify the correctness of the method, we have the following lemma.

\begin{lemma}\label{weaksoln-proof}
 If $X^{[k]}(t)$ is the solution to Eq.~\eqref{eq:upde2},
 then Eq.~\eqref{eq:phiaequation} is the weak solution, in the sense of distribution, to the problem in Eq.~\eqref{eq:pdephi}.
  \end{lemma}
  \begin{proof}
  See Appendix~\ref{app:odejustification}. 
  \end{proof}
  
Suppose we are interested in solving \eqref{eq:upde2} with $M=N^D$ initial data, {\it equally spaced}
on a $D$-dimensional uniform mesh on $[0,1]^D$, with a uniform mesh size $h=1/N$. This will be the mesh to solve the
PDE problem in Eq.\eqref{eq:pdephi}. The discretisation process proceeds in the same way as Section~\ref{sec:HJ} and we can similarly convert the linear PDE in Eq.~\eqref{eq:pdephi} into a linear algebra problem. Let the discretisation of $\Phi(t,q)$ be
\begin{align}
    & \Phi (t, q) \rightarrow \Phi^{\omega}_{n, \vect{j}}
\end{align}
where $n$ denotes the time step from $t \rightarrow n \Delta t$ and $n=0,...,N_t$. The vector $\vect{j}=(j_1,...,j_D)$, denotes spatial grid positions with grid size $h=1/N$ where $q \rightarrow h \vect{j}$. The initial state at $n=0$ can be approximated by 
\begin{align} \label{eq:Phiinitialtext}
    \Phi^{\omega}_{0, \vect{j}}= \frac{1}{M}\sum_{k=1}^M  \delta_\omega(\vect{j}h-X^{[k]}_0)
\end{align}
where $\delta_{\omega}$ is the smoothed delta function, defined in Section~\ref{sec:levelset}, and $\omega$ is a smoothing parameter. The linear algebra problem to solve for $\Phi^{\omega}_{n, \vect{j}}$ is then identical to Eq.~\eqref{B11}, with the replacement $\psi_{n, \vect{j}, \vect{l}} \rightarrow \Phi_{n, \vect{j}}$ and $\mathcal{M} \rightarrow \mathcal{M}_{\text{ODE}}$. Here $\mathcal{M}_{\text{ODE}}$ is an $N_t N^D \times N_t N^D$ Hermitian matrix. See Appendix~\ref{app:discretisationode} for the precise form of $\mathcal{M}_{\text{ODE}}$ and details of the discretisation procedure. 
\begin{lemma} \label{lem:odesparsecondition} The condition number of $\mathcal{M}_{\text{ODE}}$ is $\kappa \leq O(D NT)$ and sparsity is $s=O(D)$. 
\end{lemma}
\begin{proof}
The proof is identical to the proof in Appendix~\ref{app:skappaproof} with the replacement $d \rightarrow D/2$. 
\end{proof}

\subsection{Observables of the ODE}
Any ensemble-averaged observable obtained from the solutions $X^{[k]}$ of the system of nonlinear ODEs can be written in the following form.  
\begin{definition} \label{def:observablea}
For any function $A(X)$, one can define the ensemble average
via
\begin{align} \label{Ens-G}
  \langle A (t) \rangle \equiv \int  A(q)  \Phi(t, q)\, dq=\frac{1}{M} \sum_{k=1}^M  A(X^{[k]}(t))
\end{align}
which we identify to be the observable associated with the ODE whose solutions are $\{X^{[k]}(t)\}$. 
\end{definition}
The second equality in Definition~\ref{def:observablea} is justified in Lemma~\ref{lem:AA}, where we prove that, if $\Phi(x,q)$ is the solution to 
(\ref{eq:pdephi}), and $\langle A(t)\rangle$ is defined as in Eq.~\eqref{Ens-G}, then
\begin{equation}
    \label{AA}
\langle A(t)\rangle=\frac{1}{M} \sum_{k=1}^M  A(X^{[k]}(t)),
\end{equation}
where $X^{[k]}$ is the solution to Eq.~\eqref{eq:upde2}. Since $\langle A(t) \rangle$ is an integral of the product of $A(q)$ and $\Phi(t,q)$, we can approximate this quantity with a very similar quantum algorithm that we used in Sections~\ref{sec:HJ} and \ref{sec:hyperbolic}.

\begin{lemma} \label{lem:AA}
 Assume $\Phi$ is the solution to Eq.~\eqref{eq:pdephi}, and $X^{[k]}$ is the solution to Eq.~\eqref{eq:upde},
 then Eq.~\eqref{AA} holds.
\end{lemma}
\begin{proof}
See Appendix~\ref{app:proofobservablea}. 
\end{proof}

We remark that here the observable $\langle A(t) \rangle$ is not a sum of terms $A(X^{[k]}(t))$ weighted by Jacobians, unlike in the case of nonlinear PDEs.  This key difference arises from the fact that during the evolution of the function $\Phi(t,q)$, the form of the delta function is preserved, i.e., from $t=0$ to $T$, the delta function terms in $\Phi(t,q)$ remains the form of  $ \delta(q-X^{[k]}(t))$, so there is no Jacobian term (or, in fact, the determinant of the Jacobian is always $1$ here).  This is unlike the case for Hamilton-Jacobi equation and hyperbolic PDEs, where for general $t$ the form of the delta function terms at $t=0$ are not preserved (from Eq.~\eqref{Liou-I} to Eq.~\eqref{eq:psisolution}). \\

Similarly to Section~\ref{sec:levelsetobservable} we can also define a new observable $A_O(t)$ that is the ratio of two observables where the second observable is the identity function. This can be considered as a normalised observable
\begin{align}
  A_O(t) \equiv \frac{\langle A (t)\rangle}{\mathbf{1}(t)}=\frac{\int  A(q)  \Phi(t, q)\, dq}{\int \Phi(t, q)\, dq}.  
\end{align}

\noindent{\bf Remark:}
Our algorithm does not recover individual $X^{[k]}(t)$,  only  the ensemble average. Moreover, we can
recover the ensemble average in any part of the domain, say $\mathcal{D}$, by
\[
A_P^{\mathcal{D}} (t) = \frac{1}{\text{vol}(\mathcal{D})}\int  A(q) \chi_{\mathcal{D}}(q) \Phi(t, q)\, dq
\]
where $\chi(\mathcal{D})$ is the characteristic function on $\mathcal{D}$ and $\text{vol}(\mathcal{D})$ is the volume of $\mathcal{D}$. We can do the simulation in some bounded domain in $q$, so the support of $\Phi$ remains inside the domain throughout the computational time. This is possible since the transport equation \eqref{eq:pdephi} has a finite propagation speed, if the support of $\Phi(0,q)$ is bounded. \\

The observable defined by the integral in Eq.~\eqref{Ens-G} can then be approximated by a quadrature method:
\begin{equation}\label{quadrature}
\langle A(t) \rangle \approx \langle A^{\omega}_n \rangle  = \frac{1}{N^D} \sum_{\vect{j}}^N A_{\vect{j}} \Phi^{\omega}_{n, \vect{j}},
\end{equation}
where one uses the discretisation $A(q) \rightarrow A_{\vect{j}}$. For details of the discretisation see Appendix~\ref{app:discretisationode}.\\

 The total computational cost in computing observables from Eq.~\eqref{eq:pdephi} is similar that in Lemma \ref{lem:liouville},  where Eq.~\eqref{eq:pdephi} has dimension $D+1$. 
\begin{lemma} \label{lem:liouville-1}
 The overall computational cost will be of $O(D N_{t, \text{CL}} N_{\text{CL}}^{D})=O(D^2T N_{\text{CL}}^{D+1})=O(D^{D+3}T(1/\epsilon_{\text{CL}})^{3D+3})$.  
\end{lemma}


\noindent{\bf Remark:} This method is related to the particle methods \cite{Rav85} which is a popular classical method used to solve linear transport equations in high dimensions. There one uses particles, which satisfy the ODEs \eqref{eq:upde2}, to solve the linear transport (Liouville \eqref{eq:pdephi} here) equation. Our method here is, in a sense, like an \textit{inverted} particle method--here we use the Liouville equation \eqref{eq:pdephi} to approximate the particle (ODE) system \eqref{eq:upde2}, instead of the other way around. 

\subsection{The quantum algorithm to estimate the ensemble average}

We proceed in a similar way to Section~\ref{sec:levelsetobservable} and define the initial quantum state that encodes the initial condition 
\begin{align}
    |\Phi_{0}\rangle=\frac{1}{N_{\Phi_{0}}}\sum_{\vect{j}}^{N} \Phi_{0, \vect{j}} |\vect{j}\rangle |n=0\rangle
\end{align}
where $\Phi_{0, \vect{j}}$ is defined in Eq.~\eqref{eq:phiinitial}. The normalisation is given by $N_{\Phi_{0}}=\sqrt{\sum_{\vect{j}}^{N}|\Phi_{0, \vect{j}}|^2}$. We can also define the state 
\begin{align}
    |A^n\rangle \equiv \frac{1}{N_A}\sum_{\vect{j}}^{N}A^*_{\vect{j}}|\vect{j}\rangle |n\rangle
\end{align}
where we note that here we do \textit{not} sum over the time step index $n$, where $n\Delta t$ is the time-step that we later want to compute $\langle A^{\omega}_n \rangle$ in Eq.~\eqref{lem:expectationa}. The normalisation is $N_A=\sqrt{\sum_{\vect{j}}^N |A^*_{\vect{j}}|^2}$. \\

Using Eq.~\eqref{quadrature} and given the states $|\Phi\rangle$ and $|A^n\rangle$, one can recover the observable from
\begin{align} \label{lem:expectationa}
    \langle A(t_n)\rangle \approx \langle A^{\omega}_n \rangle \equiv \frac{1}{N^D}\sum_{\vect{j}=1}^N A_{\vect{j}} \Phi_{\omega}(t_n, \vect{j})=n_{\Phi_{0}}n_A |\sqrt{\Upsilon_{\text{ODE}}}| 
\end{align}
where $n_{\Phi_{0}} \equiv N_{\Phi_{0}}/N^{D/2}$, $n_A \equiv N_A/N^{D/2}$ and a quantum algorithm allows one to compute
\begin{align}
    \Upsilon_{\text{ODE}} \equiv \langle \Phi_{0}|(\mathcal{M}_{\text{ODE}}^{-1})^{\dagger} \mathcal{A} \mathcal{M}_{\text{ODE}}^{-1}|\Phi_{0}\rangle.
\end{align}
Now the normalisations can similarly be estimated like in Section~\ref{sec:levelsetobservable}

\begin{lemma} \label{lem:normalisationsode}
The constants $n_A \equiv N_A/N^{d/2}=O(1)$ and $n_{\Phi_{0}} \equiv N_{\Phi_{0}}/N^{D/2}=O(1)$.
\end{lemma}
\begin{proof}
See Appendix~\ref{sec:normalisations}. 
\end{proof}

The quantum algorithm only allows one to estimate $\Upsilon_{\text{ODE}}$ to finite precision, so we denote the quantum estimate as 
\begin{align}
    \langle \tilde{A}^{\omega}_n\rangle \equiv n_{\Phi_{0}}n_A |\sqrt{\tilde{\Upsilon}_{\text{ODE}}}|
\end{align}
where $|\sqrt{\Upsilon_{\text{ODE}}}-\sqrt{\tilde{\Upsilon}_{\text{ODE}}}| \leq \epsilon_A$. The error analysis then proceeds in exactly the same way as Section~\ref{sec:levelsetobservable}. Then we have the following contributions to the total error in estimating $A_O(t_n)$. 

\begin{lemma} \label{lem:3errorsode} The error $|\langle A(t_n) \rangle-\langle \tilde{A}^{\omega}_n\rangle|$ can be broken into two independent sources of error
\begin{align}
    & |\langle A(t_n) \rangle-\langle \tilde{A}^{\omega}_n\rangle| \leq |\langle A(t_n) \rangle-\langle A^{\omega}_n \rangle|+|\langle A^{\omega}_n \rangle-\langle \tilde{A}^{\omega}_n\rangle| \nonumber \\
    & =\epsilon_C+\epsilon_{Q}=\epsilon.
\end{align}
Here $\epsilon_C$ comes from approximating the solutions of the discretised linear PDE \eqref{eq:pdephi}, as estimated in Eq.~\eqref{eq:acerror}. The quantum error comes from the estimate of $\Upsilon_{\text{ODE}}$, where $|\sqrt{\Upsilon_{\text{ODE}}}-\sqrt{\tilde{\Upsilon}_{\text{ODE}}}| \leq \epsilon_A$. Imposing $\epsilon_C \sim \epsilon_Q \sim \epsilon$, then we require $\epsilon_A \sim \epsilon/(n_{\Phi_{0}}n_A)\sim \epsilon_C/(n_{\Phi_{0}}n_A)$.
\end{lemma}


The query and gate complexity for estimating the observable $\langle A(t_n) \rangle$ is then identical to Theorem~\ref{thm:qquery} with the replacement $d \rightarrow D/2$.

\begin{theorem} \label{thm:qqueryode} The worst-case total query complexity $\mathcal{Q}$ to estimate the observable $\langle A(t_n) \rangle$ to precision $\epsilon$ on a quantum algorithm that takes sparse access 
 $(s=O(D), \|\mathcal{M}\|_{max}=O(1), O_{M}, O_{F})$ to $\mathcal{M}_{\text{ODE}}$, access to the unitaries $J(n)$ where $J(n)|0\rangle=|A^n\rangle$ and $U_{\Phi_{0}}$ where $U_{\Phi_{0}}|0\rangle=|\Phi_{0}\rangle$, is 
 \begin{align} \label{eq:qqueryode1}
  \mathcal{Q}= \mathcal{O}\left(\frac{n_{\Phi_{0}}^2D^7T^3}{\epsilon^{10}}\log \left(\frac{n_{\Phi_{0}}^2D^4T^2}{\epsilon^{7}}\right) \right)
\end{align}
where all constant terms $O(1)$ are suppressed except $T=O(1)$ and $n_{\Phi_{0}}=O(1)$. This complexity is independent of $M$. 
\end{theorem}
\begin{proof}
The proof is the same as Theorem~\ref{thm:qquery} except with the replacement $D=2d$. 
\end{proof}

\begin{corollary} We say there is a quantum advantage in estimating the observables $\langle A(T)\rangle$ to precision $\epsilon$ when $\mathcal{Q}=\littleo(\mathcal{C})$, which requires
\begin{align}
   \littleo\left(\frac{M}{n^2_{\Phi_{0}} D^4T^2(1/\epsilon)^9}\right)=\tilde{O}(1).
\end{align}
\end{corollary}
\begin{proof}
From Lemma~\ref{lem:odeclassical}, the classical cost is $\mathcal{C}=O(MD^3T/\epsilon)$. Since $n_{\Phi_0}=O(1)$, then from Theorem~\ref{thm:qqueryode}, we see $\mathcal{Q}= \tilde{O}(n^2_{\Phi_{0}}D^7T^3/\epsilon^{10})$ where $\tilde{O}$ suppresses all logarithmic terms in $D, \epsilon$.
\end{proof}

Here the only possible quantum advantage is with respect to $M$. Ignoring $n^2_{\Phi_{0}}=O(1)$ and $T=O(1)$ constant factors, quantum advantage is only possible for $M >\tilde{O}(D^4 (1/\epsilon)^9)$.\\ 


We note that here, unlike in the Hamilton-Jacobi and hyperbolic PDEs where the level set formalism is used, the $M$ initial data in the ODE case needs to be close to each other for neighboring initial points to form a mesh
good enough for the accuracy of solving the PDE in Eq.~\eqref{eq:pdephi}. There is no such constraint for the nonlinear PDE problems in Section~\ref{sec:HJ}. \\

In addition, the quantum algorithm for the nonlinear PDEs are much more efficient than the quantum  ODE solver, relatively speaking, when compared to their respective classical counterparts. In the ODE case, $D$ ODEs need to be solved by a $(D+1)$-dimensional linear PDE, while a $(d+1)$-dimensional Hamilton-Jacobi PDE, describes an infinite
dimensional dynamical system, can be represented by only a  $(2d+1)$-dimensional PDE.


\section{Solving more general nonlinear PDEs}
\label{sec:generalnonlinear}
In the special cases of nonlinear Hamilton-Jacobi and hyperbolic partial differential equations we saw how it is possible to reformulate the problem to linear partial differential equations by just (at most) doubling the dimension using the level set formalism. This same result, however, cannot be done analytically for  general nonlinear PDEs. \\

To devise quantum algorithms to solve more general $(d+1)$-dimensional nonlinear PDEs, like the Euler and Navier-Stokes equations in fluid dynamics, a naive way is to  first  numerically approximating the system so they become a system of nonlinear ODEs, and then use the quantum algorithm in Section~\ref{sec:ODEs}. Below we present two methods to achieve this: (i) the Lagrangian discretisation method and (ii) the Eulerian discretisation method. 


\subsection{The Lagrangian discretisations} 

One of the most important Lagrangian discretisation method in incompressible flows is the vortex method \cite{majda-Bertozzi}.   For the incompressible Euler equations the vortex method solves a Hamiltonian particle system like in Eq.~\eqref{eq:upde}. Although the diffusion term in the Navier-Stokes equation is usually modelled by stochastic ODEs (with Brownian motion for diffusion), there are also deterministic versions. This method of solving a $(d+1)$-dimensional nonlinear Euler or Navier-Stokes equation involves a description of $dN$ particles each obeying an ODE, where $N=O(d/\epsilon)$ and $\epsilon$ is the error in computing observables from the resulting ODE solutions. This means these PDEs can be reduced to a system of $O(d^2/\epsilon)$ nonlinear ODEs. \\

One can employ the quantum algorithm from Section~\ref{sec:ODEs} with quantum query and gate complexity cost $\mathcal{Q}=\tilde{\mathcal{O}}(d^{14}/\epsilon^{17})$, where we used Theorem~\ref{thm:qqueryode} with the replacement $D \rightarrow d^2/\epsilon$ and $\tilde{\mathcal{O}}$ suppressing logarithmic terms in $d$, $\epsilon$ and all $O(1)$ terms. The cost of classically solving the $(d+1)$-dimensional nonlinear PDE with a system of $dN$ nonlinear ODEs is $\mathcal{C}=O(MD^3/\epsilon)=O(Md^6/\epsilon^4)$ from Lemma~\ref{lem:odeclassical}. This quantum cost is independent of $M$ whereas the classical algorithm would depend linearly on $M$. Thus for a large enough $M>\tilde{O}(d^7/\epsilon^{13})$, one can expect an advantage in employing the quantum algorithm.\\



 There are other particle or mesh-free methods used for some other nonlinear PDEs like smoothed particle hydrodynamics (SPH) in solid mechanics and fluid flows \cite{SPH}. The Boltzmann equation for rarefied gas is often solved by the Direct Simulation Monte-Carlo method, which is a stochastic particle method \cite{bird1994}, and particle-in-cell (PIC) methods are often used for the Vlasov-Poisson or Vlasov-Maxwell systems \cite{SonKor} in plasma physics.

\subsection{The Eulerian discretisations}
Grid-based  Eulerian discretisations solve the PDEs on a fixed grid. They offer higher order accuracies but suffer from the curse-of-dimensionality.   When given the most general nonlinear $(d+1)$-dimensional PDE, one can first discretise the $d$ spatial dimensions with mesh size $1/N$. This then becomes a system of $D=N^d$ nonlinear ODEs. Since $N= O(d/\epsilon)$, where $O(\epsilon)$ is the error in computing observables from the ODE system as well as the error in the original PDE, this gives  $D=(d/\epsilon)^d$. Using Lemma~\ref{lem:odeclassical} and Theorem~\ref{thm:qqueryode} with the replacement $D \rightarrow (d/\epsilon)^d$, one obtains $\mathcal{Q}=\tilde{\mathcal{O}}(d^{7d}/\epsilon^{10+7d})$ and $\mathcal{C}=O(Md^{3d}/\epsilon^{1+3d})$. Thus one requires  $M>\tilde{O}(d^{4d}/\epsilon^{9+4d})$ for a quantum advantage. Since the size of $M$ scales exponentially with $d$, it does not fit known realistic problems. Better quantum methods to tackle these more general PDEs are thus still required.  



\section{Summary and discussion}

We introduced  quantum algorithms for computing physical observables of nonlinear Hamilton-Jacobi and scalar hyperbolic PDEs with arbitrary  nonlinearities. There are potential quantum speedups with respect to dimension $d$, precision $\epsilon$, the number of initial data $M$ and for arbitrary nonlinearity, with no large overheads in any other parameters. This would be of potential interest for the field of high-dimensional Hamilton-Jacobi equations, which appear in many areas including optimal control, mean-field games and machine learning. We found tight upper asymptotic query and gate complexities of quantum algorithms that compute ensemble averages of nonlinear PDEs and ODEs. This allows us to find a lower bound on the ratio between the classical and quantum cost to obtain an identical output. This can be equivalently expressed as 
\begin{align}
    \mathcal{O}\left(\frac{\mathcal{C}}{\mathcal{Q}}\right)=\tilde{O}\left(\frac{M}{T^2}d^{r_1}\left(\frac{1}{\epsilon}\right)^{r_{2}}\right)
    \end{align}
    where $(d+1)$ denotes the dimension of the PDE and $\tilde{O}$ suppresses  all logarithmic factors in $d$, $T$, $1/\epsilon$. For a system of $D$ nonlinear ODEs, we replace $d$ with $D$. The different values of the exponents $r_1$, $r_2$ depends on the type of PDE and its initial conditions. Here one sees that there is always a quantum advantage with respect to $M$ and no quantum advantage for any $T>1$. For a quantum advantage in $d$ and $\epsilon$, it is sufficient to have $r_1>0$ and $r_2>0$ respectively in the asymptotic limit. These results are summarised in Table~\ref{table:nonlin}.\\

    
   The query and gate complexities presented could be improved for instance by using higher order methods or using preconditioners for solving the Liouville equation. The quantum query and gate complexity with respect to the condition number can also be further optimised.  It is also worthwhile to consider these quantum algorithms for real applications with large $M$ and to extend these methods to stochastic PDEs. \\
   

       We see that even though quantum algorithms that utilise the linear representation for nonlinear PDEs (where the level set is one example) have advantages for Hamilton-Jacobi and scalar hyperbolic equations, this method is not necessarily generally applicable. For more general nonlinear PDEs, we see that even though quantum algorithms with quantum advantage are theoretically possible in regimes of large $M$,  there are not necessarily generally useful for real applications, so other methods are required. Other linear representations for nonlinear PDEs, rather than linear approximations or linear representations for ODEs, are desired if it is possible to find them. It may be that some nonlinear PDEs are more amenable to a quantum treatment than others and different classes of PDEs can be treated more efficiently with different quantum methods, rather than trying to find a one-method-fits-all strategy. 
       
       \section*{Acknowledgements} 
      N. Liu thanks Barry Sanders for very interesting and insightful discussions. S. Jin was partially supported by the NSFC grant No.~12031013 and the Shanghai Municipal Science and Technology Major Project (2021SHZDZX0102).  N. Liu acknowledges funding from the NSFC International Young Scientists Project (no.~12050410230), the Science and Technology Program of Shanghai (no.~21JC1402900) and the Natural Science Foundation of Shanghai grant 21ZR1431000.
\appendix 

\section{Proof of Lemma~\ref{thm-psi}} \label{app:hjjustification}
Denote the bicharacteristics of Eq.~\eqref{Liouville-delta} by
\begin{equation}\label{H-ODE}
    \frac{\partial x}{\partial t} = \nabla_p H, \quad
    \frac{\partial p}{\partial t} = -\nabla_x H, \quad 
    \qquad x(0)=x^{(0)}, \quad p(0)=p^{(0)}.
\end{equation} 
Assume sufficient smoothness of $H$, then the above ODE problem has a unique solution, denoted by
\[
(x(t; x^{(0)}, p^{(0)}),p(t; x^{(0)}, p^{(0)})),
\]
which can be inverted to get the inverse functions
$(x^{(0)}(t; x,p), p^{(0)}(t; x, p))$. This is possible since the Jacobian matrix of the map from $(x^{(0)}, p^{(0)})$ to $(x(t),p(t))$ has determinant $1$. Then by the method of characteristics, 
\begin{equation}
    \phi_i^{[k]}(t,x,p)=\phi_i^{[k]}(0,x^{(0)}(t; x,p), p^{(0)}(t; x, p))=p_i^{(0)}(t; x, p)-u_i^{[k]}(0, x^{(0)}(t; x,p))
\end{equation}
and, also by the method of characteristics and using the above solution for $\phi^{[k]}$,
\begin{equation}
    \psi(t,x,p)=\psi(0,x^{(0)}(t; x,p), p^{(0)}(t; x, p))=
    \frac{1}{M}\sum_{k=1}^M\prod_{i=1}^d\delta\left(p_i^{(0)}(t; x, p)-u_i^{[k]}(0, x^{(0)}(t; x,p))\right)
     =\frac{1}{M} \sum_{k=1}^M \delta(\phi^{[k]}(t,x,p))
\end{equation}

\section{Discretised PDE} \label{app:discretisation}
We proceed by first discretising the PDE in Eq.~\eqref{Liouville-delta} by, for example, finite difference schemes. The state with amplitudes proportional to the solutions of this discretised PDE can be created by using the quantum linear systems solver subroutine. For instance, we use the forward Euler method in time and upwind approximation in $x$ and $p$-derivatives, 

\begin{align}
    & \frac{\partial \psi(t_n, x, p)}{\partial t} \rightarrow \frac{\psi_{n+1, \vect{j}, \vect{l}}-\psi_{n, \vect{j}, \vect{l}}}{\Delta t} \nonumber \\
   & \frac{\partial H}{\partial p_i}\frac{\partial \psi(t_n, x, p)}{\partial x_i} \rightarrow \frac{1}{h}\left\{\frac{\partial H}{\partial p_i}\right\}_-\left[T^+_i\psi_n-\psi_n\right]_{\vect{j}, \vect{l}}
   +\frac{1}{h}\left\{\frac{\partial H}{\partial p_i}\right\}_+\left[\psi_n-T^-_i\psi_n\right]_{\vect{j}, \vect{l}}
   \nonumber \\
     & \frac{\partial H}{\partial x_i}\frac{\partial \psi(t_n, x, p)}{\partial q_i} \rightarrow 
     \frac{1}{h}\left\{\frac{\partial H}{\partial x_i}\right\}_+\left[P^+_i\psi_n-\psi_n\right]_{\vect{j}, \vect{l}}
   +\frac{1}{h}\left\{\frac{\partial H}{\partial x_i}\right\}_-\left[\psi_n-P^-_i\psi_n\right]_{\vect{j}, \vect{l}}
     \nonumber
\end{align}
Here  $\alpha_+ = \max\{\alpha, 0\},\alpha_- = \min\{\alpha, 0\}$ for a general quantity $\alpha$, $ \vect{j} \equiv (j_1,...,j_d)$ and $\vect{l} \equiv (l_1,...,l_d)$ denote vectors corresponding to the discretised $x_i \rightarrow j_i h$ and $p_i \rightarrow l_i h$ respectively, where $j_i, l_i=1,...,N$ and $i=1,...,d$,   $n=0,1,...,N_t$. The time step is denoted $t_n=n \Delta t$. In addition, $T^{\pm}_i \psi_{n, \vect{j}, \vect{l}}=\psi_{n, j_1,...,j_{i \pm 1},...,j_d, \vect{l}}$ are translation operators with respect to $x$ while $P^{\pm}_i \psi_{n, \vect{j}, \vect{l}}=\psi_{n,\vect{j}, l_1,...,l_{i \pm 1},...d}$ are translation operators with respect to $p$ while operating on $\psi_{n, \vect{j}, \vect{l}}$.  Note here $\frac{\partial H}{\partial p_i}$ and $\frac{\partial H}{\partial x_i}$ can be evaluated analytically, for given $H(x,p)$, thus do not need to be approximated.

Define $\lambda = \Delta t/h$. We require the CFL condition 
\begin{equation}\label{CFL}
     d\lambda \,  \max_i \sup_{x,p} \left\{ \left|\frac{\partial H}{\partial x_i}\right|, \left|\frac{\partial H}{\partial p_i} \right| \right\} \le 1
\end{equation}
for numerical stability. Then the discretised version of  Eq.~\eqref{Liouville-delta} can be rewritten as 
\begin{align}
 &&   \psi_{n+1, \vect{j}, \vect{l}}-\psi^{n}_{\vect{j}, \vect{l}}+\lambda\sum_{i=1}^d \left[\left\{\frac{\partial H}{\partial p_i}\right\}_-(T^+_i-I)
   +\left\{\frac{\partial H}{\partial p_i}\right\}_+(I-T^-_i)\right.\\
   && 
  \left.-\left\{\frac{\partial H}{\partial x_i}\right\}_+(P^+_i-I)
   -\left\{\frac{\partial H}{\partial x_i}\right\}_-(I-P^-_i) \right]\psi_{n, \vect{j}, \vect{l}} 
   =0,
\end{align}
where $I$ is the identity operator,
with the initial condition (for $n=0$)
\begin{align} \label{eq:psiinitial}
    \psi_{0, \vect{j}, \vect{l}}= \frac{1}{M}\sum_{k=1}^M  \prod_{i=1}^d\delta_\omega(l_ih-u^{[k]}_i(n=0, \vect{j})),
\end{align}
where $\delta_{\omega}$ is the smoothed delta function. \\

Now the discretised PDEs can be written as a matrix equation
\begin{align} \label{eq:psimatrix}
    \mathcal{K}\begin{pmatrix}
    \psi_{1, \vect{j}, \vect{l}} \\
    \psi_{2, \vect{j} ,\vect{l}} \\
    \vdots \\
    \psi_{N_t-1, \vect{j}, \vect{l}} \\
    \psi_{N_t, \vect{j}, \vect{l}}
    \end{pmatrix}=\begin{pmatrix}
    \psi_{0, \vect{j}, \vect{l}} \\
    0 \\
    \vdots \\
    0 \\
    0
    \end{pmatrix}
\end{align}
and $\mathcal{K}$  is a $N_t N^{2d} \times N_t N^{2d}$ Toeplitz matrix of the form
\begin{align} \label{eq:kmatrixpsi}
    \mathcal{K}=\begin{pmatrix}
    \mathbf{1} & 0 & 0 & \hdots & 0 & 0 & 0 \\
     K & \mathbf{1} & 0 & \hdots & 0 & 0 & 0 \\
    \vdots & \vdots & \vdots & \vdots & \vdots & \vdots & \vdots \\
    0 & 0 & 0 & \hdots &  K & \mathbf{1} & 0 \\
    0 & 0 & 0 & \hdots & 0 & -\mathbf{1} & \mathbf{1} 
    \end{pmatrix}
\end{align}
where each $\mathbf{1}$ above is the $N^{2d}\times N^{2d}$ identity matrix and $K$ is the $N^{2d}\times N^{2d}$ matrix: 

\begin{align}\label{matrixK}
    K=-I +\lambda\sum_{i=1}^d \left[\left\{\frac{\partial H}{\partial p_i}\right\}_-(T^+_i-I)
   +\left\{\frac{\partial H}{\partial p_i}\right\}_+(I-T^-_i)
 -\left\{\frac{\partial H}{\partial x_i}\right\}_+(P^+_i-I)
   -\left\{\frac{\partial H}{\partial x_i}\right\}_-(I-P^-_i) \right].
\end{align}
 We can then solve for $\psi_{n, \vect{j}, \vect{l}}$ by matrix inversion
\begin{align} \label{eq:matrixinvpsi}
    \begin{pmatrix}
    \psi_{1, \vect{j}, \vect{l}} \\
    \psi_{2, \vect{j}, \vect{l}} \\
    \vdots \\
    \psi_{N_t-1, \vect{j}, \vect{l}} \\
    \psi_{N_t, \vect{j}, \vect{l}}
    \end{pmatrix}=\mathcal{K}^{-1} \begin{pmatrix}
    \psi_{0, \vect{j}, \vect{l}} \\
    0 \\
    \vdots \\                                                  
    0 \\
    0
    \end{pmatrix}. 
\end{align}
Both QLSP and SLEP involve the inversion of an Hermitian matrix $\mathcal{M}$. Since $\mathcal{K}$ from Eq.~\eqref{eq:kmatrixpsi} is not Hermitian, we can define a new Hermitian matrix 
\begin{align}
    \mathcal{M}=\begin{pmatrix} 0 & \mathcal{K} \\
    \mathcal{K}^{\dagger} & 0 
    \end{pmatrix}
\end{align}
which has the same sparsity and condition number as $\mathcal{K}$. \\

Using $\mathcal{M}$, the matrix inversion problem in order to solve $\psi_{n, \vect{j}, \vect{l}}$ s
\begin{align} \label{eq:matrixinvpsi2}
    \begin{pmatrix}
    \vect{0} \\
    \psi_{n, \vect{j}, \vect{l}}
    \end{pmatrix}=\mathcal{M}^{-1} \begin{pmatrix}
    \psi_{0, \vect{j}, \vect{l}} \\
    \vect{0}
    \end{pmatrix}
\end{align}
   where $\vect{0}$ is a zero-vector of the same dimension as $\mathcal{K}$.
   

\section{Condition number of matrix $\mathcal{M}$} \label{app:skappaproof}

For clarity we will only derive this for linear constant coefficient transport equations. We start with the  simple one-dimensional equation
\begin{equation}
    \frac{\partial}{\partial t} u + \frac{\partial}{\partial x} u=0, \quad  0<x<1, \quad 
u(t, 0)=g(0),
\end{equation}
which will be discretised by the upwind scheme
\begin{equation}
    u^{n+1}_j+(\lambda -1)u_j^n -\lambda u^n_{j-1} =0,
\end{equation}
where $\lambda=\Delta t/h, j=1, \cdots,N, n=0, \cdots N_t-1$. 
Now the matrix $K$ in \eqref{matrixK} is
\begin{align}
K =
\begin{bmatrix}
\lambda-1  &        &    &       &         \\
 -\lambda  &  \lambda-1       &    &      &    \\
 &     \ddots  & \ddots    &     &    \\
    &         & \ddots    & \ddots  &  \\
    &          &           &  -\lambda  & \lambda-1 \\
\end{bmatrix}_{(N) \times (N_t)}.
\end{align}
Since the eigenvalues of $\mathcal{M}$ is just the singular values $\sigma$ of $\mathcal{K}$ we now estimate
the latter. Assume the CFL condition $\lambda \le 1$. By Gershogorin's theorem it is easy to see
$\sigma_{\text{max}}\le 2$.  To estimate the smallest singular value, let $L=-K$. Then
\begin{align}
\mathcal{K}^{-1} =
\begin{bmatrix}
I &        &    &       &         \\
 L  &  I      &    &      &    \\
 L^2&    L & I   &     &    \\
   \vdots  &  \vdots   & \ddots    & \ddots  &  \\
  L^{N_t-1}  &   L^{N_t-2}       &  \cdots         &  L & I \\
\end{bmatrix}_{(N) \times (N_t)}
=
\begin{bmatrix}
I &        &    &       &         \\
   &  I      &    &      &    \\
 &    & I   &     &    \\
    &    &    & \ddots  &  \\
    &         &        &  & I \\
\end{bmatrix}
+
\begin{bmatrix}
 &        &    &       &         \\
L   &      &    &      &    \\
 &  L  & &     &    \\
  &    & \ddots    &   &  \\
    &         &        & L &  \\
\end{bmatrix}
+\begin{bmatrix}
 &        &    &       &         \\
  &      &    &      &    \\
L^2 &    &  &     &    \\
  &  \ddots  &    &   &  \\
    &         &     L^2   &  &  \\
\end{bmatrix}
+\cdots
\end{align}
Hence
\begin{equation}
    \sigma_{\text{max}}(\mathcal{K}^{-1}) = \| \mathcal{K}^{-1}\|_2\le \|I\|_2+\|L\|_2+ \cdots +\|L^{N_t-1}\|_2.
\end{equation}
By Gershigorin's theorem,
\[
  \eta:=\|L\|_2 = \|K\|_2 \le (1-\lambda)+\lambda =1,
 \]
 therefore
 \begin{equation}
    \sigma_{\text{max}}(\mathcal{K}^{-1}) \le 1+\eta+\cdots \eta^{N_t-1} \le N_t.
\end{equation}
Consequently
\[
\kappa(\mathcal{K})=\frac{\sigma_{\text{max}}}{\sigma_{\text {min}} }\le N_t/2.
\]

We next consider the $d$-dimensional equation:
\begin{equation}
    \frac{\partial}{\partial t} u + \sum_{l=1}^d \frac{\partial}{\partial x_l} u=0, \quad  0<x_l<1, \quad
{\text{with incoming boundary conditions,}}
\end{equation}
which will be discretised by the upwind scheme
\begin{equation}\label{d-upwind}
    u^{n+1}_{\vect{j}}- u^{n}_{\vect{j}} + \lambda \sum_{l=1}^d (u_{\vect{j}}^n - u^n_{\vect{j}-e_l}) =0,
\end{equation}
where $e_l$ is the unit vector with $1$ in the $l$-th entry and $0$ elsewhere, $\lambda=\Delta t/h, j_l=1, \cdots,N, 1\le l\le N, n=0, \cdots N_t-1$. 

Define
\begin{align}
    {u} = (u_{1, 1, \cdots, 1},  \cdots,  u_{N, 1, \cdots, 1}, u_{1, 2, \cdots, 1}, \cdots, u_{1, N, \cdots, 1}, \cdots, u_{N, N, \cdots, 1}, \cdots, u_{N, N, \cdots, N})^T \nonumber 
\end{align}
then scheme \eqref{d-upwind} can be written in vector form as
\begin{align}
{u}^{n+1}-B {u}^n=0, \nonumber 
\end{align}
where 
\begin{align}
B=\lambda[T_h \otimes I^{\otimes (n-1)} +I \otimes T_h \otimes I^{\otimes (n-2)}+ \cdots
+I^{\otimes (n-1)} \otimes T_h] + (1-d\lambda ) I^{\otimes n}, \nonumber 
\end{align}
with
\begin{align}
T_h =
\begin{bmatrix}
0 &        &    &       &         \\
 1 &  0   &    &      &    \\
  &  \ddots  &   \ddots &   &  \\
    &         &   & 1 & 0 \\
\end{bmatrix}_{N\times N}. \nonumber 
\end{align}

Let 
\[
\tilde{T}_h=(1-d\lambda)I + \lambda T_h
\]
For clarify, consider the case of $d=3$. A direction calculation shows that the first few rows of $B$ are
\[
B_1=
\begin{bmatrix}
\tilde{T}_h &        &    &       &    & & & &     \\
 \lambda T_h  &   \tilde{T}_h  &    &      &    \\
  &  \ddots  &   \ddots &   &  &&&& \\
    &         &   & \lambda T_h &   &&&& \\
   \lambda I &        &    &       &  \tilde{T}_h  & & & &     \\
 &\lambda I  &        &      &   \lambda T_h & \tilde{T}_h  &&& \\
  &   &   \ddots &   &  & \ddots& \ddots && \\
    &         &   & \lambda I &   &&  \lambda T_h &  \tilde{T}_h& \\
\end{bmatrix}
\]
after ignoring subsequent repeating blocks. Assume the stability (CFL) condition $3\lambda \le 1$, 
then Gershgorin's theorem implies
\[
\|B\|_2 \le [(1-3\lambda)+\lambda]+\lambda+\lambda \le 1\,.
\]
The rest of the proof of the condition number of the HHL matrix  will be the same as the case of $d=1$.

Similarly, for general $d$, assume the stability (CFL) conditon
\begin{equation}\label{d-CFL}
\d\lambda \le 1,
\end{equation}
one gets the condition number of the HHL matrix 
\begin{equation}
    \kappa(\mathcal{M}) \lesssim N_t
\end{equation}
while the sparsity of $\mathcal{M}$ is obviously $\sim d$.

Assume one wants to compute to time $T$, namely $N_t \Delta t=T$. Using the stability condition
\eqref{d-CFL}, one gets
\[
\kappa(\mathcal{M}) \lesssim N_t \sim dNT\,.
\]
\section{Proof of Lemma~\ref{thm:classicalerror}}\label{app:classicalerrorproof}

Let $\psi^\omega$ be the analytical solution of \eqref{Liouville-delta} with initial data ~\eqref{Liou-I} in which $\delta$ is replaced by $\delta_\omega$, and $\psi^{\omega,h}$ is the upwind discretisation of
$\psi^\omega$. Then

\begin{eqnarray*}
 &&\epsilon_C \equiv |\langle G(t_n,x_{\vect{j}}) \rangle-G^n_{\omega} (\vect{j})|\\
\le && \left| \int_{\Omega_p} G(p) (\psi(t_n, x_{\vect{j}}, p) -\psi^\omega(t_n, x_{\vect{j}}, p))\, dp \right|
+ 
\left|\int_{\Omega_p} G(p) \psi^\omega(t_n, x_{\vect{j}}, p))\, dp -\frac{1}{N_p} \sum_{\vect{l}} G(p_{\vect{l}} )
  \psi^\omega_{n, \vect{j}, \vect{l}}\right|\\
  &&  \qquad + \left|\frac{1}{N_p} \sum_{\vect{l}} G(p_{\vect{l}}) 
  (\psi^\omega_{n, \vect{j}, \vect{l}}-\psi^{\omega,h}_{n, \vect{j}, \vect{l}})\right| \\
 = && I+II+III
  \end{eqnarray*}

Denote the bicharacteristics of \eqref{Liouville-delta} by
\begin{equation}\label{H-ODE}
    \frac{\partial x}{\partial t} = \nabla_p H, \quad
    \frac{\partial p}{\partial t} = -\nabla_x H, \quad 
    \qquad x(0)=x^{(0)}, \quad p(0)=p^{(0)}
\end{equation} 
Assume sufficient smoothness of $H$, and denoted  by 
$(x(t; x^{(0)}, p^{(0)}),p(t; x^{(0)}, p^{(0)}))$ the unique solution of  the above ODE problem, which can be inverted to get the inverse functions
$(x^{(0)}(t; x,p), p^{(0)}(t; x, p))$. Note that this is possible since the Jacobian matrix of the map from $(x^{(0)}, p^{(0)})$ to $(x(t),p(t))$ has determinant $1$. Then by method of characteristics, 
\begin{eqnarray*}
I = && \left| \int_{\Omega_p} G(p) \left[\psi(0, x^{(0)}(t_n,x_{\vect{j}},p), p^{(0)}(t_n,x_{\vect{j}},p))-\psi^\omega(0, x^{(0)}(t_n,x_{\vect{j}},p), p^{(0)}(t_n,x_{\vect{j}},p))\right]\, dp \right|\\
= && \left| \int_{\Omega_p} G(p) \left[\delta(p^{(0)}(t_n,x_{\vect{j}},p)-u_0(x^{(0)}(t_n,x_{\vect{j}},p)))-\delta_\omega(p^{(0)}(t_n,x_{\vect{j}},p)-u_0(x^{(0)}(t_n,x_{\vect{j}},p)))\right]\, dp \right|.
\end{eqnarray*}
Here for notation clarity we assume $M=1$. The more general case is obtained simply by linear superposition.

Assume $\chi(t_n,x_{\vect{j}},p)=p^{(0)}(t_n,x_{\vect{j}},p)-u_0(x^{(0)}(t_n,x_{\vect{j}},p))$ has (just one)
root $p_*\in \Omega_p$.  The case of multiple roots can be dealt with similarly. Then
\begin{equation}
    \int_{\Omega_p} G(p) \delta(p^{(0)}(t_n,x_{\vect{j}},p)-u_0(x^{(0)}(t_n,x_{\vect{j}},p))) dp
    =\frac{G(p_*)}{\left|\frac{\partial\chi}{\partial p}(p_*)\right|},
\end{equation}
while, by the mean value theorem,
\begin{eqnarray}
&&\int_{\Omega_p} G(p) \delta_\omega(p^{(0)}(t_n,x_{\vect{j}},p)-u_0(x^{(0)}(t_n,x_{\vect{j}},p)))\, dp\\
=&& \int_{\Omega_p} G(p) \delta_\omega\left( \frac{\partial\chi}{\partial p}(p_{**})(p-p_*)\right)\, dp
=
\int_{\Omega_p} G\left(\frac{q}{\frac{\partial\chi}{\partial p}(p_{**}(p))+ p_*}\right) \frac{\delta_\omega(q)}{\left|\frac{\partial\chi}{\partial p}(p_{**}(p))\right|}\, dp
\\
=&& \int_{|q|\le \omega} G\left(\frac{q}{\frac{\partial\chi}{\partial p}(p_{**})}+ p_* \right) \frac{\delta_\omega(q)}{\left|\frac{\partial\chi}{\partial p}(p_{**})\right|}\, dp
= 
G\left(\frac{\tilde{q}}{\frac{\partial\chi}{\partial p}(\hat{q})}+ p_* \right) \frac{1}{\left|\frac{\partial\chi}{\partial p}(\hat{q})\right|} \int_{|q|\le \omega}\delta_\omega(q)\, dq\\
=&& 
G\left(\frac{\tilde{q}}{\frac{\partial\chi}{\partial p}(\hat{q})}+ p_* \right) \frac{1}{\left|\frac{\partial\chi}{\partial p}(\hat{q})\right|}
=\frac{G(p_*)}{\left|\frac{\partial\chi}{\partial p}(p_*)\right|} + O(\omega), \nonumber 
\end{eqnarray}
since $|\tilde q|\le \omega, |\hat{q}-q_*|\le \omega$. 
Combining the above two estimates gives $I=O(\omega)$.

II is just the error of the (midpoint) quadrature rule, hence 
$II=O(d(h/\omega)^2)$, where the $1/\omega^2$ factor comes from
the second derivative of $\psi^\omega$ in the truncation error.

By standard error analysis for linear hyperbolic equation \cite{LeV-book}, if the first order upwind scheme is used, under the CFL condition, one has
\[
|\psi^\omega_{n, \vect{j}, \vect{l}}-\psi^{\omega,h}_{n, \vect{j}, \vect{l}}| = O(d h/\omega^2)
\]
where $1/\omega^2$ factor again comes from
the second derivative of $\psi^\omega$ in the consistency error.

By combining I, II and III the lemma is proved.

\section{Estimating normalisation constants: proof of Lemmas~\ref{lem:normalisations}, \ref{lem:normalisations2} and \ref{lem:normalisationsode}}\label{sec:normalisations}

Here we prove Lemma~\ref{lem:normalisations}.
\begin{proof}
To estimate $n_G$, one can use $O(1)\sim\int dp |G(p)|^2$. This is possible since to compute any observables we only need to consider the domain of $\phi(t,x,p)$, which has compact support in $p$ with support of size $O(1)$ domain, thus we only need to integrate $p$ in an $O(1)$ domain. Then one can write
\begin{align}
    O(1)\sim \int dp|G(p)|^2\approx \frac{1}{N^d}\sum_{\vect{l}}|G_{\vect{l}}|^2=\frac{N_G^2}{N^d}=n^2_G,
\end{align}
thus $n_G\sim O(1)$. For a lower bound of $n^2_{\psi_0}$, we use $\Upsilon \leq 1$ since we estimate $\Upsilon \approx \langle 0|\mathcal{U}|0\rangle$ from a large unitary matrix $\mathcal{U}$ (see Appendix~\ref{app:algorithmsteps}) and can therefore view it as a quantum fidelity. Since the ensemble average is an observable and should be also independent of $N$ and we assume to be independent of $d$, then $O(1)\sim G^{\omega}_{n,\vect{j}}$. Together with $\Upsilon\leq 1$ and $n_G\sim O(1)$, we have $n_{\psi_0}\gtrsim O(1)$. \\

To obtain an upper bound of $n_{\psi_0}$, one considers the case where there are no assumption on the initial data. One can estimate  $N_{\psi_{0}}$ with $\psi (0,x,p)$ defined by Eq.~\eqref{Liou-I} with the delta function approximated by $\delta_{\omega}$. We start with $M=1$. Without loss of generality we assume 
$(p,x) \in [0,1]^{2d}$ and  
\begin{eqnarray}\nonumber
 &&  \int\int |\psi(0,x,p)|^2 dp\,dx
   \approx\int\int \delta_\omega(p-u^{[k]})^2 dp\,dx 
  \le \frac{1}{\omega^d} \int\int \delta_\omega(p-u^{[k]}) dp\,dx\\
\label{psi-init-size}   = &&\frac{1}{\omega^d} \int 1\, dx=\frac{1}{\omega^d}.
\end{eqnarray}
To obtain an upper bound of $n_{\psi_0}$, one considers the case where there are no assumptions on the initial conditions. One can estimate  $N_{\psi_{0}}$ with $\psi (0,x,p)$ defined by Eq.~\eqref{Liou-I} with the delta function approximated by $\delta_{\omega}$. We start with $M=1$. Without loss of generality we assume $(p,x)\in [0,1]^{2d}$ and 
 \begin{eqnarray}\nonumber
 &&  \int\int |\psi(0,x,p)|^2 dp\,dx
   \approx\int\int \delta_\omega(p-u^{[k]})^2 dp\,dx 
  \le \frac{1}{\omega^d} \int\int \delta_\omega(p-u^{[k]}) dp\,dx\\
\label{psi-init-size}   = &&\frac{1}{\omega^d} \int 1\, dx=\frac{1}{\omega^d}
 \end{eqnarray}
 Next we approximate the above integral by the quadrature rule. Here we use $N^d$ mesh points in $x$. Note $\delta_\omega(p-u^{[k]})$  has support of the size of only $O(\omega^d)=O((mh)^d)=O((m/N)^d)$ in $p$, thus the number of mesh points in $p$ is of
 order $m^d$, where $m$ is the number of non-zero entries in the discrete delta function $\delta_{\omega}$. Then using the quadrature rule,
 \begin{equation}\label{quad-n}
   \int\int |\psi(0,x,p)|^2 dp\,dx
   \approx
\frac{1}{(mN)^d}\sum_{\vect{j}}^N\sum_{\vect{l}}^N |\psi_{0, \vect{j}, \vect{l}}|^2
=\frac{1}{(mN)^d}N_{\psi_{0}}^2.
\end{equation}

 From Eqs.~\eqref{psi-init-size} and \eqref{quad-n} one gets
 \begin{equation}\label{psi-init-n}
 N_{\psi_{0}}=O\left( \frac{(mN)^{d/2}}{\omega^{1/2}}\right)
 =O\left( \frac{(mN)^{d/2}}{(m/N)^{d/2}}\right)
 =O\left(N^{d}\right)
 \end{equation}
 
 More generally, suppose the initial data has support in
a box $\Omega_\beta$ of size $\beta$. This means one can write
 \begin{equation}
  \int_{\Omega_\beta}\int |\psi(0,x,p)|^2 dp\,dx
\label{psi-init-size2}   \le \frac{1}{\omega^d} \int_{\Omega_\beta} 1\, dx=\frac{\beta^d}{\omega^d}
 \end{equation}
 
 while its numerical approximation is
 \begin{equation}\label{quad-n2}
   \int\int |\psi(0,x,p)|^2 dp\,dx
   \approx
\frac{1}{(m\beta N)^d}\sum_{\vect{j}}^N\sum_{\vect{l}}^N |\psi_{0, \vect{j}, \vect{l}}|^2 
=\frac{1}{(m\beta N)^d}N_{\psi_{0}}^2.
\end{equation}

Eqs.~\eqref{psi-init-size2} and~\eqref{quad-n2} give
 \begin{equation}\label{psi-init-n2}
 N_{\psi_{0}}=O\left((\beta N)^{d}\right).
 \end{equation}
 We recover the upper bound in the limit $\beta=O(1)$. 
 The case of $M>1$ will yield the same estimate since it is just
 the average of $M$ different initial data of the same type.
\end{proof}

The proof of Lemma~\ref{lem:normalisations2} proceeds in exactly the same way for $n_g$ and the lower bound of $n_{\Phi_0}$. For more general cases, there is a small modification in the proof.
\begin{proof}
    The estimate  of $N'_{\psi_{0}}$ is pretty much the same as that in Lemma \ref{lem:normalisations}, except here $p \in \mathbb{R}^1$. Thus
 \begin{eqnarray}\nonumber
 &&  \int_{\Omega_{\beta}} \int |\psi(0,x,p)|^2 dp\,dx
    \approx\int\int \delta_\omega(p-u^{[k]})^2 dp\,dx 
   \le \frac{1}{\omega} \int\int \delta_\omega(p-u^{[k]}) dp\,dx\\
\label{psi-init-size}   = &&\frac{1}{\omega} \int_{\Omega_{\beta}} 1\, dx=\frac{\beta^d}{\omega},
 \end{eqnarray}
while its numerical quadrature approximation is 

 \begin{equation}\label{quad-n3}
   \int\int |\psi(0,x,p)|^2 dp\,dx
   \approx
\frac{1}{m(\beta N)^d}\sum_{\vect{j}}^{N}\sum_{\vect{l}}^N |\psi_{0, \vect{j}, \vect{l}}|^2 
=\frac{1}{m(\beta N)^d}(N'_{\psi_{0}})^2.
\end{equation}
Eqs.~\eqref{psi-init-size} and~\eqref{quad-n3} give
 \begin{equation}\label{psi-init-n3}
 N_{\psi_{0}}=O\left(\beta^d N^{(d+1)/2}\right),
 \end{equation}
 so the upper bound when $\beta=O(1)$ is $N_{\psi_{0}}=O\left(N^{(d+1)/2}\right)$. 
\end{proof}

The proof of Lemma~\ref{lem:normalisationsode} for estimating $n_A$ and the lower bound of $n_{\Phi_0}$ is identical to the estimation of $n_G$ and the lower bound to $n_{\psi_0}$, respectively, in Lemma~\ref{lem:normalisations}. The proof for the upper bound of $n_{\Phi_0}$ is the following.
\begin{proof}
The estimate  of $N_{\Phi_{0}}$ is pretty much the same as that in Lemma \ref{lem:normalisations}, except that  here there is no $x$ and $q \in \mathbb{R}^D$ thus
\begin{equation}\nonumber
 \int |\Phi(0,q)|^2 dq
    \approx\int \delta_\omega(q-X_0^{[k]})^2 dq 
   \le \frac{1}{\omega^D} \int \delta_\omega(q-u^{[k]}) dq
  =\frac{1}{\omega^D},
 \end{equation}
while
the quadrature rule becomes
\begin{equation}\label{quad-p2}
  \int |\Phi(0,q)|^2 dq \approx \frac{1}{m^D} \sum_{\vect{j}}^N |\Phi_{0, \vect{j}}|^2  \nonumber \\
 =\frac{1}{m^D}N_{\Phi_{0}}^2.
\end{equation}
Consequently 
$N_{\Phi_{0}}=O(N^{D/2})$.
\end{proof}

\section{Algorithm for computing $\Upsilon$ and proof of Lemma~\ref{lem:bslep2}} \label{app:A}

Here we prove Lemma~\ref{lem:bslep2} which is used for computing $\Upsilon \equiv \langle \psi_{0}|(\mathcal{M}^{-1})^{\dagger}\mathcal{G}\mathcal{M}^{-1}|\psi_{0}\rangle$. We attempt to keep this as self-contained as possible, so we include in full similar arguments and steps made in \cite{barrynew}, but with modifications required due to our assumption of beginning with sparse access $(s, \|\mathcal{M}\|_{max}, O_{M}, O_{F})$ to $\mathcal{M}$, instead of block access to $\mathcal{M}$. We also do not begin with block access to $\mathcal{G}$. However, block access to $\mathcal{M}$ can be constructed from sparse access from the following lemma.

\begin{lemma}\label{lem:sparsetoblock}
(\cite{barrynew,low2019hamiltonian}) Given $m$ qubits and sparse access $(s, \|\mathcal{M}\|_{max}, O_{M}, O_{F})$ to a $m$-qubit Hermitian matrix $\mathcal{M}$, block access $(\alpha_{\mathcal{M}}=s\|\mathcal{M}\|_{max}, n_{\mathcal{M}}=2, \delta_{\mathcal{M}}, U_{\mathcal{M}})$ to $\mathcal{M}$ can be implemented by making $\mathcal{O}(1)$ queries to sparse access oracles for $\mathcal{M}$ and $\mathcal{O}(m+\log^{2.5}(s\|\mathcal{M}\|_{max}/\delta_{\mathcal{M}}))$ additional $2$-qubit gates.
\end{lemma}

\begin{proof}
See Lemma II.5 in \cite{barrynew} and also Lemma 6 in \cite{low2019hamiltonian} for an explicit construction. 
\end{proof}
This means one cannot construct a block access to $\mathcal{M}$ with $\delta_{\mathcal{M}}=0$ using a finite number of $2$-qubit gates if one begins with sparse access to $\mathcal{M}$. Thus one cannot apply Lemma~\ref{lem:bslep} from \cite{barrynew} directly at this point and needs to extend to the case where one is instead given block access $(\alpha_{\mathcal{M}}, n_{\mathcal{M}}, \delta_{\mathcal{M}}>0, U_{\mathcal{M}})$ to $\mathcal{M}$ and block access $(\alpha_{\mathcal{G}}=1, n_{\mathcal{G}}, 0, U_{\mathcal{G}})$ to $\mathcal{G}$. In our scenario, $\mathcal{G}$ is a pure density matrix, which means it is simple to construct a $\delta_{\mathcal{G}}=0$ block access if only given access to unitaries that construct the state, using the following lemma.

\begin{lemma} \label{lem:gconstruction}(\cite{qsvdarxiv}) Given $\mathcal{G}=|G_{n, \vect{j}}\rangle \langle G_{n,\vect{j}}|$ is a $m$-qubit density matrix and $L(n, \vect{j})$ is a $m$-qubit unitary operator such that $L(n, \vect{j})|0^{m}\rangle=|G\rangle$, then $U_{\mathcal{G}}=(L(n, \vect{j})^{\dagger}\otimes \mathbf{1}^{\otimes m+1})(\mathbf{1}^{\otimes m+1}\otimes SWAP_m)(L(n, \vect{j})\otimes \mathbf{1}^{m+1})$ is a $(1,n_{\mathcal{G}}=2m+1,\delta_{\mathcal{G}}=0,U_{\mathcal{G}})$ block encoding of $\mathcal{G}$. 
\end{lemma}

\begin{proof}
Apply Lemma 45 in \cite{qsvdarxiv} to the case of a pure state $\mathcal{G}$. Note that the SWAP$_m$ gate can be constructed from $O(m)$ 2-qubit gates. 
\end{proof}
The block access $(\alpha_{\mathcal{M}}, n_{\mathcal{M}}, \delta_{\mathcal{M}}, U_{\mathcal{M}})$ can then be used to construct a block-access $(\alpha_{\mathcal{M}^{-1}}, n_{\mathcal{M}^{-1}}, \delta, U_{\mathcal{M}^{-1}})$ to $\mathcal{M}^{-1}$ as seen from Lemma~\ref{lem:finiteerrorm}. The block access $(\alpha_{\mathcal{G}}=1, n_{\mathcal{G}}, 0, U_{\mathcal{G}})$ can be used to construct block access $(1, n_{\mathcal{G}}+1, 0, U_{\mathcal{G}'})$ to $\mathcal{G}'=|0^{n_{\mathcal{M}^{-1}}}\rangle \langle 0^{n_{\mathcal{M}^{-1}}} |\otimes \mathcal{G}$, which will be shown in Lemma~\ref{lem:g0}.

\begin{lemma} \label{lem:g0} (\cite{barrynew}). Given block-access $(1,n_{\mathcal{G}}, 0, U_{\mathcal{G}})$ to $\mathcal{G}$, block-access $(1, n_{\mathcal{G}}+1, 0, U_{\mathcal{G}'})$ to $\mathcal{G}'=|0^{n_{\mathcal{M}^{-1}}}\rangle \langle 0^{n_{\mathcal{M}^{-1}}}|\otimes \mathcal{G}$ can be constructed where $U_{\mathcal{G}}$ is queried once with $\mathcal{O}(n_{\mathcal{M}^{-1}}^2)$ additional $2$-qubit gates. 
\end{lemma}
\begin{proof}
See Lemma IV.16 in \cite{barrynew} for an explicit construction.
\end{proof}

Given these ingredients, we are now ready to estimate $\Upsilon$ to error $\epsilon'$. Observe that 
\begin{align} \label{eq:udef}
 &u\equiv\langle 0^{m+n_{\mathcal{G}}+1+n_{\mathcal{M}^{-1}}}|(U^{\dagger}_{initial}\otimes \mathbf{1}^{n_{\mathcal{G}}+1+n_{\mathcal{M}^{-1}}})(U^{\dagger}_{\mathcal{M}^{-1}}\otimes \mathbf{1}^{n_{\mathcal{G}}+1}) U_{\mathcal{G}'}(U_{\mathcal{M}^{-1}}\otimes \mathbf{1}^{n_{\mathcal{G}}+1})(U_{initial}\otimes \mathbf{1}^{n_{\mathcal{G}}+1+n_{\mathcal{M}^{-1}}})|0^{m+n_{\mathcal{G}}+1+n_{\mathcal{M}^{-1}}}\rangle \nonumber \\
 &=\langle 0^{m+n_{\mathcal{M}^{-1}}}|(U^{\dagger}_{initial}\otimes \mathbf{1}^{n_{\mathcal{M}^{-1}}})U^{\dagger}_{\mathcal{M}^{-1}}(|0^{n_{\mathcal{M}^{-1}}}\rangle \langle 0^{n_{\mathcal{M}^{-1}}}|\otimes \mathcal{G})U_{\mathcal{M}^{-1}}(U_{initial}\otimes \mathbf{1}^{n_{\mathcal{M}^{-1}}})|0^{m+n_{\mathcal{M}^{-1}}}\rangle \nonumber \\
 &=\langle \psi_{0}| (0^{n_{\mathcal{M}^{-1}}}|U^{\dagger}_{\mathcal{M}^{-1}}|0^{n_{\mathcal{M}^{-1}}}\rangle) \mathcal{G} (0^{n_{\mathcal{M}^{-1}}}|U_{\mathcal{M}^{-1}}|0^{n_{\mathcal{M}^{-1}}}\rangle)|\psi_{0}\rangle=\langle \psi_{0}|\mathcal{L}^{\dagger}\mathcal{G}\mathcal{L}|\psi_{0}\rangle,
 \end{align}
 where $\mathcal{L}\equiv \langle 0^{n_{\mathcal{M}^{-1}}}|U_{\mathcal{M}^{-1}}|0^{n_{\mathcal{M}^{-1}}}\rangle$ and $|\psi_{0}\rangle=U_{initial}|0^m\rangle$.
 Since block-access to $\mathcal{M}^{-1}$ can be made with small enough error $\delta$ as  will be seen later, then 
 \begin{align}
 \alpha^2_{\mathcal{M}^{-1}}u \approx \langle \psi_{0}|(\mathcal{M}^{-1})^{\dagger}\mathcal{G}\mathcal{M}^{-1}|\psi_{0}\rangle \equiv \Upsilon.
\end{align}
 The benefit of approximating $\Upsilon$ using $u$ is that $u$ is written in terms of the expectation value with respect to unitary operators, so one is able to more elegantly employ the amplitude estimation algorithm \cite{knill2007optimal} to gain a quadratic speedup with respect to error in computing $\Upsilon$. This allows an optimal estimation of the expectation value. \\

The amplitude estimation algorithm \cite{knill2007optimal} is an algorithm for estimating the value $u\equiv |\langle 0^r|U^{\dagger}VU|0^r\rangle|_u$ 
to $\epsilon_u$-additive precision with success probability at least $2/3$, 
where $U$, $V$ are $r$-qubit unitary black-boxes. Here we have
\begin{align} \label{eq:uvdef}
    &U \equiv (U_{\mathcal{M}^{-1}}\otimes \mathbf{1}^{n_{\mathcal{G}}+1})(U_{initial}\otimes \mathbf{1}^{n_{\mathcal{G}}+1+n_{\mathcal{M}^{-1}}}), \nonumber \\
    &V \equiv U_{\mathcal{G}'}, \nonumber \\
    & r \equiv m+n_{\mathcal{G}}+1+n_{\mathcal{M}^{-1}}.
\end{align}
The key is to construct a unitary operation $S$ whose eigenvalue is $\exp(i\theta)$ where $u=|\cos(\theta/2)|$. Then a quantum phase estimation algorithm is used to extract $u$. Defining $|\phi_0\rangle \equiv U|0^r \rangle$ and $|\phi_1\rangle \equiv VU|0^r\rangle$, our desired quantity is then the inner product $u=\langle \phi_0|\phi_1\rangle$. The aim is then to construct $S$ as a rotation operator with eigenvalue $\exp(i\theta)$, that rotates $|\phi_0\rangle$ to $|\phi_1\rangle$ and the two states are separated by angle $2\theta=4\cos^{-1}(\langle \phi_0|\phi_1\rangle)$. Just like in Grover's search algorithm, one can construct the rotation operator as a combination of two reflection operators $S=S_0S_1$ where $S_0=\mathbf{1}^r-2|\phi_0\rangle \langle \phi_0|=UP_0 U^{\dagger}$ and $S_1=\mathbf{1}^r-2|\phi_1\rangle \langle \phi_1|=VUP_0U^{\dagger}V^{\dagger}$ where $P_0=\mathbf{1}^r-2|\phi_0\rangle \langle \phi_0|$. This means that one can construct the unitary operator
\begin{align} \label{eq:sdef}
 S=UP_0 U^{\dagger}VUP_0U^{\dagger}V^{\dagger}
\end{align}
by concatenating the unitary black-boxes $U_{\mathcal{G}}$, $U_{\mathcal{M}^{-1}}$ and $U_{initial}$ and their adjoints, which we already assume one has access to. Then it is straightforward to apply the standard quantum phase estimation algorithm using controlled-$S$, i.e., $|0\rangle\langle 0|\otimes \mathbf{1}+|1\rangle \langle 1|\otimes S$, which can be easily constructed from controlled-$U$, controlled-$V$ and are in turn constructed from controlled-$U_{\mathcal{M}^{-1}}$ and controlled-$U_{\mathcal{G}}$. The output of the quantum phase estimation algorithm is then a $2\epsilon_u$-additive error estimate of $\theta$, which leads to an $\epsilon_u$-additive error estimate of  $u=|\cos(\theta/2)|$. \\

Note that $u=|u|$ for our problem since it has the interpretation of being proportional to quantum fidelity, so we don't need to be concerned about possible negative expectation values $\langle 0|U^{\dagger}VU|0\rangle$ that can happen for general $U,V$. This means one also does not require the extra steps in \cite{barrynew} to deal with possible negative expectation values. \\

To obtain an $\epsilon_u$-additive error estimate to $u$, one can use the following lemma. 

\begin{lemma} (\label{lem:qpe}\cite{barrynew, knill2007optimal}) An amplitude estimation algorithm exists that gives a $\epsilon_u$-additive estimate of $u$ that makes $\mathcal{O}(1/\epsilon_u)$ queries to $U$, $V$ and $\mathcal{O}(r/\epsilon_u)$ additional $2$-qubit gates. 
\end{lemma}
\begin{proof}
See \cite{knill2007optimal} and Lemma II.15 from \cite{barrynew}.
\end{proof}
Since $\Upsilon\approx \alpha^2_{\mathcal{M}^{-1}}u$, one also needs to identify $\alpha_{\mathcal{M}^{-1}}$ to find the total query complexity. To demonstrate an appropriate value, one first requires three other lemmas.

\begin{lemma} \label{lem:newinvm} (\cite{qsvdarxiv, barrynew})
Given block-access $(\alpha_{\mathcal{M}}, n_{\mathcal{M}}, \delta_{\mathcal{M}}, U_{\mathcal{M}})$ to matrix $\mathcal{M}$, and error $\sigma>0$ and a polynomial $\mathcal{P}:\mathbb{R} \rightarrow \mathbb{R}$ of degree $D(\mathcal{P})$ satisfying $|\mathcal{P}(x)|\leq 1/2$ for all $x\in[-1,1]$, then block access $(1, n_{\mathcal{M}}, 4D(\mathcal{P})\sqrt{\delta_{\mathcal{M}}/\alpha_{\mathcal{M}}}+\sigma, U_{\mathcal{P}(\mathcal{M}/\alpha_{\mathcal{M}})})$ to $\mathcal{P}(\mathcal{M}/\alpha_{\mathcal{M}})$ can be constructed where each $U_{\mathcal{P}(\mathcal{M}/\alpha_{\mathcal{M}})}$ makes $2D(\mathcal{P})+1$ queries to $U_{\mathcal{M}}$ and $\mathcal{O}(n_{\mathcal{M}}D(\mathcal{P}))$ additional $2$-qubit gates.
\end{lemma}
\begin{proof}
See Theorem 56 \cite{qsvdarxiv} and Theorem II.6 from \cite{barrynew}. 
\end{proof}

\begin{lemma} \label{lem:poly} (\cite{qsvdarxiv})
For any $\chi, \zeta \in (0,1/2]$, a polynomial $\mathcal{P}:\mathbb{R} \rightarrow \mathbb{R}$ with odd degree $\mathcal{O}(\log(1/\zeta)/\chi)$ exists such that for all $x\in[-1,1]\backslash[-\chi, \chi]$, $|\mathcal{P}(x)|\leq 1$ and $|\mathcal{P}(x)-\chi/(2x)|<\zeta$.
\end{lemma}
\begin{proof}
See Corollary 67 in \cite{qsvdarxiv}.
\end{proof}
Now we can prove the following lemma and show how the choice $\alpha_{\mathcal{M}^{-1}}=4\theta\kappa/\|\mathcal{M}\|$ for any $\theta>1$ is possible. 

\begin{lemma}\label{lem:finiteerrorm} (\cite{qsvd})
If given $m$ qubits and block access $(\alpha_{\mathcal{M}}, n_{\mathcal{M}}, \delta_{\mathcal{M}}, U_{\mathcal{M}})$ to $\mathcal{M}$ with $\|\mathcal{M}\|\leq \alpha_{\mathcal{M}}$, then block access $(\alpha_{\mathcal{M}^{-1}}=4\kappa \theta/\|\mathcal{M}\|, n_{\mathcal{M}^{-1}}=n_{\mathcal{M}}+2, \delta, U_{\mathcal{M}^{-1}})$ to $\mathcal{M}^{-1}$ can be constructed, for any $\theta>1$, such that $U_{\mathcal{M}^{-1}}$ makes $\mathcal{O}((\alpha_{\mathcal{M}}\kappa\theta/\|\mathcal{M}\|) \log(\theta\kappa/(\|\mathcal{M}\|\delta)))$ queries to $U_{\mathcal{M}}$ and $\mathcal{O}((\alpha_{\mathcal{M}}\kappa n_{\mathcal{M}}\theta/\|\mathcal{M}\|)\log (\kappa \theta/(\|\mathcal{M}\|\delta)))$ additional $2$-qubit gates, with $\delta \geq (2\theta \kappa\zeta+4\theta \kappa\sigma+16\theta^2\kappa^2v \log(1/\zeta)\sqrt{\delta_{\mathcal{M}}\alpha_{\mathcal{M}}})/\|\mathcal{M}\|\
$, where appropriate $\sigma>0$, $\zeta \in[0,1/2)$ can be chosen, $v>0$ and $\delta \in (0,2]$.
\end{lemma}

\begin{proof}
This proof is an extension of Corollary IV. 15 in \cite{barrynew} for the case of $\delta_{\mathcal{M}}>0$ and similar steps are used. Here in order to use Lemma~\ref{lem:newinvm} where $x \in [-1,1]\backslash [-\chi, \chi]$ and we want $x$ to represent the spectrum of the operator $\mathcal{M}/\alpha_{\mathcal{M}}$, its spectrum should also lie within $[-1,1]\backslash [-\chi, \chi]$. One
can just look at the positive eigenvalues without losing generality. If $\lambda_{min}$ and $\lambda_{max}$ represent the smallest and largest absolute values of the eigenvalues of $\mathcal{M}$, then the above requirement demands
\begin{align} \label{chi-ineq}
    \chi < \frac{\lambda_{min}}{\alpha_{\mathcal{M}}} \leq \text{Spec}\left(\frac{\mathcal{M}}{\alpha_{\mathcal{M}}}\right) \leq \frac{\lambda_{max}}{\alpha_{\mathcal{M}}}\le 1\,.
\end{align}
Since by definition $\|\mathcal{M}\|=\lambda_{max}$ and $\|\mathcal{M}^{-1}\|=1/\lambda_{min}$, the last inequality gives the condition $\|\mathcal{M}\| \leq \alpha_{\mathcal{M}}$, which we note to be sufficient for the inequality $\|\mathcal{M}\| \leq \alpha_{\mathcal{M}}+\delta_{\mathcal{M}}$ demanded by the definition of the block access to $\mathcal{M}$. Then the first equality in combination with the definition $\|\mathcal{M}\|\|\mathcal{M}^{-1}\|=\kappa$ gives
\begin{align}
    \chi < \frac{\|\mathcal{M}\|}{\alpha_{\mathcal{M}}\kappa}.
\end{align}
When $\|\mathcal{M}\| \leq \alpha_{\mathcal{M}}$, $\chi< 1/\kappa$ suffices. (In our application to PDEs and ODEs,  $\|\mathcal{M}\|<2$).  One can choose for instance $\chi=\|\mathcal{M}\|/(\theta \alpha_{\mathcal{M}}\kappa)$ for any constant $\theta>1$.  
Then interpreting $x$ in Lemma~\ref{lem:poly} to be the spectrum of $\mathcal{M}/\alpha_{\mathcal{M}}$ one can replace the inequality
\begin{align}
    |\chi/4x-\mathcal{P}(x)/2|\leq \zeta/2
\end{align}
by 
\begin{align} \label{eq:mzeta}
    \|\mathcal{M}^{-1}-(4\theta \kappa/\|\mathcal{M}\|) \mathcal{P}(\mathcal{M}/\alpha_{\mathcal{M}})/2\| \leq 2\theta \kappa \zeta/\|\mathcal{M}\|
\end{align}
By defining the block encoding $U_{\mathcal{M}^{-1}}=U_{\mathcal{P}(\mathcal{M}/\alpha_{\mathcal{M}})/2}$ then the definition of $\delta$ requires 
\begin{align} \label{eq:deltarequire}
    \|\mathcal{M}^{-1}-(4 \theta\kappa/\|\mathcal{M}\|) \langle 0^{n_{\mathcal{M}}+2}|U_{\mathcal{P}(\mathcal{M}/\alpha_{\mathcal{M}})/2}|0^{n_{\mathcal{M}}+2}\rangle\|\leq \delta.
\end{align}
Lemmas~\ref{lem:newinvm} and \ref{lem:poly} imply
\begin{align} \label{eq:polyineq}
    \|\langle 0^{n_{\mathcal{M}}+2}|U_{\mathcal{P}(\mathcal{M}/\alpha_{\mathcal{M}})/2}|0^{n_{\mathcal{M}}+2}\rangle-\mathcal{P}(\mathcal{M}/\alpha_{\mathcal{M}})/2\| \leq 4D(\mathcal{P})\sqrt{\delta_{\mathcal{M}}/\alpha_{\mathcal{M}}}+\sigma
\end{align}
where it is sufficient to choose $D(\mathcal{P})=v\log(1/\zeta)/\chi=v \theta \alpha_{\mathcal{M}}\kappa\log(1/\zeta)/\|\mathcal{M}\|$ for some constant $v>0$. Putting together Eqs.~\eqref{eq:mzeta}, ~\eqref{eq:deltarequire} and ~\eqref{eq:polyineq} gives
\begin{align}
     \|\mathcal{M}^{-1}-(4\theta \kappa/\|\mathcal{M}\|) \langle 0^{n_{\mathcal{M}}+2}|U_{\mathcal{P}(\mathcal{M}/\alpha_{\mathcal{M}})/2}|0^{n_{\mathcal{M}}+2}\rangle\| \leq \frac{1}{\|\mathcal{M}\|}(2\theta \kappa \zeta+4\theta \kappa \sigma+16\theta^2 \kappa^2 v\log(1/\zeta) \sqrt{\delta_{\mathcal{M}}\alpha_{\mathcal{M}}})\leq \delta.
\end{align}
For example, in the limit $\delta_{\mathcal{M}}=0$, one can choose the parameters $\zeta=\delta\|\mathcal{M}\|/(8\kappa)$ and $\sigma=\delta \|\mathcal{M}\|/(16\kappa)$. However, we are interested in the case $\delta_{\mathcal{M}}>0$. Then it is sufficient to choose $\zeta=\delta\|\mathcal{M}\|(1-\theta')/(4\theta\kappa)$,  $\sigma=\delta\|\mathcal{M}\|(1-\theta')/(8\theta \kappa)$ and $\sqrt{\delta_{\mathcal{M}}\alpha_{\mathcal{M}}} \leq \delta \|\mathcal{M}\|\theta'/(16 \theta^2 v\kappa^2 \log (8\kappa/(\|\mathcal{M}\|\delta)))$ for any constant $0< \theta'<1$. The latter implies  $\delta_{\mathcal{M}}<\delta^2 \|\mathcal{M}\|^2(\theta')^2/(16^2\alpha_{\mathcal{M}}\kappa^4\log^2(\kappa/\delta))$. From Lemma~\ref{lem:sparsetoblock} and ignoring constant factors $\theta, \theta', v,  \|\mathcal{M}\|=O(1)$ this gives rise to an additional $O(\log^{2.5}(\alpha_{\mathcal{M}}^2\kappa^4 \log^2(\kappa/\delta)/\delta^2))$ 2-qubit gates to create block access to $\mathcal{M}$ from sparse access. Later from Eq.~\eqref{eq:deltaepsilon'} we see that $\delta \sim \epsilon'/\kappa^2$ can be chosen where $\epsilon'$ is the final error in $\Upsilon$, so this gives the additional gate cost $O(\log^{2.5}(s^2\|\mathcal{M}\|_{max}^2\kappa^8 \log^2(\kappa^3/\epsilon')/(\epsilon')^2)<O(\log^{2.5}(s\|\mathcal{M}\|_{max}\kappa^4/\epsilon'))$.   \\

Since one can choose $\mathcal{D}(\mathcal{P})=v\theta \kappa\alpha_{\mathcal{M}}\log(4\theta \kappa/(\|\mathcal{M}\|\delta(1-\theta')))/\|\mathcal{M}\|$ and then ignoring all constants except $\|\mathcal{M}\|$ and $\theta$, the rest of the proof follows using Lemma~\ref{lem:newinvm}. 
\end{proof} 

From the above lemma one sees one can set $\alpha_{\mathcal{M}^{-1}}=4\theta \kappa/\|\mathcal{M}\|$ for any $\theta>1$ and $n_{\mathcal{M}^{-1}}=n_{\mathcal{M}}+2$. Let the total error in estimating $\Upsilon$ be $\epsilon'$. There are two sources of error: one in the amplitude estimation algorithm that outputs $\tilde{u}$, which approximates $u$ with error $\epsilon_u$ and the other error is in the block-encoding of $\mathcal{M}^{-1}$, which has error $\delta$. This means 
\begin{align}
    |u-\tilde{u}|\leq \epsilon_u
\end{align}
where $\tilde{u}$ is the estimate of $u$ from the amplitude estimation algorithm and 
\begin{align}
    \|\alpha_{\mathcal{M}^{-1}}\mathcal{L}-\mathcal{M}^{-1}\|\leq \delta.
\end{align}
Then the total error 
\begin{align}
 &|\alpha^2_{\mathcal{M}^{-1}}u-\Upsilon|=|\alpha^2_{\mathcal{M}^{-1}}\langle  \psi_{0}|\mathcal{L}^{\dagger}\mathcal{G}\mathcal{L}|\psi_{0}\rangle-\langle \psi_{0}|(\mathcal{M}^{-1})^{\dagger}\mathcal{G}\mathcal{M}^{-1}|\psi_{0}\rangle|\nonumber \\
 &=|\alpha^2_{\mathcal{M}^{-1}}\|\mathcal{F}\mathcal{L}|\psi_0\rangle\|^2-\|\mathcal{F}\mathcal{M}^{-1}|\psi_0\rangle\|^2|\nonumber \\
 & \leq \|\alpha^2_{\mathcal{M}^{-1}}\mathcal{F}\mathcal{L}|\psi_0\rangle-\mathcal{F}\mathcal{M}^{-1}|\psi_0\rangle\|^2 \nonumber \\
 & \leq \|\mathcal{F}\|^2 \|\alpha^2_{\mathcal{M}^{-1}}\mathcal{L}|\psi_0\rangle-\mathcal{M}^{-1}|\psi_0\rangle\|^2 \nonumber \\
 & \leq \|\mathcal{F}\|^2 \delta^2= \delta^2
\end{align}
where in the second line we used $\mathcal{G}=\mathcal{F}^{\dagger}\mathcal{F}$ since $\mathcal{G}$ is positive semi-definite. In the last line we used $\|\mathcal{F}\|=\|\mathcal{G}\|=1$. This means that the total error in $\Upsilon$ can be written as  
\begin{align} \label{eq:firstuest}
    |\alpha^2_{\mathcal{M}^{-1}}\tilde{u}-\Upsilon| \leq |\alpha^2_{\mathcal{M}^{-1}}\tilde{u}-\alpha^2_{\mathcal{M}^{-1}}u|+|\alpha^2_{\mathcal{M}^{-1}}u-\Upsilon|\leq \alpha^2_{\mathcal{M}^{-1}}\epsilon_u+\delta^2 \leq \epsilon' 
\end{align}
 Then a choice of $\alpha^2_{\mathcal{M}^{-1}}\epsilon_u \sim \delta^2 \sim \epsilon'/2$ is sufficient. Since $\alpha_{\mathcal{M}^{-1}}=4 \theta \kappa/\|\mathcal{M}\|$ one can choose
\begin{align} \label{eq:deltaepsilon'}
    & \delta=\sqrt{\epsilon'/2} \nonumber \\
    & \epsilon_u=\epsilon'\|\mathcal{M}\|/(32 \theta^2 \kappa^2).
\end{align}
For instance, it is possible to set $\theta^2=33/32$ and let $\epsilon_u=\epsilon'\|\mathcal{M}\|/(33 \kappa^2)$. 
Using Lemma~\ref{lem:finiteerrorm}, inserting Eq.~\eqref{eq:deltaepsilon'} and ignoring all constants except for $\|\mathcal{M}\|=O(1)$ for convenience, it is sufficient for the amplitude estimation algorithm to make $\mathcal{O}(1/\epsilon_u)\sim \mathcal{O}(\kappa^2/(\|\mathcal{M}\|\epsilon'))$ queries to $U_{\mathcal{G}'}$ and  $U_{initial}$. To query $U_{\mathcal{M}}$, one must multiply the amplitude estimation cost with the query cost for $\mathcal{M} \rightarrow \mathcal{M}^{-1}$, hence a total $\mathcal{O}(\alpha_{\mathcal{M}} \kappa \log(\kappa/\delta)/\epsilon_u)\sim \mathcal{O}(\alpha_{\mathcal{M}}\kappa^3 \log(\kappa^2/(\epsilon')/(\|\mathcal{M}\|\epsilon'))$ queries to $U_{\mathcal{M}}$, and $\mathcal{O}(r/\epsilon_u)\sim \mathcal{O}(\kappa^2(m+n_{\mathcal{G}}+1+n_{\mathcal{M}^{-1}})/(\|\mathcal{M}\|\epsilon'))$ additional $2$-qubit gates. \\

Since we are given the block encoding to $\mathcal{G}$ rather than $\mathcal{G}'$, Lemma~\ref{lem:g0} requires only one query to $U_{\mathcal{G}}$ and $\mathcal{O}(n^2_{\mathcal{M}})\sim \mathcal{O}(1)$ additional $2$-qubit gates, since $n_{\mathcal{M}}=2$ from Lemma~\ref{lem:sparsetoblock}. \\

From Lemmas~\ref{lem:newinvm} and ~\ref{lem:finiteerrorm}, to construct block access to $\mathcal{M}^{-1}$ from block access to $\mathcal{M}$ requires an additional $\mathcal{O}(n_{\mathcal{M}}\kappa\alpha_{\mathcal{M}}\log (\kappa^3/\epsilon'))$ $2$-qubit gates. Although from the proof of Lemma~\ref{lem:finiteerrorm} one sees that $O(\log^{2.5}(s\|\mathcal{M}\|_{max}\kappa^4/\epsilon'))$ additional 2-qubit gates are required to construct block access from sparse access to $\mathcal{M}$, these are all logarithmic factors, which we will ignore in the final expression.\\

Putting all these results together with $\alpha_{\mathcal{M}}=s\|\mathcal{M}\|_{max}$, $n_{\mathcal{M}^{-1}}=n_{\mathcal{M}}+2$ and ignoring all constants except $\|\mathcal{M}\|$, one finds that to approximate $\Upsilon$ to precision $\epsilon'$, one needs $\mathcal{O}(\kappa^2/(\|\mathcal{M}\|\epsilon'))$ queries to $U_{\mathcal{G}}$ and $U_{initial}$,  $\mathcal{O}(s\|\mathcal{M}\|_{max}\kappa^3\log(\kappa^2/\epsilon')/(\|\mathcal{M}\|\epsilon'))$ queries to sparse oracles for $\mathcal{M}$,   and $\mathcal{O}((\kappa^2/(\|\mathcal{M}\|\epsilon'))(m+n_{\mathcal{G}}+1+n_{\mathcal{M}}+2+s\|\mathcal{M}\|_{max}n_{\mathcal{M}}\kappa\log(\kappa^2/(\|\mathcal{M}\|\epsilon'))))$ additional $2$-qubit gates. Since $n_{\mathcal{M}}=2$, $n_{\mathcal{G}}=2m+1$ for our scenario, this leads to $\mathcal{O}((\kappa^2/(\|\mathcal{M}\|\epsilon'))(3m+6+2s\|\mathcal{M}\|_{max}\kappa\log(\kappa^2/(\|\mathcal{M}\|\epsilon'))))\sim \mathcal{O}((\kappa^2/(\|\mathcal{M}\|\epsilon'))(m+s\|\mathcal{M}\|_{max}\kappa\log(\kappa^2/(\|\mathcal{M}\|\epsilon'))))$ where we ignore the constant factor terms.

\section{Proof of Theorem~
\ref{thm:qquery}} \label{app:theoremhj}

To estimate $\langle G(t_n,x)\rangle$ to precision $\epsilon$, Lemma~\ref{lem:3errors} requires the estimation of $\sqrt{\Upsilon}$ to precision $\epsilon_G \sim \epsilon/(n_Gn_{\psi_0})$. Using the relation $\Delta (x^2)=2x \Delta x$, where $\Delta x$ is the error in $x$ and $\Delta (x^2)$ is the error in $x^2$, taking $x=\sqrt{\Upsilon}$, one finds that the additive
error $\epsilon'$ in $\Upsilon$ has size $\epsilon'\sim \sqrt{\Upsilon}\epsilon_G$. To identify how $\sqrt{\Upsilon}$ would scale with $N$, from Eq.~\eqref{eq:gequation} one observes that $O(1)=G^n_{\omega}(\vect{j})$ implies $\sqrt{\Upsilon}\sim 1/(n_Gn_{\psi_0})$ which gives $\epsilon'\sim \epsilon_G/(n_G n_{\psi_0}) \sim \epsilon/(n_G n_{\psi_0})^2\sim \epsilon/ n_{\psi_0}^2$ when we suppress the $n_G=O(1)$ factor.  \\

Then following the quantum algorithm outlined in Lemma~\ref{lem:bslep2} with $2^m=N_t N^{2d}$, we include $s=O(d)$ and $\kappa\leq O(N_t)$ from Lemma~\ref{lem:skappa}. Since $N_t=T/\Delta t$, where from the stability condition we have $\Delta t \sim 1/(Nd)$, this gives $\kappa \leq O(T Nd)$. Inserting this and $\epsilon'\sim \epsilon/n_{\psi_0}^2$ into Lemma~\ref{lem:bslep2}, one directly finds that one needs to make $\mathcal{O}(T^2 d^2N^2n^2_{\psi_0}/\epsilon)$ queries to $L(n, \vect{j})$ and $U_{initial}$ and $\mathcal{O}((n_{\psi_0}^2/\epsilon)d^4N^3T^3 \log (d^2N^2T^2 n^2_{\psi_0}/\epsilon))$ queries to sparse oracles for $\mathcal{M}$. One also needs an additional $\mathcal{O}\left((n_{\psi_0}^2d^2N^2T^2/\epsilon)(\log(TNd)+d\log(N)+d^2N T \log(T^2d^2N^2n_{\psi_0}^2/\epsilon)\right)$ two-qubit gates. The largest of these terms is $\mathcal{O}\left((n^2_{\psi_0}d^4N^3T^3/\epsilon)\log(d^2N^2T^2n^2_{\psi_0}/\epsilon)\right)$. \\

To express this entirely in terms of the natural parameters $\epsilon$, $T$, $M$ and $d$ only, we note that in the quantum algorithm we require $\epsilon\sim \epsilon_{\text{CL}}$ from Lemma~\ref{lem:3errors} and here $N=N_{\text{CL}}$, so from Lemma~\ref{lem:liouville} we have $N\sim d/\epsilon^3$. This easily gives us $\mathcal{Q}= \mathcal{O}(n^2_{\psi_0}T^3d^7(1/\epsilon)^{10}\log(n^2_{\psi_0}T^2d^4(1/\epsilon)^{7})$.\\

The above analysis is only for computing $\Upsilon$ to the required accuracy to estimate $\langle G(t_n, x)\rangle$. However, to obtain our observable it is insufficient to compute $\Upsilon$ alone since the normalisation constants need to be considered to derive $\langle G^n_{\omega}(\vect{j})\rangle$ from Eq.~\eqref{eq:gequation}. The normalisation constant $n_{\mathcal{G}}$ can be computed easily on a classical device since the function $G(p)$ is known and is generally a simple polynomial. Since the initial data is also known, it is also sufficient to use a classical device to first compute $n_{\psi_0}$. Alternatively, if one doesn't wish to compute $n_{\psi_0}$, one can instead estimate the observable in Definition~\ref{def:levelsetobservable}, where $G_O(t_n, x) \equiv \langle G(t_n, x)\rangle /\langle G'(t_n, x) \rangle$ with $G'(p)\equiv 1$, where $n_{G'}=1$ exactly. This means we estimate
\begin{align}
    G_O(t_n,x)\equiv \frac{\langle G^n_{\omega}(\vect{j})\rangle}{\mathbf{1}^n_{\omega}(\vect{j})}=n_G\frac{|\sqrt{\Upsilon}|}{|\sqrt{\Upsilon_+}|}
\end{align}
where $\Upsilon_+ \equiv \sqrt{\langle \psi_{0}|(\mathcal{M}^{-1})^{\dagger}\mathcal{H}\mathcal{M}^{-1}|\psi_{0}\rangle}$ and $\mathcal{H}=|H\rangle\langle H|$ and $|H\rangle$ is the state with equal superposition across the $\vect{l}$ basis and can be easily created by applying Hadamard gates on $|0\rangle$. Since $G_O(t_n,x)=O(1)$ just like $\langle G(t_n,x)\rangle=O(1)$, the error  $|G_O(t_n,x)-\tilde{G}_O(t_n,x)|$ is also $O(\epsilon)$.

\section{Basic summary of quantum algorithm to estimate $\Upsilon$} \label{app:algorithmsteps}

Here we present a condensed version of the main steps of the quantum algorithm in Theorems~\ref{thm:qquery} and \ref{thm:qqueryode} (where we use $\mathcal{M}=\mathcal{M}_{\text{ODE}}$). For query and gate complexity at each step refer to the lemmas referenced:
\begin{enumerate}
    \item (Step 0) Inputs: sparse access $(s, \|\mathcal{M}\|_{max}, O_{\mathcal{M}}, O_F)$ to $m$-qubit operator $\mathcal{M}$, access to $U_{initial}$ and access to unitary $L(n, \vect{j})$;
    
    \item (Step 1) Using sparse access to $\mathcal{M}$ to construct block access $(s\|\mathcal{M}\|_{max}, 2, \delta_{\mathcal{M}}, U_{\mathcal{M}})$ to $\mathcal{M}$, using Lemma~\ref{lem:sparsetoblock}, so $\alpha_{\mathcal{M}}=s\|\mathcal{M}\|_{max}$ and $n_{\mathcal{M}}=2$;
    
    \item (Step 2) Using block access to $(s\|\mathcal{M}\|_{max}, 2, \delta_{\mathcal{M}}, U_{\mathcal{M}})$ to $\mathcal{M}$ to construct block access $(4\theta\kappa/\|\mathcal{M}\|, 4, \delta, U_{\mathcal{M}^{-1}})$ to $\mathcal{M}^{-1}$ where $\alpha_{\mathcal{M}^{-1}}=4\theta\kappa/\|\mathcal{M}\|$, $n_{\mathcal{M}^{-1}}=n_{\mathcal{M}}+2=4$, from Lemma~\ref{lem:finiteerrorm};
    
    \item (Step 3) Starting from access to unitary $L(n,\vect{j})$ can create block access $(1,n_{\mathcal{G}}, 0, U_{\mathcal{G}})$ to $\mathcal{G}$ with construction given in Lemma~\ref{lem:gconstruction} with choice $n_{\mathcal{G}}=2m+1$;
    
    \item (Step 4) Using block access $(1,n_{\mathcal{G}}, 0, U_{\mathcal{G}})$ to $\mathcal{G}$ to construct block access $(1, 2m+2, 0, U_{\mathcal{G}'})$ to $\mathcal{G}'=|0^4\rangle \langle 0^4|\otimes \mathcal{G}$, where $\alpha_{\mathcal{G}'}=1$, $n_{\mathcal{G}'}=n_{\mathcal{G}}+1=2m+2$, from Lemma~\ref{lem:g0};
    
    \item (Step 5) Given $U_{\mathcal{M}^{-1}}$ and $U_{\mathcal{G}'}$, and also assume access to their adjoints and their respective controlled- unitaries. From this to construct controlled-$S$ and the corresponding quantum phase estimation circuit, where $S$ is in Eqs.~\eqref{eq:uvdef} and~\eqref{eq:sdef}. Extract $u$;
    
    \item Output: Multiply $u$ by $\alpha^2_{\mathcal{M}^{-1}}$ where $\alpha_{\mathcal{M}^{-1}}=4\kappa \theta/\|\mathcal{M}\|$. This estimates $\Upsilon$ to precision $\epsilon'$.
\end{enumerate}
\begin{figure}[t] \label{fig:Algorithmflow}
\includegraphics[width=15cm]{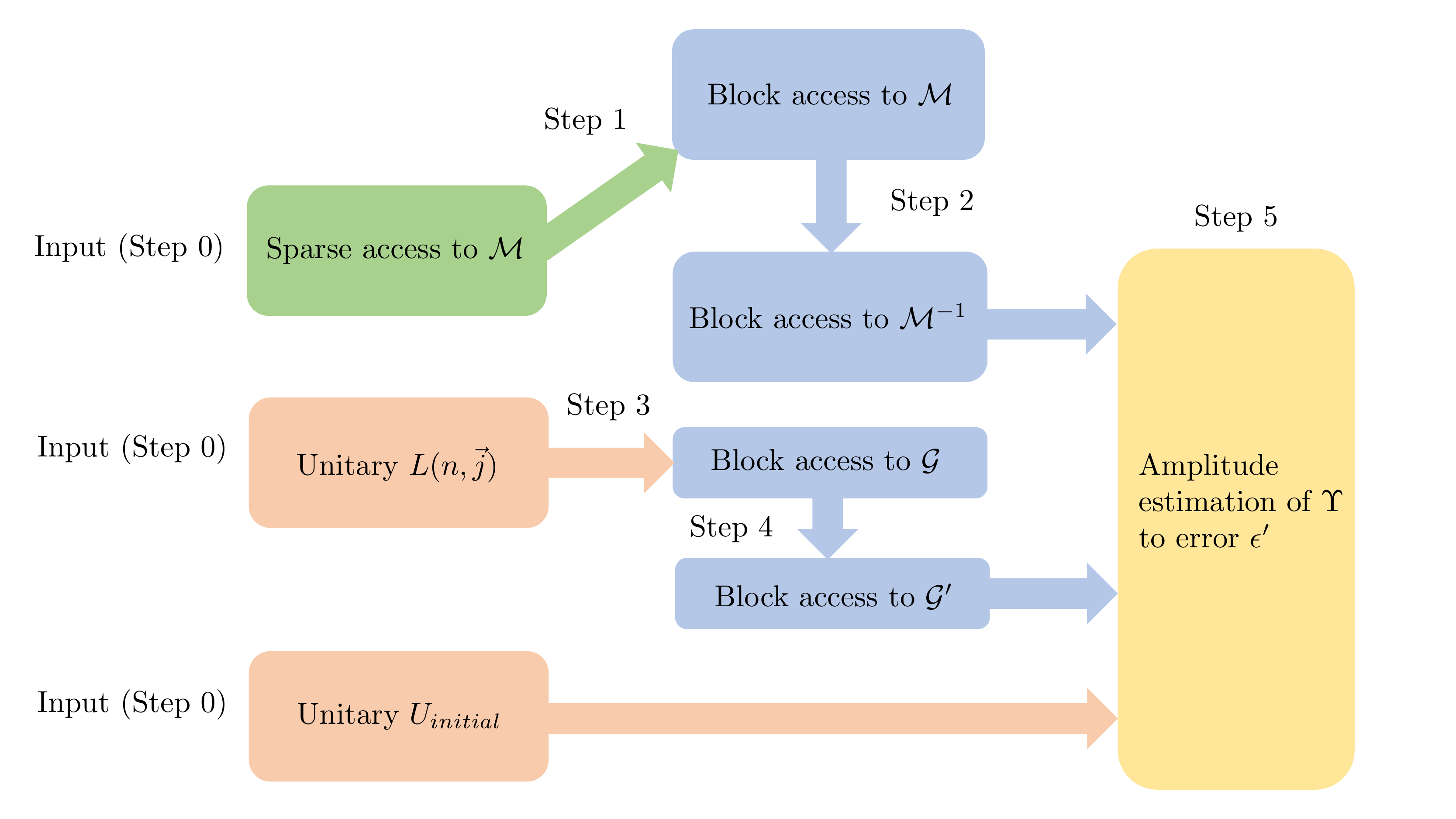}
\centering
\caption{Summary of algorithm to estimate $\Upsilon$ to error $\epsilon'$}
\end{figure}

\section{Proof of Lemma~\ref{weaksoln-proof}} \label{app:odejustification}

  Let $\chi(q) $ be a smooth test function, and consider
 \[
 \int_{\mathbb{R}^d} \chi(q) \Phi(t, q) dq
 = \frac{1}{M}\sum_{k=1}^M \chi(X^{[k]}(t)).
 \]
 Taking the time derivative of this equation and using (\ref{eq:upde}), one deduces
 \begin{eqnarray}\nonumber
&&\quad  \int_{\mathbb{R}^d} \chi(q) \partial_t \Phi(t, q) dq
 = \frac{1}{M}\sum_{k=1}^M \partial_t \chi(X^{[k]})
 = \frac{1}{M}\sum_{k=1}^M  \nabla_q  \chi(X^{[k]}) \cdot \partial_t
 X^{[k]}(t) 
 = \frac{1}{M}\sum_{k=1}^M  \nabla_q  \chi(X^{[k]}) \cdot F(X^{[k]})  \\ 
&&=  \frac{1}{M} \int_{\mathbb{R}^d} \nabla_q  \chi(q) \cdot F(q)
\, \sum_{k=1}^M  \delta(q-X^{[k]}(t))\, dq
=  \int_{\mathbb{R}^d}\nabla_q  \chi(q) \cdot F(q) \Phi(t,q)\, dq = -\int_{\mathbb{R}^d} \chi(q) \nabla_q \cdot [F(q) \Phi(t,q)]\, dq. 
 \end{eqnarray}
 The uniqueness of the solution is also classical, see \cite{Rav85}. We omit the details.

\section{Proof of Lemma~\ref{lem:AA}}\label{app:proofobservablea}

Consider the problem 
 \begin{align}  \label{eq:ode-1}
 &   \frac{\partial X(t)}{\partial t}=F( X), \qquad X\in \mathbb{R}^D\,,\\
 & X(s)=x, \quad s\in [0, T].
\end{align}
Let the solution to (\ref{eq:ode-1}) be $X(t; x,s)$. Define the
Jacobian determination of the map from $x$ to $X$ be
\[ 
J(t; x,s)= {\text{det}} \left( \frac{\partial X_i}{\partial x_j}(t; x,s)\right)
\]
then classical result (see for example \cite{Rav85})  shows that
\[
  J(t; x, s)>0,\qquad
  J(t; x,s) = \exp \int_s^t \nabla \cdot F(X(\sigma; x,s)) d\sigma
  \]
 and, by the method of characteristics,
 \[
 \Phi(t, p)=\Phi(0, X(0; p,t))J(0; p,t)
 \]
 Now,
 \begin{eqnarray}\nonumber
 A_p(t) &&=\int_{\mathbb{R}^D} A(p)\Phi(t,p)\, dp
 =\int_{\mathbb{R}^D} A(p)\Phi(0, X(0; p,t))J(0; p,t)\, dp\\
\nonumber && = \int_{\mathbb{R}^D} A(p)\Phi(0, X(0; p,t))J(0; p,t)\, dp
  = \frac{1}{M} \sum_{k=1}^M  \int_{\mathbb{R}^D} A(p)\delta(X(0; p,t)-X_0^{[k]})J(0; p,t)\, dp\\
\nonumber &&= \frac{1}{M} \sum_{k=1}^M  \int_{\mathbb{R}^D} A(p)J(0; p,t)^{-1}\delta(p-X(t;X_0^{[k]},0))J(0; p,t)\, dp \\
\nonumber &&= \frac{1}{M} \sum_{k=1}^M  \int_{\mathbb{R}^D} A(p)\delta(p-X(t;X_0^{[k]},0))\, dp 
  = \frac{1}{M} \sum_{k=1}^M   A(X(t;X_0^{[k]},0))=A_0(t)\,.
\end{eqnarray}
Note here, by the definition in (\ref{eq:ode-1}), $X(t;X_0^{[k]},0)$ is the solution to (\ref{eq:upde}), hence the last equality holds.

\section{Discretised System of ODEs} \label{app:discretisationode}

We now discretise the linear PDE in Eq.~\eqref{eq:pdephi} by  finite difference schemes. As an example we use the upwind scheme, which takes the following form:
\begin{align}
    & \frac{\partial \Phi(t_n, q)}{\partial t} \rightarrow \frac{\Phi_{n+1, \vect{j}}-\Phi_{n, \vect{j}}}{\Delta t} \nonumber \\
   & \frac{\partial}{\partial q_i} (F_i(q) \Phi(t_n, q)) \rightarrow \frac{1}{h}\left[(F^-_i(q_{i+1/2})T^+_i\Phi_n)_{\vect{j}}-(F^+_i(q_{i-1/2})T^-_i\Phi_{n, \vect{j}}
    +[(F^+_i(q_{j+1/2})-F^-_i(q_{j-1/2}))\Phi_n]_{\vect{j}}\right]\nonumber
\end{align}
Here $\vect{j} \equiv (j_1,...,j_D)$, $j_i=1,...,N$ for $i=1,...,D$, $n=1,...,N_t$, $t_n=n \Delta t$ and $q_i= j_i h$, where $q_i$ are the $i^{th}$ components of the vector $q$. In addition, $T^{\pm}_i\Phi_{n, \vect{j}}=\Phi_{n, j_1,...,j_{i \pm 1},...,j_d}$, $(F_i(q_{i+1/2}))_{\vect{j}}=\frac{1}{2}[(F_i(q_{j_1,...,j_{i+1},...,j_d})+F_i(q_{j_1,...,j_{i},...,j_d})]$, and $F^+(q)=\max(F(q), 0)$ and $F^-(q)=\min(F(q), 0)$. \\

Define $\lambda = \Delta t/h$. We require 
\begin{equation}\label{CFL}
     D\lambda = D\frac{\Delta t}{h} \le 1
\end{equation}
for numerical stability. 

Then the discretised version of Eq.~\eqref{eq:pdephi} can be rewritten as 
\begin{align}
    \Phi_{n+1, \vect{j}}+\lambda \sum_{i=1}^D \left(F^-_i(q_{i+1/2})_{\vect{j}}T^+_i\Phi_{n, \vect{j}}- F^+_i(q_{i-1/2})_{\vect{j}}T^-_i\Phi_{n, \vect{j}}\right)
    +[1+\lambda \sum_{i=1}^D(F^+_i(q_{j+1/2})-F^-_i(q_{j-1/2}))_{\vect{j}}]\Phi_{n, \vect{j}}=0
\end{align}
with the initial condition (for $n=0$)
\begin{align} \label{eq:phiinitial}
    \Phi_{0, \vect{j}}= \frac{1}{M}\sum_{k=1}^M  \Pi_{i=1}^D\delta_\omega(j_ih-(X_0^{[k]})_i).
\end{align}
Then the discretised PDEs can be written as a matrix equation
\begin{align} \label{eq:Phimatrix}
    \mathcal{K}\begin{pmatrix}
    \Phi_{1, \vect{j}} \\
    \Phi_{2, \vect{j}} \\
    \vdots \\
    \Phi_{N_t-1, \vect{j}} \\
    \Phi_{N_t, \vect{j}}
    \end{pmatrix}=\begin{pmatrix}
    \Phi_{0, \vect{j}} \\
    0 \\
    \vdots \\
    0 \\
    0
    \end{pmatrix}
\end{align}
and $\mathcal{K}_{\text{ODE}}$ is a $(N_t N^D) \times (N_t N^D)$ Toeplitz matrix of the form
\begin{align} \label{eq:kmatrix}
    \mathcal{K}_{\text{ODE}}=\begin{pmatrix}
    I & 0 & 0 & \hdots & 0 & 0  \\
     K & I & 0 & \hdots & 0 & 0  \\
    \vdots & \vdots & \vdots & \vdots & \vdots & \vdots  \\
    0 & 0 & 0 & \hdots &  K & I  
    \end{pmatrix}
\end{align}
where each $I$ is the $(N^D)\times (N^D)$ identity matrix and $K$ is the $(N^D)\times (N^D)$ matrix 
\begin{align}
    K=-I+ \lambda \sum_{i=1}^D \left( F^-_i(q_{i+1/2})_{\vect{j}}T^+_i- F^+_i(q_{i-1/2})_{\vect{j}}T^-_i
    + (F^+_i(q_{j+1/2})-F^-_i(q_{j-1/2}))_{\vect{j}}I\right).
\end{align}
 We can then solve for $\Phi_{n, \vect{j}}$ by matrix inversion
\begin{align} \label{eq:matrixinv}
    \begin{pmatrix}
    \Phi_{1, \vect{j}} \\
    \Phi_{2,\vect{j}} \\
    \vdots \\
    \Phi_{N_t-1, \vect{j}} \\
    \Phi_{N_t, \vect{j}}
    \end{pmatrix}=\mathcal{K}^{-1}_{\text{ODE}} \begin{pmatrix}
    \Phi_{0,\vect{j}} \\
    0 \\
    \vdots \\
    0 \\
    0
    \end{pmatrix}
\end{align}
Then given the discretised solutions $\Phi_{n, \vect{j}}$, we can estimate the ensemble average by  
\begin{align} \label{eq:discreteensemble}
    \langle A_\omega^n \rangle = h^D \sum_{\vect{|j|=1}}^{N} A_{\vect{j}} \Phi_{n, \vect{j}} 
\end{align}
where $A_{\vect{j}}$ is the discretisation of $A(q)$, i.e. $G(q) \rightarrow A(\vect{j} h) \equiv A_{\vect{j}}$. From Section~\ref{sec:classicalerrorbounds}, we see that we also use $\langle A^{\omega}_n \rangle$ to estimate $A_O(t_n)$ with error 
\begin{align} \label{eq:acerror}
    \epsilon_c \equiv |A_o(t_n)-\langle A^{\omega}_n \rangle| \leq C(Dh)^{1/3})
\end{align}
for $\omega=(Dh)^{1/3}$.

{\bf Remark:}  The stability condition in (\ref{CFL}) requires $\Delta t$ to be proportional to $1/d$. This is the issue for
an explicit scheme. This means $N_t$, and consequently the size of matrix $\mathcal{K}$, will be $O(d)$ times larger than a scheme in which $\Delta t$ is independent of $d$.
An implicit scheme for \eqref{eq:pdephi}--which is usually not needed in classical algorithms for a transport equation-- will remove such an dependence. Since the HHL algorithm is based on inverting the matrix $\mathcal{K}$, which is basically implicit anyway,
one can starts with an implicit scheme for equation \eqref{eq:pdephi} \cite{JLY22}. Since the HHL algorithm just  depends on logarithmic of
the matrix size, an $d$ times larger matrix size will not increase significantly the computational cost, hence we do not
explore this issue and leave it for a future work.
\\

Just like in previous analyses, since our aim is to approximate the observable $A_O(t_n)$ by using quantum algorithms that require a matrix inversion subroutine $\mathcal{M}^{-1}\vec y$ where $\mathcal{M}$ is Hermitian, we can define a new Hermitian matrix 
\begin{align}
    \mathcal{M}_{\text{ODE}}=\begin{pmatrix} 0 & \mathcal{K}_{\text{ODE}} \\
    \mathcal{K}_{\text{ODE}}^{\dagger} & 0 
    \end{pmatrix}
\end{align}
which has the same sparsity and condition number as $\mathcal{K}_{\text{ODE}}$. \\

The condition number of $\mathcal{M}_{\text{ODE}}$ is $\kappa \lesssim D NT$ and sparsity is $s \sim D$, from
Appendix~\ref{app:skappaproof} with the replacement $d \rightarrow D/2$.

\bibliography{Nonlinear}

\end{document}